\documentclass[10pt, twocolumn]{IEEEtran}

\usepackage{calc}
\usepackage{cite} 
\usepackage{enumerate}
\usepackage{booktabs}
\usepackage[colorlinks,citecolor=black, linkcolor=black]{hyperref}
\usepackage{url}
\usepackage{cite}
\usepackage{amsmath} 
\usepackage[dvips]{graphicx}  
\usepackage{booktabs}
\usepackage{multirow}
\usepackage{multicol} 
\usepackage{algorithm}
\usepackage{algorithmic}
\usepackage{amsmath}
\usepackage{amsthm}
\usepackage{amsfonts} 
\usepackage{caption}
\usepackage{color,cite} 
\usepackage{xcolor}

\usepackage{subfig}

\def\z{\mathbf{z}}
\def\w{\mathbf{w}}
\def\x{\mathbf{x}}
\def\E{\mathbb{E}}
\def\y{\mathbf{y}}
\def\v{\mathbf{v}}
\def\u{\mathbf{u}}

\def\g{\mathbf{g}}
\def\h{\mathbf{h}}
\def\a{\mathbf{a}}
\def\e{\mathbf{e}}
\def\T{\mathcal{T}}

\def\R{\mathcal{R}}
\def\P{\mathbb{P}}
\def\TT{{\mathcal T}} 

\def\C{{\mathbb C}}

\DeclareMathOperator*{\supp}{supp}

\newtheoremstyle{IEEEThmStyle}{}{}{}{\parindent}{\itshape}{:}{ }{}
\theoremstyle{IEEEThmStyle}
\newtheorem{thm}{Theorem}
\newtheorem{lem}{Lemma}
\newtheorem{prop}{Proposition}

\newtheorem{claim}{Claim}

\newcommand{\black}[1]{\textcolor{black}{#1}}

\definecolor{darkspringgreen}{rgb}{0.09, 0.45, 0.27}

\begin{document}
\title{Subspace Phase Retrieval}

\author{Mengchu~Xu,~
Dekuan~Dong~
and~Jian Wang,~\IEEEmembership{Member,~IEEE}    	
\thanks{
The authors are with the School of Data Science, Fudan University, Shanghai 200433, China. E-mail: \{mcxu21, dkdong21\}@m.fudan.edu.cn;  jian\_wang@fudan.edu.cn. Corresponding author: Jian Wang.} 
}

\maketitle

\maketitle

\begin{abstract}
In recent years, phase retrieval has received much attention in statistics, applied mathematics and optical engineering. In this paper, we propose an efficient algorithm, termed Subspace Phase Retrieval (SPR), which can accurately recover an $n$-dimensional $k$-sparse complex-valued signal $\x$ given its $\Omega(k^2\log n)$ magnitude-only Gaussian samples if the minimum nonzero entry of $\x$ satisfies $|x_{\min}| = \Omega(\|\x\|/\sqrt{k})$. Furthermore, if the energy sum of the most significant $\sqrt{k}$ elements in $\x$ is comparable to $\|\x\|^2$, the SPR algorithm can exactly recover $\x$ with $\Omega(k \log n)$ magnitude-only samples, which attains the information-theoretic sampling complexity for sparse phase retrieval. Numerical Experiments demonstrate that the proposed algorithm achieves the state-of-the-art reconstruction performance compared to existing ones.

\end{abstract}
 
\begin{IEEEkeywords}
Phase retrieval, information-theoretic bound, nonconvex optimization, sparsity, support index.
\end{IEEEkeywords}

\IEEEpeerreviewmaketitle

\section{Introduction}\label{sec:intro}

\IEEEPARstart{P}{hase} retrieval arises in dozens of optical imaging applications, such as X-ray imaging, crystallography, coherent diffraction imaging, and atmospheric imaging~\cite{XrayLit,OpitcalLit,Lit3,Lit4,Lit5}. The goal of phase retrieval is to reconstruct a signal $\x \in \mathbb{C}^n$ from its phaseless samples:
\begin{equation} \label{eq:yi}
y_i = |\langle \a_i, \x \rangle|,~i = 1,\cdots, m,
\end{equation}
where $\a_i \in \mathbb{C}^n$, $i = 1, \cdots, m$, are sampling vectors, such as the discrete Fourier transform basis. Put it in the matrix form, we can write

\begin{equation} \label{eq:yAxmatrixPR}
\mathbf y= | \mathbf A \x |.
\end{equation}
In general, the signal reconstruction process involves the following nonconvex  optimization problem~\cite{lossfun1}:
\begin{equation} \label{eq:optimization1}
\min_{\x \in \mathbb{C}^n}~\frac{1}{2m} \sum_{i=1}^{m} \left(y_{i}^2-|\a_i^*\x|^2\right)^{2}. 
\end{equation}

Over the years, much effort has been made to find out the optimal solution for this problem. Gerchberg and Saxton first proposed the GS algorithm~\cite{GS}, which alternates the projections between the image and Fourier domains for error reduction. Fienup suggested a refined version of GS called the hybrid-input output (HIO) algorithm~\cite{HIO}. Although these methods can reconstruct the original signal effectively, they suffer from many drawbacks, such as slow convergence, requiring strong priors (e.g., knowing the support of input signal in advance), and lack of theoretical guarantee~\cite{criticalreviewForHIO}.

Convex optimization-based approaches for phase retrieval usually enjoy rigorous performance guarantees~\cite{PhaseliftOnlogn,PhaseCut,PhaseMax,PhaseEqual}.  PhaseLift rewrites the phaseless samples in~\eqref{eq:yi} in a linear form (i.e., 
$y_{i}^2 = \mathbf{x}^{*} \mathbf{a}_{i} \mathbf{a}_{i}^{*} \mathbf{x} = \operatorname{Tr}\left(\mathbf{x}^{*} \mathbf{a}_{i} \mathbf{a}_{i}^{*} \mathbf{x}\right) = \operatorname{Tr}\left(\mathbf{a}_{i} \mathbf{a}_{i}^{*} \mathbf{x}\mathbf{x}^*\right)$),
thus relaxing~\eqref{eq:optimization1} to a convex optimization problem~\cite{PhaseliftOnlogn}: 
\begin{eqnarray}
& \underset{\mathbf{X} \in \mathbb{C}^{n \times n}}{\min} & \operatorname{Tr}(\mathbf X) \nonumber \\
& \text { s.t. }    & y_{i}^2 = \operatorname{Tr}\left(\mathbf{a}_{i} \mathbf{a}_{i}^{*} \mathbf{X}\right),~i=1, \cdots, m, \\
&                   &\mathbf{X} \succeq 0. \nonumber
\end{eqnarray}
where $\mathbf{X} \doteq \x\x^* \in \mathbb{C}^{n \times n}$ and $\operatorname{Tr}(\cdot)$ denotes the trace of a matrix. It is shown that when $\a_i$'s are random Gaussian vectors, the PhaseLift algorithm perfectly recovers the original signal (up to global phase) provided that~\cite{PhaseliftOn}
\begin{equation} \label{eq:On}
m = \Omega (n).\footnotemark
\end{equation}\footnotetext{\label{ftnt:asynotation}The notation $\phi(n) =  \Omega(g(n))$ means that there exists a constant $c > 0$ such that $\phi(n)\geq c g(n)$ when $n$ goes to infinity. In the latter context, we use the notation $\phi(n) = \mathcal O(g(n))$ to represent that there exists a constant $c > 0$ such that $\phi(n)\leq c g(n)$ when $n$ goes to infinity, and use the notation $\phi(n) =  \Theta (g(n))$ to represent that $\phi(n)$ is of the same order as $g(n)$, i.e., there exist constants $c_1, c_2 > 0$ such that $c_1 g(n) \leq \phi(n)\leq c_2 g(n)$ when $n$ goes to infinity.}However, this algorithm relies on a semi-definite programming (SDP) approach, which could be computationally burdensome for large-scale applications. 

In order to improve the computational efficiency, gradient descent-based methods have been suggested. A well-known representative is Wirtinger flow (WF)~\cite{WF}. It has been shown that the WF algorithm with an elaborated initialization can accurately recover the target signal when the number of samples satisfies
\begin{equation} \label{eq:Onlogn}
m = \Omega (n \log n).
\end{equation}
There have also been some powerful variants of WF, such as the reshaped WF (RWF)~\cite{RWF} and truncated WF (TWF)~\cite{TrWF}. Compared to those in the convex optimization family, the gradient descent-based methods have demonstrated to be more efficient both in theory and practice. 

Recently, there has been much evidence~\cite{XuZQ} that exploiting the sparse prior of input signals can facilitate phase retrieval. In~\cite{infobound1, infobound2}, the authors show that it suffices to recover an $n$-dimensional $k$-sparse signal with the information-theoretic sampling complexity\footnote{This information-theoretic result is obtained only for the real case. The more difficult complex case should require at least this sampling complexity.}
	\begin{equation} \label{eq:klognForPR}
		m = \Omega \left (k \log n \right ).
	\end{equation}
	However, the sampling complexity required by practical algorithms, such as sparse truncated amplitude flow (SPARTA)~\cite{SPARTA}, and sparse WF (SWF)~\cite{SWF}, tends to exceed this  bound. 
	Indeed, it has been shown that both SPARTA and SWF allow the exact recovery of a $k$-sparse signal from 
	\begin{equation} \label{eq:k2logn}
		m = \Omega (k^2 \log n)
	\end{equation}
	phaseless Gaussian samples when the minimum nonzero entry of the target signal $\x$, denoted by $x_{\min}$, obeys 
	\begin{equation} \label{eq:xmin}
		{|x_{\min}|} = \Omega \left(\frac{\|\x\|}{\sqrt{k}}\right).
	\end{equation}

	The gap of sampling complexity between~\eqref{eq:klognForPR} and~\eqref{eq:k2logn} is known as the computational-to-statistical gap. It is also present in other well-known non-convex optimization approaches for sparse phase retrieval. See, for instance, Compressive Phase Retrieval with Alternating Minimization (CoPRAM)~\cite{CoPRAM}, Hard Thresholding Pursuit (HTP)~\cite{HTP}, and Stochastic Alternating Minimization (SAM)~\cite{SAM}. Those algorithms usually have a two-stage structure: i) initialization and ii) non-initial iteration. For the non-initial stage, their desired sampling complexity is approximately
	\begin{equation} 
		m = \Omega \left (k \log n \right ),
	\end{equation}
	which coincides exactly with the information-theoretic bound in~\cite{infobound1, infobound2}. Therefore, the computational-to-statistical gap is believed to arise due to the initialization stage.

	More recently, Wu and Rebeschini~\cite{HWF} put forward a new point of view on the spectral initialization used in SPARTA, which suggests a good initialization (accurate support recovery) when  
	\begin{equation}
		m = \Omega \left ( \max \left \{k\log n, \log^3 n \right \}\frac{\|\x\|^2}{|x_{\max}|^2} \right ), 
	\end{equation} 
	under the same assumption on $|x_{\min}|$ as~\eqref{eq:xmin}. 
 Then, Cai~{\it et~al.}~\cite{CJF} employed the truncated power method~\cite{TP1,SPCA2021} and derived the best initialization result to date.  
 Without relying on any assumption on $|x_{\min}|$, this method initializes an $\hat{\x}$ that falls into the $\delta$-neighborhood of $\x$ if 
	\begin{equation}\label{eq:xmaxcomplexiy}
		m = \Omega\left(\frac{\|\x\|^2}{|x_{\max}|^2 }k \log n\right).
	\end{equation}

	Clearly the result of~\eqref{eq:xmaxcomplexiy} 
	depends on the term ${\|\x\|}/{|x_{\max}|}$. For example, in the ideal case where 
	\begin{equation}
		\frac{\|\x\|}{|x_{\max}|} = \Theta (1)
	\end{equation}
	(i.e., when the maximum nonzero entry $|x_{\max}|$ dominates the entire signal power of $\x$), the sampling complexity in~\eqref{eq:xmaxcomplexiy} can be reduced to 
	\begin{equation} \label{eq:reducedcomplexity}
		m = \Omega(k\log n).
	\end{equation}
	However, in the typical case where 
	\begin{equation}
		{|x_j|} = \Theta \left ( \frac{\|\x\|}{\sqrt{k}}  \right ),~j \in \supp(\x)
	\end{equation}
	(i.e., when $\x$ is a ``flat'' signal), the sampling complexity is still $m = \Omega (k^2\log n)$.

In this paper, with an aim of enhancing the phase retrieval performance and meanwhile optimizing the theoretical guarantee, we propose a novel algorithm termed subspace phase retrieval (SPR). The principle behind the SPR algorithm is rather simple. Initializing with a correlation-promoting method, it maintains an estimated support (i.e., index set of the nonzero entries in the input signal $\x$) of size $k$, while refining this set iteratively until convergence.  

\begin{itemize}
\item 
In the initialization stage, SPR does not seek a good estimation that falls into the $\delta$-neighborhood of $\x$, or an accurate support recovery. Instead, it only requires to capture a subset of support indices with {\it sufficient energy}.  Define $\bar{s} := |\bar S|$ with $\bar S \subseteq \supp(\x)$ satisfying  
		\begin{equation} \label{eq:Sbar}
			\bar{S} = {\arg \min}_{S:  \|\x_S\|^2 \geq 0.999 \|\x\|^2}~|S|, 
		\end{equation}
		where $\x_S$ keeps the elements of $\x$ indexed by ${S}$ and sets others to zero. Then, our theoretical analysis shows that
		\begin{equation}
			m = \Omega  (\bar{s}^2 \log n) \label{eq:ourinitialization}
		\end{equation} 
		guarantees the initialization of SPR to catch a certain amount of support indices, whose corresponding nonzero entries occupy at least $90\%$ energy of $\x$ ({\bf Theorem~\ref{thm:1}}). 
		Our initialization has the following advantages: 
		
		\begin{itemize} 
			\item On the one hand, to achieve the information-theoretic sampling complexity $\Omega (k \log n)$, the value of $|x_{\max}|$ in~\eqref{eq:xmaxcomplexiy} must be comparable to $\|\x\|$. 
					Whereas in~\eqref{eq:ourinitialization},  $|x_{\max}|$ is not necessarily comparable to $\|\x\|^2$, but can be much smaller, e.g., 
					\begin{equation}\label{eq:forrebucomment11}
						{|x_{\max}|} = \Theta \left(\frac{\|\x\|}{\sqrt[4]{k}}\right),
					\end{equation}
					since it only requires $\bar{s} = \sqrt k$ nonzeros in $\x$ to have a squared energy comparable to $\|\x\|$. In this sense,~\eqref{eq:ourinitialization} may allow a wider range of $\x$ to attain this optimal sampling complexity.
			
			\vspace{1mm}
			
			\item On the other hand, when recovering a signal $\x$ with some specific structures,~\eqref{eq:ourinitialization} may indicate a lower sampling complexity than~\eqref{eq:xmaxcomplexiy}. 
					For example, if $\sqrt k$ most significant nonzero elements of $\x$ are all on the order of ${\|\x\|}/{\sqrt[4]{k}}$, then
					\begin{equation}\label{eq:forrebucomment12}
						\bar{s} = \sqrt{k}~~\text{and}~~|x_{\max}| = \Theta \left(\frac{\|\x\|}{\sqrt[4]{k}}\right).
					\end{equation}
					In this case, our result in~\eqref{eq:ourinitialization} can be reduced to $\Omega (k \log n)$, which outperforms the result 
					\begin{equation}
						m = \Omega  (k^{3/2}\log n)
					\end{equation}
					derived from~\eqref{eq:xmaxcomplexiy}. More detailed comparisons on the sampling complexity can be found in Section~\ref{sec:disTobars}. 
		\end{itemize}

\item In the non-initial stage, SPR iteratively refines the estimated support set of $\x$ with matching and pruning operations, while estimating the signal itself by minimizing a quartic objective function (i.e., a natural least-squares formulation for signal estimating; see~\eqref{eq:objfunc}). Interestingly, capturing sufficient energy in initialization directly leads to a benign geometric property of the target subspace for estimating the sparse signal ({\bf Proposition~\ref{prop:geoproOnsub}}). In particular, all local minima of our objective function are clustered around the expected global optimum with arbitrarily small distances when 
\begin{equation}
	m  =  \Omega  (k \log^3 k). \label{eq:ourgeomatric}
\end{equation}
~~The geometric structure allows the nonconvex signal estimation problem to be ``globally'' optimized by efficient iterative methods, which plays a vital role in analyzing the non-initial step of SPR. Notably, this property, in conjunction with the concentration of gradient ({\bf Proposition~\ref{lm:elemenEzxtbound}}) with an assumption on $|x_{\min}| = \Theta  ({\|\x\|}/{\sqrt{k}})$, ensures SPR to catch all the remaining support indices via the matching operation when 
\begin{equation}
	m = \Omega  \left(\max\left\{k\log n,\sqrt{k} \log^3 n\right\}\right). \label{eq:matchingcomplexity}
\end{equation}
As long as the true support is selected, the SPR algorithm exactly recovers the original signal with high probability. In summary, the sampling complexity for the non-initial stage of SPR approximately attains the information-theoretic bound in~\eqref{eq:klognForPR}.

\end{itemize}

It is worth noting that the non-initial step of SPR has a similar flavor to greedy sparse phase retrieval (GESPAR)~\cite{Gespar} in that they both refine an estimated support set iteratively through a replacement-type operation. Nevertheless, our greedy principle is sufficiently distinct. In contrast to the GESPAR algorithm that maintains an estimated support set by adding and removing one candidate at each iteration through a 2-opt approach, SPR allows multiple candidates to be updated at a time with a properly designed pruning strategy. Furthermore, we provide a sophisticated initialization as well as a recovery guarantee for SPR, for which there are no counterparts in the GESPAR study.

The rest of this paper is organized as follows. In Section~\ref{sec:SPRresult}, we introduce the sparse phase retrieval problem and the proposed SPR algorithm, and present our main results. In Section~\ref{sec:analysis}, we provide the geometric property for a subproblem of sparse phase retrieval and a probabilistic guarantee for the matching operation, and present a detailed proofs for exact recovery of sparse signals via SPR. In Section~\ref{sec:discussion}, we discuss our analysis and several issues. Numerical results are illustrated in Section~\ref{sec:simulation}. Finally, conclusion remarks are provided in Section~\ref{sec:conclusion}.
Table~\ref{tab:nota} summarizes some notations used throughout this paper. 

\begin{table}[t!]
\centering
\caption{\textsc{Summary of Notations}} 
\label{tab:nota}
\begin{tabular}{ll}
\toprule 
Notation             & Description \\
\midrule 
$|\cdot|$         & absolute value or modulus \\
$\mathsf{j}$    & imaginary unit, i.e., $\mathsf{j}^2=-1$\\
$\overline{\x}$  & conjugate of $\x\in \C^n$\\
$\x^*$            & conjugate transpose of $\x \in \mathbb {C}^n$ \\
$x_j$             & $j$th entry of $\x \in \mathbb{C}^n$ \\		
$\|\cdot\|$       & $\ell_2$-norm for a vector or spectral norm for a matrix\\
$\|\x\|_0$        & number of nonzero entries in $\x \in \mathbb{C}^n$ \\
$\text{supp}(\x)$ & support set (i.e., index set of non-zero indices) of $\x \in \mathbb{C}^n$ \\
$\x_D$           & keep the elements of $\x$ indexed by ${D}$ and set others to $0$. \\
$\hat{\x}$        & recovered signal \\
$a_{ij}$          & ($i,j$)th entry of $\mathbf A \in \mathbb{C}^{m \times n}$ \\	
$\mathbf A_D$    & keep all the rows and columns of $\mathbf A$ indexed by $D$ and \\ &  set  others to $0$ \\
$\mathcal{CN}(n)$& $n$-dimensional vector whose i.i.d. entries obey\\
& standard Complex distribution, i.e., $\mathcal{N}\left(0,\frac{1}{2}\right) + \mathsf{j}\mathcal{N}\left(0,\frac{1}{2}\right)$\\
$\nabla f$        & (Writinger) gradient\\
$\nabla^2 f$      & (Writinger) Hessian\\
$\mathbb C^D$ & subspace $\left\{\x \in \C^n | \operatorname{supp}(\x) \subseteq D \right\}$\\ 
$\mathcal C (\x, k)$ & index set of the $k$ largest entries of $\x$ in magnitude \\
$\Re[\x]$ & the real part of $\x \in \C^n$\\
$\Im[\x]$ & the imaginary part of $\x \in \C^n$\\
$\langle \x, \y\rangle$ & the inner product of $\x$ and $\y$, i.e., $\x^*\y$\\
$\mathbb{C}\mathbb{B}^k(r)$ & the sphere on the $k$-dimension space with radius $r$\\
\bottomrule 
\end{tabular}
\end{table}

\section{Phase Retrieval via Subspace Phase Retrieval}
\label{sec:SPRresult}

\subsection{Preliminaries}
Consider the sparse phase retrieval problem: 
\begin{eqnarray}
\label{prob:SPR}
& \text{Find}       & \x\in \C^n \nonumber \\ 
& \text{s.t.}	    & y_i = |\a_i^*\x|,~i = 1, 2, \cdots, m, \\
&                   & \|\x\|_0 \leq k, \nonumber
\end{eqnarray}  
which can be reformulated in a natural least-squares form~\cite{WF}:
\begin{eqnarray}
\label{prob:SPR1}
\underset{\x \in \C^n}{\min}~ \frac{1}{2m} \sum_{i=1}^{m} \left(y_i^2 - \left|\mathbf a_i^* \x \right|^2 \right)^2~~\text{s.t.}~~\|\x\|_0 \leq k. 
\end{eqnarray} 
Due to the nonconvex sparse constraint, however, it requires a combinatorial search over all possible cases of $\text{supp}(\x)$, and thus is NP hard. In this paper, we propose to solve $\x$ in a greedy fashion. Before we proceed to the details of our algorithm, we give some useful definitions. 

For simplicity, define the loss function of problem~\eqref{prob:SPR1} as:
\begin{equation} \label{eq:objfunc}
f:\C^n \mapsto \mathbb R,~f_m(\z; \mathbf A) = \frac{1}{2m} \sum_{i=1}^{m} (y_i^2 -  |\mathbf a_i^* \z |^2 )^2,
\end{equation}
where $\mathbf A$ represents the matrix composed of the measurement vectors, i.e., $\a_i$'s. We abbreviate $f_m(\z; \mathbf A)$ to $f(\z)$ when there is no ambiguity. Observe that $f(\x \mathrm{e}^{\mathsf{j}\phi}) = f(\x)$, $\forall \phi$. In other words, $\x \mathrm{e}^{\mathsf{j}\phi}$ and $\x$ are equivalent solutions to the problem~\eqref{prob:SPR1}. Thus, it is necessary to define the distance between two points in $\C^n$ under such equivalence class. A nature definition on the distance between the point $\z \in\C^n$ and $\x$ is given by
\begin{equation}
\operatorname{dist}(\mathbf{z}, \x )  \doteq
\underset{\phi \in[0,2 \pi)}{\min} \|\mathbf{z}-\mathbf{x} \mathrm{e}^{\mathsf{j}\phi}\|.
\end{equation} 

Furthermore, the derivative $\nabla f(\x)$ does not exist in general since it does not satisfy the Cauchy-Riemann equation. In other words, $f(\x)$ is not a holomorphic function. Thus, we introduce the Wirtinger derivative: 
\begin{eqnarray}
\frac{\partial f}{\partial \z} & \hspace{-2mm} \doteq & \hspace{-2mm}  \frac{\partial f(\z, \overline{\z})}{\partial \z}
= \left[\frac{\partial f(\z, \overline{\z})}{\partial z_{1}}, \cdots, \frac{\partial f(\z, \overline{{\z}})}{\partial z_{n}}\right]; \\
\frac{\partial f}{\partial \overline\z} & \hspace{-2mm}  \doteq & \hspace{-2mm}  \frac{\partial f(\z, \overline{\z})}{\partial \overline \z}
= \left[\frac{\partial f(\z, \overline{\z})}{\partial \overline z_{1}}, \cdots, \frac{\partial f(\z, \overline{{\z}})}{\partial \overline z_{n}}\right]. 
\end{eqnarray} 
The Wirtinger derivative treats $f$ as a binary function of $\z$ and $\overline \z$, and computes their derivatives, respectively; see~\cite{wirtingerref} for a comprehensive tutorial of Wirtinger Calculus. 

So far, the Wirtinger derivative has been widely adopted to analyze the phase retrieval problem where the input signal is complex~\cite{WF, sunju}. 
In terms of the Wirtinger derivative, the Wirtinger gradient and Hessian matrix can be given by
\begin{equation}\label{eq:wirtingerG}
\nabla f(\z)  =\left[ \frac{\partial f}{\partial \z}, \frac{\partial f}{\partial \overline{\z}} \right]^{*}  
\end{equation}
and 
\begin{eqnarray} \label{eq:wirtingerH}
\nabla^{2} f(\z) = 
\begin{bmatrix}
{ \frac{\partial}{\partial \z} \left( \frac{\partial f}{\partial \z} \right)^{*} } 
& 
{\frac{\partial}{\partial \overline{\z}}\left(\frac{\partial f}{\partial \z}\right)^{*}} \\ 
{\frac{\partial}{\partial \z}\left(\frac{\partial f}{\partial \overline{\z}}\right)^{*}} &
{\frac{\partial}{\partial \overline{\z}}\left(\frac{\partial f}{\partial \overline{\z}}\right)^{*}}
\end{bmatrix},
\end{eqnarray}
respectively. Hereafter, we write 
\begin{equation}
\nabla_1 f(\z)=  \left(\frac{\partial f}{\partial \z}\right)^*~\text{and}~~\nabla_2 f(\z)=  \left(\frac{\partial f}{\partial \overline{\z}}\right)^* 
\end{equation}
for notational convenience.  
The concrete forms of Wirtinger gradient and Hessian can be given by
\begin{eqnarray} \label{eq:10daoshu}
\nabla f(\z) &\hspace{-2mm}=&\hspace{-2mm} \frac1m \sum_{i=1}^m \begin{bmatrix}((|\a_i^* \z|^2 - y_i^2)\a_i\a_i^*)\z\\
((|\a_i^* \z|^2 - y_i^2)(\a_i\a_i^*)^\top)\overline{\z}\end{bmatrix} \\
\nabla^2 f(\z) &\hspace{-2mm}=&\hspace{-2mm}    \frac1m  \nonumber  \\
&\hspace{-2mm} &\hspace{-2mm} \times \sum_{i=1}^m \begin{bmatrix}(2|\a_i^* \z|^2 - y_i ^ 2)\a_i\a_i^* &\hspace{-2mm} (\a_i ^* \z)^2 \a_i\a_i^\top \\ (\z^*\a_i)^2 \overline{\a}_i \a_i^* &\hspace{-2mm} (2|\a_i^* \z|^2 - y_i^2) \overline{\a}_i\a_i^\top\end{bmatrix}\hspace{-1mm}. \nonumber \\ \label{eq:11hessian}
\end{eqnarray}

Particularly, for the Wirtinger gradient and Hessian matrix 
indexed on a set $D \subseteq \{1, \cdots, n\}$, we denote
\begin{equation}\label{eq:wirtingerGindexedT}
\nabla f(\z)_D  =\left[ \left(\frac{\partial f}{\partial \z}\right)_D, \left(\frac{\partial f}{\partial \overline{\z}}\right)_D \right]^{*},
\end{equation}
and
\begin{eqnarray} \label{eq:wirtingerHindexedT}
\nabla^{2} f(\z)_D = 
\begin{bmatrix}
{ \left(\frac{\partial}{\partial \z} \left( \frac{\partial f}{\partial \z} \right)^{*} \right)_D} 
& 
{\left(\frac{\partial}{\partial \overline{\z}}\left(\frac{\partial f}{\partial \z}\right)^{*}\right)_D} \\ 
{\left(\frac{\partial}{\partial \z}\left(\frac{\partial f}{\partial \overline{\z}}\right)^{*}\right)_D} &
{\left(\frac{\partial}{\partial \overline{\z}}\left(\frac{\partial f}{\partial \overline{\z}}\right)^{*}\right)_D}
\end{bmatrix}.
\end{eqnarray}

\subsection{The SPR algorithm}\label{SEC:IIB}

Now we are ready to describe the SPR algorithm, which consists of two major steps: i) initialization and ii) non-initial step.  {In the initialization step, a support estimate is generated by means of a correlation-promoting initialization; see Algorithm~\ref{alg:sp_in}. Specifically, 
it computes the correlations
\begin{equation}
    Z_j = \frac{1}{m}\sum_{i=1}^m y_i |a_{ij}|,~j = 1, \cdots, n
\end{equation}
and sorts them to obtain an index set ${S}^{0}$, whose elements correspond to the $k$ largest values of $Z_j$'s. Intuitively, this step can be viewed as finding out the strongest $k$ correlations between $\y$ and the columns of $|\mathbf A|$, where $|\mathbf A| \in \mathbb{R}^{m \times n}$ has each element being the modulus of that in $\mathbf A$. The correlation-promoting strategy for generating the initialized support is different from that of SPARTA~\cite{SPARTA}. }

In the non-initial step of SPR, three major operations are involved. 
\begin{enumerate}[i)]
\item The first operation is called ``matching'', in which the Wirtinger gradient $\nabla_1 f(\x^{t - 1})$ is computed. Then, indices corresponding to the largest $k$ elements (in magnitude) are chosen as the new elements of the estimated support set:
\begin{equation}
\hat {S}^{t} \leftarrow S^{t - 1} \cup \mathcal C(\nabla_1 f(  \x^{t - 1}), k),
\end{equation}
where $\mathcal C (\cdot, k)$ is a function that returns the index set of the largest $k$ entries (in magnitude) of a vector. 

\item In the second operation, a signal estimate whose support is $\hat S^t$ is obtained via the following optimization:
\begin{equation} \label{eq:subproblem1}
\mathbf x^{t}\gets   {\arg \min}_{\mathbf z : \operatorname{supp}(\mathbf z) = \hat{S}^{t}}  f(\mathbf z),
\end{equation}
This operation, termed ``estimation'', can be viewed as solving a subproblem of~\eqref{prob:SPR1}. The solution almost surely exists since $f(\z)$ is almost surely coercive\footnote{\label{ftnt:coercive}This is because the probability of $m$ $k$-dimensional independent Gaussian random vectors generating the $k$-dimensional space is $1$ when $m\geq k$. Thus, we have $\P (\lim_{\operatorname{supp}(\mathbf z) = \hat{S}^{t}, \|\z\|\to \infty} f(\z)=\infty )=1$, which means $f(\z)$ is almost surely coercive.}. 

~~We shall show from theoretical analysis that the geometric property of this subproblem is incredibly sound when the estimated support set $\hat{S}^{t}$ contains a certain amount of support indices of $\x$. As a result, a wide range of iterative algorithms (e.g., the perturbed gradient descent (PGD)~\cite{PGD} and Barzilai-Borwein method (BB)~\cite{BBalg}) can solve~\eqref{eq:subproblem1} efficiently.

\item The third operation, which is referred to as ``pruning'', narrows down the number of candidates in $\hat S^t$ to $k$ corresponding to the most significant $k$ entries in the signal estimate $\x^{t}$:
\begin{equation}
S^{t} \gets \mathcal C (\x^{t}, k). 
\end{equation}
The role of pruning is to prevent catching too many
indices. Specifically, if we keep performing the matching and projection operations, the size of the estimated support set will increase fast, which brings difficulties to the signal estimation. For example, when $|\hat{S}^{t}| > m$, it would by no means be possible to obtain an accurate estimation of signal even for the linear model $\y = \mathbf A \x$, as the least square projection (i.e., $(\mathbf A_{\hat{S}^{t}}^*\mathbf A_{\hat{S}^{t}})^{-1}\mathbf A_{\hat{S}^{t}}^*  \y$) does not apply in this under-determined case. 
\end{enumerate}

\begin{algorithm}[t!]
\caption{\textsc{Correlation-Promoting Initialization}}
\label{alg:sp_in}
\begin{algorithmic}[1]
\STATE {\textbf{Input}: sparsity $k$, $\mathbf y, \mathbf A=\{a_{ij}\}_{i=1,j=1}^{m\times n}$.} 
\STATE Compute  $Z_j = \frac{1}{m}\sum_{i=1}^m y_i |a_{ij}|$.
\STATE Sort $\{Z_j\}_{j=1}^n$ to generate an index set ${S}^{0}$ whose elements correspond to the $k$ largest values of $Z_j$'s.
\STATE \textbf{Output}: ${S}^{0}$
\end{algorithmic}
\end{algorithm} 

\begin{algorithm}[t!]
\caption{\textsc{{Subspace Phase Retrieval}} }
\label{alg:SPR}
\begin{algorithmic}[1]
\STATE {\textbf{Input}: sparsity $k$, samples $\mathbf y$, sampling matrix $\mathbf A$, tolerance $\delta$, maximal iteration number $t_{\max}$.} 
\STATE {\textbf{Initialize}: iteration count $t \gets 0$,\\
	${S}^{0}  \gets $ correlation-promoting initialization, \\  
	$\mathbf x^{0} \gets   {\arg \min}_{\mathbf z : \operatorname{supp}(\mathbf z) = {S}^{0}} f(\mathbf z)$}.

\WHILE {$t < t_{\max}$ and \(f ( \mathbf x^{t} )\geq \delta \)}	
\STATE $t \gets t + 1$.
\STATE \textbf{Matching}: \(\hat{S}^{t} \leftarrow S^{t - 1} \cup \mathcal C(\nabla_1 f(\x^{t - 1}), k)\).
\STATE \textbf{Estimation}: \(\mathbf \x^{t} \gets {\arg \min}_{\mathbf z: \operatorname{supp}(\mathbf z)=\hat{S}^{t}} f(\mathbf z)\).
\STATE \textbf{Pruning}: \(S^{t} \gets \mathcal C (\x^{t}, k)\).
\ENDWHILE
\STATE \textbf{Output}: estimated signal $\hat{\mathbf x} = \mathbf x^{t}$.
\end{algorithmic}
\end{algorithm}

The mathematical formulation of SPR is formally specified in Algorithm~\ref{alg:SPR}. In the non-initial step, the idea of maintaining an estimated support of size $k$ with a pruning strategy is inspired from Compressive Sampling Matching Pursuit (CoSaMP)~\cite{Cosamp} and Subspace Pursuit (SP)~\cite{SP_CS}, which are well-known algorithms for recovering sparse signals from linear samples.

\subsection{Main Results} 
\label{sec:TheoreticalResults}
In order to demonstrate the effectiveness of SPR, we analyze the initialization and non-initial steps separately. As in~\cite{sunju},  we consider that the sampling matrix $\mathbf A$ is generic with independent identically distributed ({\it i.i.d.}) entries obeying the standard complex Gaussian distribution (i.e., $a_{ij} \overset{i.i.d.}{\sim}  \mathcal{CN}(1) = \mathcal N\left(0, \frac{1}{2}\right)+\mathsf{j}\mathcal N\left(0, \frac{1}{2}\right)$).

First, the following theorem gives a condition that guarantees the SPR initialization to catch a subset of $\supp(\x)$ with sufficient energy. 
 
\begin{thm}[\textit{Initialization}]\label{thm:1}
Consider the sparse phase retrieval problem~\eqref{prob:SPR1}. With probability exceeding $1- \exp\left( -c \bar{s} \log  \frac{n}{\bar{s}} \right)$, SPR initializes an estimated support $S^0$ of size $k$ satisfying $\frac{ \| \x_{S^0} \|^2}{\| \x \|^2} > \frac{9}{10}$ provided that $m \geq C \bar{s}^2\log \frac{n}{\bar{s}}$. 
\end{thm}

Throughout the paper, we follow the convention that letters $c$ and $C$, and their indexed versions (e.g., $c_1$, $C_1$, and etc.) indicate positive, universal constants that may vary at each appearance.

Next, we move on to characterizing the behavior of SPR in the non-initial step, given that $S^0$ was initialized with sufficient energy. 
If $S^0$
is sufficiently accurate, we expect the sequence
{$S^t$} to converge toward the true support of $\x$. 
However, deriving a guarantee for this step is more complex. One reason is that the estimation operation of SPR in~\eqref{eq:subproblem1} has no closed-form solution, since it minimizes a quartic objective function. This is in contrast to the WF algorithm~\cite{WF} that produces a signal estimate in each iteration with an explicit expression (i.e., the gradient descent form)~\cite{WF}. 

Our estimation operation also differs with the conventional least squares projection in the linear sampling model, which directly has a pseudo-inverse solution over the estimated support set (see compressive sensing, e.g.,~\cite{OMP1,davenport2010analysis,OMP2,livshitz2014sparse, SP_CS, Cosamp,D2DPeng}). In fact, the estimation operation of SPR having no analytical solution has posted a major challenge for evaluating the recovery performance theoretically.

To deal with this challenge, we connect it to an ideal case where there are infinitely many samples available, which is essentially the expectation case~\cite{sunju}. In this case, the estimation operation of SPR does have an analytical solution. We will show that if at least one support index was caught in the previous iteration (i.e., $S^0 \cap \text{supp}(\x) \neq \emptyset$), then the original signal $\x$ can be recovered correctly.

\begin{thm}[{\it Expectation case}]
\label{thm:2}
Consider the sparse phase retrieval problem~\eqref{prob:SPR1} with $m \rightarrow \infty$. If the estimated support in the initialization of SPR satisfies $S^0 \cap \operatorname{supp}(\x) \neq \emptyset$, then 
$ 
\x^{1} = \x \mathrm{e}^{\mathsf{j}\phi}
$
for some~$\phi \in [0, 2\pi)$. 
\end{thm}

Then, we will show that SPR does not really need infinite number of samples to ensure the success of recovery. Our result is formally described in the following theorem.

\begin{thm}[\textit{Finite sampling case}]
\label{thm:3}
Consider the sparse phase retrieval problem~\eqref{prob:SPR1}, where the minimum nonzero entry of $\x$ satisfies $|x_{\min}| = \Omega (\frac{\|\x\|}{\sqrt{k}})$ and 
\begin{equation}
m \geq  C\max\left\{k\log^3 k, k\log n,\sqrt{k}\log^3 n\right\}.
\end{equation}
If the estimated support in the initialization of SPR satisfies $|S^0|=k$ and $\frac{ \| \x_{S^0} \|^2}{\| \x \|^2} > \frac{9}{10}$, then with probability exceeding $1-cm^{-1}$,
$
\x^{1} = \x \mathrm{e}^{\mathsf{j}\phi}
$ 
for some~$\phi \in [0, 2\pi)$.  
\end{thm}

One can interpret from this theorem that as long as the initialization step catches a set of support indices with sufficient energy (which can be readily achieved via our correlation-promoting initialization, as shown in Theorem~\ref{thm:1}),  SPR can recover the input signal exactly up to a global phase with high probability.  
Combining Theorems \ref{thm:1} and \ref{thm:3} leads to an overall guarantee for the SPR algorithm. 
\begin{thm}[\textit{Overall condition}]
\label{thm:4}
Consider the sparse phase retrieval problem~\eqref{prob:SPR1}. If 
\begin{equation}
m \geq  C\max\left\{\bar{s}^2\log (n/\bar{s}), k\log^3 k, k\log n,\sqrt{k}\log^3 n\right\}
\end{equation}
and the minimum nonzero entry of $\x$ satisfies $|x_{\min}| = \Omega ({\|\x\|}/{\sqrt{k}})$, then with probability exceeding $1-cm^{-1}$, SPR returns the true signal 
$
\hat{\x} =  \x^{1} = \x \mathrm{e}^{\mathsf{j}\phi}
$ 
for some~$\phi \in [0, 2\pi)$.  
\end{thm}

\section{Analysis} 
\label{sec:analysis}

\subsection{Proof of Theorem \ref{thm:1}} \label{sec:proofofthm1} 
We first introduce Bernstein's inequality~\cite{hdpBook}, which is useful for our analysis.
\begin{lem}[\textit{Bernstein's inequality~\cite[Chap. 2.8]{hdpBook}}]\label{lam:B}
Let $X_1, X_2,\cdots, X_m$ be {\it i.i.d.} copies of the subexponential variable $X$ with parameter $\sigma^2$ (i.e., squared subexponential norm), then
\begin{eqnarray}
\lefteqn{\P\left( \left| \frac{1}{m}\left(\sum_{i=1}^m X_i - \E[X]\right) \right|\leq \epsilon\right) } \nonumber \\
&&~~~~~~\leq 2\exp\left(-cm\min\left(\frac{\epsilon^2}{\sigma^2}, \frac{\epsilon}{\sigma}\right)
\right).
\end{eqnarray}
holds for any positive $\epsilon$, where $c$ is an absolute positive constant.
\end{lem}

Now we present the proof for Theorem~\ref{thm:1}.
\begin{proof}
Denote 
\begin{equation} \label{eq:ZjZij}
Z_{i, j} = |\a_{i}^*\x| |a_{ij}|~~\text{so that}~~Z_j = \frac{1}{m}\sum_{i=1}^m 	Z_{i, j},
\end{equation} 
where $i \in \{1, \cdots, m\}$ and $j \in \{1, \cdots, n\}$.
One can view $\frac{\a_{i}^*\x}{\|\x\|}$ and $a_{ij}$ as the random variables at two different time points of a complex Gaussian process, respectively. According to~\cite[Chap. 2.1.4]{PrinMobCom}, the expectation of $Z_{i, j}$
is equivalent to the auto-correlation of the envelope
of the complex Gaussian process. Thus, by following the result in~\cite[Eq.~(2.70)]{PrinMobCom}, we have
\begin{eqnarray}\label{eq:EZij}
\mathbb{E}\left[Z_{ i,j}\right]  &=&  \frac{\pi}{4}  \|\mathbf{x}\|
F\left(-\frac{1}{2}, - \frac{1}{2}; 1; \frac{|x_j|^2}{\|\x\|^2} \right)  \nonumber \\
&=& \frac{\pi}{4} \|\mathbf{x}\|  \left (1 + \frac{1}{4} \frac{|x_j|^2}{\|\x\|^2} + \frac{1}{64} \frac{|x_j|^4}{\|\x\|^4}  + \cdots  \right ),~~~~~\label{eq:Expectation}
\end{eqnarray} 
where $F(a,b,c;z)$ is the hypergeometric function defined by the power series:\footnote{Strictly speaking, the hypergeometric function is defined on the complex plane $\{z\in\mathbb C: |z|<1\}$. But it can be analytically continued to $\{z\in\mathbb C:|z|\geq 1\}$. The principal branch is obtained by introducing a cut from $1$ to $+\infty$ on the real $z$-axis. We refer the interested readers to~\cite[\href{https://dlmf.nist.gov/15}{Chapter 15}]{NIST:DLMF} and~\cite{geofunctionSlater}.}
\begin{align}
F(a,b,c;z) &= \sum_{l = 0}^\infty \frac{(a)_l (b)_l} {(c)_l} \frac{z^l}{l !} \nonumber \\
&= 1 + \frac{ab}{c} \frac{z}{1!} + \frac{a (a + 1) b (b + 1)}{c (c + 1)} \frac{z^2}{2 !} + \cdots. 
\end{align}  
Here
\begin{equation}
(a)_l = \begin{cases}
	1&l=0,\\
	\prod_{s=0}^{l-1} (a+s)&l\geq 1,
\end{cases}
\end{equation}
$(b)_l$ and  $(c)_l $ are defined in the same way. 

Next, we shall estimate $\mathbb{E}\left[Z_{ i,j}\right]$ by considering two cases. 
\begin{itemize}
\item {\bf Case 1:}  When $j\not \in \supp(\x)$,~\eqref{eq:Expectation} directly implies
\begin{equation}
	\mathbb{E}\left[Z_{ i,j}\right]   = \frac{\pi}{4} \|\mathbf{x}\|.
\end{equation}

\item {\bf Case 2:} When $j \in \supp(\x)$, we have $|x_j| > 0$. In this case, we can 
derive an upper and a lower bound for $\mathbb{E}\left[Z_{ i,j}\right]$, respectively.
First, since $\big ( \frac{|x_j|^2}{\|\x\|^2} \big )^N> 0$  for $N = 2, 3, \cdots$, it follows from~\eqref{eq:Expectation} that 
\begin{eqnarray}
	\mathbb{E}\left[Z_{ i,j}\right] >  \frac{\pi}{4} \|\mathbf{x}\| + \frac{\pi}{16}\frac{|x_j|^2}{\|\x\|}. \label{eq:Expectation1}
\end{eqnarray} 
As for the upper bound, since $\frac{|x_j|^{2}}{\|\x\|^{2}} \leq 1$, we have 
\begin{eqnarray}
	\mathbb{E}\left[Z_{ i,j}\right]  & \hspace{-2mm}  \leq & \hspace{-2mm}  \frac{\pi}{4} \|\mathbf{x}\|  \left (1 + \left(\frac{1}{4} + \frac{1}{64}+\cdots \right)\frac{|x_j|^2}{\|\x\|^2} \right ) \nonumber \\
	& \hspace{-2mm} = & \hspace{-2mm} \frac{\pi}{4} \|\mathbf{x}\| \nonumber \\
	& \hspace{-2mm}   & \hspace{-2mm} \times \left (1 + \left(F\left(-\frac{1}{2}, - \frac{1}{2}; 1; 1 \right)-1  \right)\frac{|x_j|^2}{\|\x\|^2} \right ) \nonumber \\
	& \hspace{-2mm} \overset{(a)}{=} & \hspace{-2mm}   \frac{\pi}{4} \|\mathbf{x}\|    \left (1 + \left(\frac{\Gamma(1)\Gamma(2)}{\Gamma \big (\frac{3}{2} \big )\Gamma \big (\frac{3}{2} \big )}-1\right) \frac{|x_j|^2}{\|\x\|^2}  \right )\nonumber\\
	& \hspace{-2mm} = & \hspace{-2mm}  \frac{\pi}{4} \|\mathbf{x}\|    \left (1 + \frac{4-\pi}{\pi}\frac{|x_j|^2}{\|\x\|^2}  \right ) \nonumber\\
	& \hspace{-2mm} < & \hspace{-2mm}  \frac{\pi}{4} \|\mathbf{x}\|    \left (1 + \frac{11}{40}\frac{|x_j|^2}{\|\x\|^2}  \right )\nonumber\\
	& \hspace{-2mm} = & \hspace{-2mm}  \frac{\pi}{4} \|\mathbf{x}\|   + \frac{11\pi}{160}\frac{|x_j|^2}{\|\x\|}, 
\end{eqnarray} 
where (a) is from the Gauss's summation identity~\cite[Appendix III, Eq.~(III.3)]{geofunctionSlater}.
\end{itemize}

In summary,\footnote{If all quantities are real, $\mathbb{E}\left[Z_{ i,j}\right] =  \frac{2}{\pi} \|\mathbf{x}\| + \frac{1}{6}\frac{|x_j|^2}{\|\x\|}$, as shown in~\cite{AltMin}.}  
\begin{equation}
\begin{cases}
	\mathbb{E}\left[Z_{ i,j}\right] - \frac{\pi}{4} \|\mathbf{x}\| = 0,                     & \text{if}~j\not \in \supp(\x), \\ 
	\mathbb{E}\left[Z_{ i,j}\right] - \frac{\pi}{4} \|\mathbf{x}\|  \in \left (\frac{\pi}{16}\frac{|x_j|^2}{\|\x\|}, \frac{11\pi}{160}\frac{|x_j|^2}{\|\x\|} \right ),                       & \text{if}~j\in \supp(\x). 
\end{cases} \label{eq:jjjjffff}
\end{equation}

One can observe a nontrivial difference in $\E[{Z_{i,j}}]$ for $j \in\supp(\x)$ or not. The lower and upper bounds of this difference will play an important role in our analysis. 

Recall the definition of $\bar{S}$:
\begin{equation}\label{eq:defbarS}
	\bar{S} = {\arg \min}_{S:  \|\x_S\|^2 \geq 0.999 \|\x\|^2}~|S|, 
\end{equation}
and $\bar{s} = |\bar{S}|$. Intuitively, $\bar{S}$ can be constructed by sequentially embracing the largest $|x_j|$'s until their squared energy exceeds $0.999\|\x\|^2$.  Let $S^\dagger$ be a subset of $S^0$ that contains $\bar{s}$  indices corresponding to the most significant $\bar{s}$ values among $Z_j$'s. Then, we shall prove that with high probability, \begin{equation}\label{eq:sstarbound}
	\frac{ \| \x_{S^\dagger} \|^2}{\| \x \|^2} > \frac{9}{10}.
\end{equation}
For simplicity, we define 
\begin{equation} \label{eq:uivi}
\begin{cases}
	V_i = \frac{1}{\bar{s}}\sum_{j\in \bar{S}} Z_{i, j}, \\ 
	U_i = \frac{1}{\bar{s}}\sum_{j\in S^\dagger} Z_{i, j}. 
\end{cases}  
\end{equation} 
Note that 
\begin{align}
\nonumber	\frac{1}{m}\sum_{i=1}^mU_i &=  \frac{1}{m}\sum_{i=1}^m\frac{1}{\bar{s}}\sum_{j\in S^\dagger} Z_{i, j} \\
\nonumber	&=\frac{1}{\bar{s}}\sum_{j\in S^\dagger}  \frac{1}{m}\sum_{i=1}^m Z_{i, j}\\
&=\frac{1}{\bar{s}}\sum_{j\in S^\dagger}  Z_j
\end{align}
and
\begin{align}
\nonumber	\frac{1}{m}\sum_{i=1}^mV_i &=  \frac{1}{m}\sum_{i=1}^m\frac{1}{\bar{s}}\sum_{j\in \bar{S}} Z_{i, j} \\
\nonumber	&=\frac{1}{\bar{s}}\sum_{j\in \bar{S}}  \frac{1}{m}\sum_{i=1}^m Z_{i, j}\\
&=\frac{1}{\bar{s}}\sum_{j\in \bar{S}}  Z_j
\end{align}
By the definition of $S^\dagger$,  it is easy to see that
\begin{equation}\label{eq:UiVi}
\frac{1}{m}\sum_{i=1}^mU_i \geq \frac{1}{m}\sum_{i=1}^m V_i.
\end{equation}

Then, we proceed to show~\eqref{eq:sstarbound} by contradiction.

\begin{itemize}
\item On the one hand, we obtain from~\eqref{eq:jjjjffff} and~\eqref{eq:defbarS} that
\begin{equation}\label{eq:Evi}
	\E\left[ V_i \right]> \left(\frac{\pi}{4}+\frac{999\pi}{16000\bar{s}}\right)\|\x\|.
\end{equation}

Note that 
\begin{equation}
V_i= |\a_{i}^*\x| \left(\frac{1}{\bar s} \sum_{j\in \bar{S}} |a_{i,j}| \right). 
\end{equation} 
It is easy to verify that $|\a_{i}^*\x|$ is a subgaussian variable with parameter $c_1\|\x\|^2$ (i.e., squared subgaussian norm) and that $\frac{1}{\bar{s}} \sum_{j\in \bar{S}} |a_{i,j}|$ is the average of $\bar{s}$ independent subgaussian variables with parameter $c_2\frac{1}{\bar{s}}$. Thus, $V_i$  is the product of two subgaussian variables, which is subexponential with parameter less than $c_3\frac{\|\x\|^2}{\bar{s}}$ (i.e., squared subexponential norm)~\cite[Lemma 2.7.7]{hdpBook}. Here, $c_1,c_2$ and $c_3$ are some positive constants. 

~~Applying Bernstein's inequality in Lemma~\ref{lam:B} yields the tail bound
\begin{eqnarray}
	\lefteqn{ \P\left(  \frac{1}{m}\sum_{i=1}^m V_i \leq \E\left[V_i\right] -\epsilon\right) } \nonumber \\
	&&~~~~\leq 2\exp\left(-c_4m \min\left(\frac{\bar{s}\epsilon^2}{\|\x\|^2}, \frac{\sqrt{\bar{s}}\epsilon}{\|\x\|}\right)
	\right),~~\label{EQ:V}
\end{eqnarray}
where $c_4$ is an absolute positive constant. Using~\eqref{eq:Evi} and taking $\epsilon=\frac{\pi}{16000\bar{s}}\|\x\|$, we further have 
\begin{equation}\label{Eq:Vconclusion0}
	\P\left( \frac{1}{m}\sum_{i=1}^m V_i\hspace{-0.5mm} \leq \hspace{-0.5mm}\left(\frac{\pi}{4}\hspace{-0.5mm}+\hspace{-0.5mm}\frac{998\pi}{16000\bar{s}}\right)\|\mathbf{x}\| \right)
	\hspace{-1mm}\leq 2\exp\left(-\frac{c_5 m}{\bar{s}}\right).	    
\end{equation}
Therefore, when $m \geq \frac{\bar{s}^2}{c_5}\log (n/\bar{s})$, it holds with probability exceeding $1-2\exp\left(-c_5  \bar{s}\log (n/\bar{s})\right) $
that 
\begin{equation}\label{Eq:Vconclusion}
	\frac{1}{m}\sum_{i=1}^mV_i >  \left(\frac{\pi}{4}+\frac{998\pi}{16000\bar{s}}\right)\|\mathbf{x}\|.
\end{equation}

\item On the other hand, if 
\begin{equation}
	\frac{\|\x_{S^\dagger}\|^2}{\|\x\|^2} \leq \frac{9}{10},
\end{equation}
we aim to prove that \begin{equation}\label{Eq:Uconclusion1}
		\frac{1}{m}\sum_{i=1}^m U_i <  \left(\frac{\pi}{4}+\frac{998\pi}{16000\bar{s}}\right)\|\mathbf{x}\| < \frac{1}{m}\sum_{i=1}^m V_i,
	\end{equation}
	which contradicts with~\eqref{eq:UiVi}. However, it is difficult to directly bound $U_i$'s as $V_i$'s since $a_{i,j}$'s with $j\in S^\dagger$ are dependent subgaussian variables. Thus the squared subgaussian norm of $U_i$ is no more $\mathcal O({\|\x\|^2}/{\bar{s}})$. 
	To bypass this obstacle, we will analyze the ``unconditional'' sum then apply a union bound. Specifically, we assume an arbitrary set $S^\ddagger$ of cardinality $\bar{s}$ such that 
	$\frac{\|\x_{S^\ddagger}\|^2}{\|\x\|^2} \leq \frac{9}{10}$ and define
	\begin{equation}
		U_i^\ddagger = \frac{1}{\bar{s}}\sum_{j\in S^\ddagger} Z_{i, j}. 
\end{equation}

We obtain from~\eqref{eq:jjjjffff} that
\begin{eqnarray}\label{eq:Eui}
	\E\left[ U_i^\ddagger \right] &<& \frac{\pi}{4}\|\x\| + \frac{11\pi}{160\bar{s}} \frac{\|\x_{S^\ddagger}\|^2}{\|\x\|} \nonumber \\
	& <&  \frac{\pi}{4}\|\x\| + \frac{99\pi}{1600 \bar{s}} \|\x\|.
\end{eqnarray}
Applying the same method used for $V_i$ yields that $U_i^\ddagger$ is also subexponential with parameter $c_6 {\|\x\|^2}/{\bar{s}}$.

~~Applying Bernstein's inequality produces the tail bound
\begin{eqnarray}
	\lefteqn{ \P\left(  \frac{1}{m}\sum_{i=1}^m U_i^\ddagger \geq \E\left[U_i^\ddagger\right] + \epsilon\right) } \nonumber \\
	&&~~~~\leq 2\exp\left(-c_7m \min\left(\frac{\bar{s}\epsilon^2}{\|\x\|^2}, \frac{\sqrt{\bar{s}}\epsilon}{\|\x\|}\right)
	\right),\label{EQ:U}~~~~~
\end{eqnarray}
where $c_7$ is an absolute positive constant. 

~~By taking $\epsilon =  \frac{\pi}{16000\bar{s}}\|\x\|$ and using~\eqref{eq:Eui}, we have
\begin{align}\label{Eq:Uconclusion0}\nonumber
	\P\left( \frac{1}{m}\sum_{i=1}^m U_i^\ddagger \geq  \left(\frac{\pi}{4}+\frac{991\pi}{16000\bar{s}}\right)\|\mathbf{x}\| \right)
	\\
	\leq 2\exp\left(-\frac{c_8m}{\bar{s}}\right).
\end{align}

Finally, since the total number of such $S^\ddagger$'s is no more than ${n\choose \bar{s}}$, taking a union bound yields
	\begin{align}\nonumber
		\P&\left( \frac{1}{m}\sum_{i=1}^m U_i \geq  \left(\frac{\pi}{4}+\frac{991\pi}{16000\bar{s}}\right)\|\mathbf{x}\| \right)
		\\
		\nonumber&\hspace{2cm}\leq 2{n\choose \bar{s}}\exp\left(-\frac{c_8m}{\bar{s}}\right)\\
		&\hspace{2cm}\leq 2\exp\left( c_9 \bar{s} \log \frac{n}{\bar{s}} - \frac{c_8m}{\bar{s}}\right).\label{eq:unionboundinInit}
\end{align}
Therefore, when $m\geq c_{10}{\bar{s}} ^2\log \frac{n}{\bar{s}}$, it holds with probability exceeding $1-2\exp\left( -c_{11} \bar{s} \log \frac{n}{\bar{s}} \right)$
that 
\begin{equation}\label{Eq:Uconclusion}
	\frac{1}{m}\sum_{i=1}^m U_i <  \left(\frac{\pi}{4}+\frac{991\pi}{16000\bar{s}}\right)\|\mathbf{x}\| .
\end{equation}
\end{itemize}

In the end, by relating~\eqref{Eq:Vconclusion} and~\eqref{Eq:Uconclusion}, we can conclude that with probability exceeding  $1- \exp\left( -c  \bar{s} \log \frac{n}{\bar{s}} \right)$, 
\begin{equation} \label{eq:condition2}
\frac{1}{m}\sum_{i=1}^mU_i < \frac{1}{m}\sum_{i=1}^mV_i, 
\end{equation}
which contradicts with~\eqref{eq:UiVi}.

Therefore,  when $m\geq C{\bar{s}}^2\log \frac{n}{\bar{s}}$, with probability exceeding  $1- \exp\left( -c \bar{s} \log \frac{n}{\bar{s}} \right)$,	we have
\begin{equation}
\frac{\|\x_{S^\dagger}\|^2}{\|\x\|^2} > \frac{9}{10}.
\end{equation}
Since $S^\dagger\subseteq S^0$, we obtain that
\begin{equation}
	\frac{\|\x_{S^0}\|^2}{\|\x\|^2} \geq\frac{\|\x_{S^\dagger}\|^2}{\|\x\|^2} > \frac{9}{10},
\end{equation}
which completes the proof.

\end{proof}

\subsection{Proof of Theorem \ref{thm:2}} \label{sec:proofofthm2} 

Before presenting the details of proof, we give some useful observations on the estimation operation of SPR in~\eqref{eq:subproblem1}. While this operation has no analytical solution in general, we will show from the geometric perspective that obtaining a closed-form solution is still possible in the expectation case, i.e., when there are infinitely many samples. 

To begin with, we introduce the concrete expected forms of~\eqref{eq:objfunc},~\eqref{eq:wirtingerG}, and~\eqref{eq:wirtingerH}, respectively, which were established in~\cite[Section VII.B]{WF} and~\cite[Lemma 6.1]{sunju}.  
\begin{eqnarray}
\label{eq:Expf}
\mathbb{E}[f(\z)] &\hspace{-2mm} = & \hspace{-2mm}  \|{\x}\|^{4}+\|{\z}\|^{4}-\|{\x}\|^{2}\|{\z}\|^{2}-\left|{\x}^{*} {\z}\right|^{2}, ~~\\
\label{eq:Expg}
\nabla_1 \mathbb{E}[f(\z)] & \hspace{-2mm} = & \hspace{-2mm} 
((2\|\z\|^2 - \|\x\|^2)\boldsymbol{I} - \x\x^*)\z,   \\
\label{eq:Exph}
\nabla^2 \mathbb{E}[f(\z)] & \hspace{-2mm} = & \hspace{-2mm} 
\begin{bmatrix}
\mathbf{B}_{11}& \mathbf{B}_{12}\\ \mathbf{B}_{21} & \mathbf{B}_{22}  
\end{bmatrix},
\end{eqnarray} 
where
\begin{align}
\mathbf{B}_{11} &= 2\z\z^* - \x\x^* + (2\|\z\|^2 - \|\x\|^2)\boldsymbol{I},\\
\mathbf{B}_{12} &= 2\z\z^\top,\label{eq:forrebucomment41}\\
\mathbf{B}_{21} &= 2\overline{\z}\z^*,\label{eq:forrebucomment42} \\
\mathbf{B}_{22} &= 2\overline{\z}\z^\top - \overline{\x}\x^\top +(2\|\z\|^2 - \|\x\|^2)\boldsymbol{I}.
\end{align}

For analytical convenience, let $\T$ denote the currently estimated support of SPR, over which the estimation operation is performed. In particular, $\T$ may be $S^0$ or $S^1$ in our analysis. Also, let $\mathbb C^\T$ denote the subspace $\left\{\z \in \C^n | \operatorname{supp}(\z) \subseteq \T \right\}$.  

When $m \rightarrow  \infty$, the optimization problem~\eqref{eq:subproblem1} can be rewritten as 
\begin{equation} \label{eq:subproblemExp}
\underset{\mathbf z: \operatorname{supp}(\mathbf z) = \T}{\arg \min}  \mathbb{E}[f(\z)].
\end{equation}
Since $\z\in \C^\T$, finding out the solution to~\eqref{eq:subproblemExp} is equivalent to solving 
\begin{equation} \label{eq:zeropoints}
\nabla_1 \mathbb{E}[f(\z)]_\T= \mathbf{0}.
\end{equation}
{Here $\nabla_1 \mathbb{E}[f(\z)]_\T$ is an $n$-dimensional $k$-sparse vector supported on $\T$, which keeps the entries of $\nabla_1 \mathbb{E}[f(\z)]$ indexed by $\T$ while setting the others to $0$.}

The following proposition characterizes the geometric property of the zero points of $\nabla_1 \mathbb{E}[f(\z)]_\T$. Its proof is left to Appendix~\ref{app:prop1}. 
\begin{prop}
\label{prop:expectedGeo}
If $\T \cap \operatorname{supp}(\x) \neq \emptyset$, then the zero points of $ \nabla_1 \mathbb{E}[f(\z)]_\T$ belong to the following three classes:
\begin{enumerate}[i)]
\item $\z = \mathbf{0}$,

\vspace{2mm}
\item $\z \in  \big \{\omega_\T \x_{\T}: \omega_\T \in \mathbb{C}, |\omega_\T| = \sqrt{ \frac{  \|\x\|^2 + \|\x_{\T}\|^2  }  {2\|\x_{\T}\|^2  } }\big \}$, 

\vspace{2mm}
\item $\z \in \mathcal S \doteq \big \{\z\in \mathbb{C}^{\T}: \x^* \z = 0, \|\z\| = \frac{  \|\x\|  } {\sqrt2  }  \big \}$,
\end{enumerate} 
which are the  local maximum, local minimum, and saddle points of the function $\mathbb{E}[f(\z)]$ on the subspace $\mathbb{C}^{\T}$, respectively. Also, the local maximum and each saddle point have at least one negative curvature.
\end{prop}

One can interpret from this proposition that when $\T\cap \supp(\x) \neq \emptyset$, the subproblem in~\eqref{eq:subproblem1} has a benign geometric property in the expectation sense. In particular, the minimum value of the objective function $\mathbb{E}[f(\z)]$ can only be attained at the local minimum $\omega_\T \x_{\T}$ (which is in closed-form). That is, $\omega_\T \x_{\T}$ must be the global minimum. 

Now we are ready to present the proof of Theorem \ref{thm:2}.

\vspace{2mm}
\noindent {\bf Proof of Theorem~\ref{thm:2}}
\begin{proof}
Since ${S}^{0} \cap \supp(\x) \not = \emptyset$, applying Proposition~\ref{prop:expectedGeo} with $\T = {S}^{0}$ yields
\begin{equation} \label{eq:xi}
\x^0 = \omega_{{S}^{0}} \x_{{S}^{0}}, 
\end{equation} which is the solution of the estimation operation (i.e., the solution to~\eqref{eq:subproblemExp}). Thus,
\begin{equation}
\supp(\x^0) = \supp(\x_{{S}^{0}}) =  {S}^{0} \cap \supp(\x),
\end{equation}
which has no more than $k$ indices. 

Then, the matching operation computes
\begin{equation} \label{eq:nabla1exp}
\nabla_1 \mathbb{E}[f(\x^0)] \overset{\eqref{eq:Expg}}{=} \omega_{{S}^{0}}\|\x_{{S}^{0}}\|^2 (\x_{{S}^{0}} - \x) 
\end{equation}
and adds $k$ indices to $S^0$ corresponding to $k$ most significant entries in $\nabla_1 \mathbb{E}[f(\x^0)]$. 
Note that the term $\omega_{{S}^{0}}\|\x_{{S}^{0}}\|^2$ on the right-hand side of~\eqref{eq:nabla1exp} is always a nonzero scalar. Thus,
\begin{eqnarray}
\label{matchingOP}
\supp {( \nabla_1 \mathbb{E}[f(\x^0)] )}  &=&  \supp {(\x_{{S}^{0}} - \x)} \nonumber \\ &=& \supp(\x)\backslash {{S}^{0}}, 
\end{eqnarray}
which is exactly the set of remaining elements of $\supp(\x)$ that haven't been selected before. Therefore, the added index set in the matching operation must include $\supp(x)\backslash {{S}^{0}}$. In other words, 
\begin{equation} \label{eq:cfk}
\mathcal C(\nabla_1 \mathbb{E}[f(\x^{0})], k) \supseteq \supp(\x)\backslash {{S}^{0}}. 
\end{equation}

In summary,
\begin{eqnarray}
{{S}^{1}} &=& { {S}^{0}} \cup \mathcal C(\nabla_1 \mathbb{E}[f(\x^{0})], k) \nonumber \\
&\overset{\eqref{eq:cfk}}{\supseteq}& { {S}^{0}}\cup \left(\supp(\x)\backslash {{S}^{0}}\right) \nonumber \\
&=& \supp(\x).
\end{eqnarray}
Again, by applying Proposition~\ref{prop:expectedGeo} with $
\T = {{S}^{1}} = \supp(\x),
$ we have $|\omega_{{S}^{1}}|=1$ and  
\begin{equation}
\x^{1} = \omega_{S^{1}} \x_{{S}^{1}} = \x \mathrm{e}^{\mathsf{j}\phi}
\end{equation}
for some $\phi\in[0,2\pi)$. Therefore, the signal recovery is exact in the expectation case.  
\end{proof}

\subsection{Proof of Theorem \ref{thm:3}} \label{sec:proofofthm3}

Theorem~\ref{thm:3} characterizes the behavior of SPR in the non-initial step, which consists of three major operations: i) matching, ii) estimation and iii) pruning. Before proving our theorem, we shall first analyze the accuracy of these operations, respectively. 

\subsubsection{Estimation}
Proposition~\ref{prop:expectedGeo} offers a favorable geometric property of the subproblem~\eqref{eq:subproblemExp} of signal estimation in the expectation case (i.e., when $m \rightarrow \infty$). Notably, this property suggests a closed-form solution $\omega_\T \x_{\T}$, which plays a vital role in the proof of Theorem~\ref{thm:2}. 
Inspired by this, we speculate if the same condition $\T \cap \operatorname{supp}(\x) \neq \emptyset$ (i.e., catching at least one support index) could also ensure a favorable geometric property with the finite amount of samples. 

In the finite sampling case, the solution to the subproblem~\eqref{eq:subproblem1} also belongs to the subspace $\mathbb C^\T$. However, a condition guaranteeing a favorable geometric property for this case is more demanding. In particular, $\T \cap \operatorname{supp}(\x) \neq \emptyset$ may not be enough. Thus, we consider a stronger condition $\frac{\|\x_\T\|^2}{\|\x\|^2} > \frac{9}{10}$ that represents catching sufficient energy of $\x$, as appeared in Theorem~\ref{thm:1}.

In the following proposition, we show that $\frac{\|\x_\T\|^2}{\|\x\|^2} > \frac{9}{10}$ indeed implies a benign geometric property for finite samples. Specifically, the local minima of $f(\z)$ are clustered around the expected global minimum
$\omega_\T \x_{\T}$. Moreover, the saddle points and the maximizers far away from $\omega_\T \x_{\T}$ possess at least one negative curvature.

\begin{prop} [Geometric Property in $\C^\T$] \label{prop:geoproOnsub}
{For any constant $\epsilon>0$, there exist positive absolute constants $C,c$ such that if
\begin{equation}
	\frac{\|\x_\T\|^2}{\|\x\|^2} > \frac{9}{10},
\end{equation}
all local minimizers of $f(\z)$ on the subspace $\C^\T$ locate in the area 
\begin{equation}\label{Eq:areaepsilon}
	\{\z\in\C^\T|\operatorname{dist}(\z, \omega_\T \x_{\T})\leq \epsilon\|\x\|\}
\end{equation}
with probability exceeding $1 - cm^{-1}$ when $m \geq Ck \log^3 k$. Also, the saddle points and maximizers out of this area possess at least one negative curvature.
Furthermore, if
\begin{equation}
	\frac{\|\x_\T\|^2}{\|\x\|^2} = 1,
\end{equation}
then $m \geq Ck \log^3 k$ ensures that all local minimizers of $f(\z)$ on $\C^\T$ are exactly $\x$ up to a global phase with probability exceeding $1 - cm^{-1}$. Also, the saddle points and maximizers possess at least one negative curvature.}
\end{prop}

Recently, an interesting work by Li {\it et al.}~\cite{newlossOn} has suggested one way to refine the geometric analysis on the subspace $\mathbb C^\T$. They showed that an analogous geometric property for a solution space in $\mathbb R^n$ can be guaranteed with only $m = \Omega(n)$ Gaussian samples with high probability. This result, however, may not be directly applied to our analysis, since it is based on a different loss function in the real case.

\subsubsection{Matching} 
Now, we proceed to analyze the accuracy of the matching operation. This operation aims to find out the remaining support indices (i.e., to identify $\supp(\x) \backslash S^0$) that haven't been selected in the initialization step. To this end, it adds to $S^0$ a set of $k$ indices corresponding to the most significant $k$ elements in $\nabla_1 f(\x^{0})$. Recall from~\eqref{eq:10daoshu} that
\begin{equation}\label{eq:gradFformatching}
\nabla_1 f(\x^{0})= \frac1m \sum_{i=1}^m \Big ((|\a_i^* \x^{0} |^2 - y_i^2)\a_i\a_i^* \Big ) \x^{0}.
\end{equation}

First, we discuss a trivial case for $\nabla_1 f(\x^{0})_l$: $l\in S^0$. Note that we derive $\x^0$ by solving the subproblem
\begin{equation} 
\mathbf x^{0}\gets    {\arg \min}_{\mathbf z : \operatorname{supp}(\mathbf z) = {S}^{0}}  f(\mathbf z).
\end{equation}
Since $\x^0$ is the local/global minimizer, it implies that it holds for all $l \in {S^0}$ that
\begin{equation}\label{eq:gradF0}
\nabla_1 f(\x^0)_l = 0.
\end{equation}

Then, we consider the indices $l \in (S^0)^c$, which can be divided into two disjoint set: $\left(\operatorname{supp}(\x) \cup S^0\right)^c$ and $ \supp(\x)\backslash S^0$. See Fig.~\ref{fig:setdividing} for an illustration.
\begin{figure}[t]
\centering
\includegraphics[width=.35\textwidth]{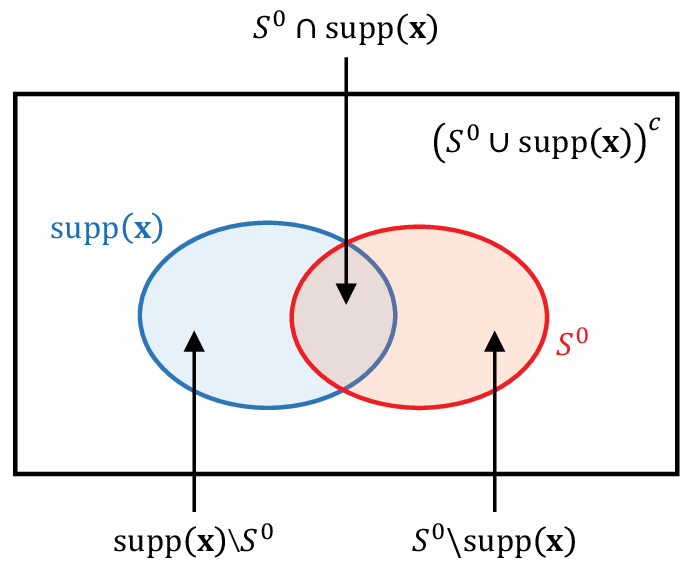}
\caption{An illustration for the separation of $S^0$ and $\supp(\x)$}
\label{fig:setdividing}
\end{figure}
Generally speaking, it is difficult to directly analyze the most significant elements of $\nabla_1 f(\x^{0})$ (i.e.~\eqref{eq:gradFformatching}) due to the random vectors $\a_i$'s. Intuitively, $\nabla_1 f(\x^{0})$ should be close to the expectation $\nabla_1 \mathbb{E} [f(\x^{0})]$ when $m$ is large enough. Here our primary novelty is to connect $\nabla_1 f(\x^{0})$ with $\nabla_1 \mathbb{E}[f(\x^{0})]$, which has a concise expression according to~\eqref{eq:Expg}.

Since $\supp(\x^{0}) = S^0$, we have
$(\x^0)_l = 0$, $\forall l \in\supp(\x)\backslash S^0$, and hence
\begin{align} \label{eq:1231}
\left|\nabla_1 \mathbb{E}[f(\x^{0})]_l \right|& = \Big |2\|\x^{0}\|^2(\x^0)_l-  \|\x\|^2(\x^0)_l- (\x^*\x^{0}) x_l \Big |\nonumber\\
&=\begin{cases}
0, & \text{if}~l \in\left(\operatorname{supp}(\x) \cup S^0\right)^c, \\ 
|\x^*\x^{0} ||x_l|, & \text{if}~l \in \supp(\x)\backslash S^0. 
\end{cases}  	
\end{align}
Therefore, this analytical solution directly allows to distinguish whether a chosen index belongs to the true support of $\x$ or not.

The next proposition reveals that the distance between $\nabla_1 f(\x^{0})_l$ and $\nabla_1 \mathbb{E} [f(\x^{0})]_l$ is indeed well bounded for all $l\in\{1,\cdots,n\}$. In particular, it is controlled by a constant $\epsilon$ that can be arbitrarily small. The proof is deferred to Appendix~\ref{sec:proofproposition3}. 

\begin{figure}[t]
\centering
\hspace{6.8mm}\includegraphics[width=0.45\textwidth]{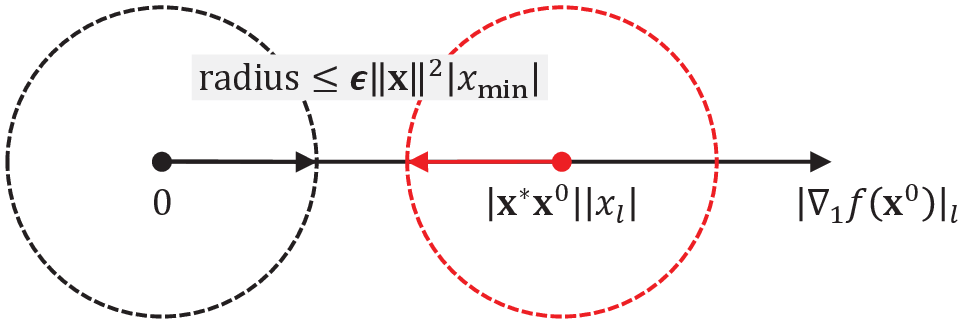}
\caption{An illustrative explanation on Proposition~\ref{lm:elemenEzxtbound}.}
\label{fig:radius}
\end{figure}

\begin{prop} [Concentration of Gradient]\label{lm:elemenEzxtbound}
Suppose $|x_{\min}| = \Omega ({\|\x\|}/{\sqrt{k}})$. For any constant $\epsilon > 0$ and vector $\z \in \mathbb C^n$ satisfying
\begin{equation}
\frac{\|\z\|^2}{\|\x\|^2} \leq \kappa
\end{equation} 
with some constant $\kappa$ it holds with probability exceeding $1-cn^{-1}$ that
\begin{equation}
\label{leq:elemenEzxtbound}
\left|\nabla_1 f(\z)_{l} - \nabla_1 \mathbb{E}[f(\z)]_{l}\right| < \epsilon \|\x\|^2 |x_{\min}|
\end{equation}
for all $l \in \{1, \cdots, n\}$ when 
\begin{equation}
m\geq C\max\{k\log n, \sqrt{k}\log^3n\},
\end{equation}
where $C$ and $c$ depend on $\kappa$ and $\epsilon$.
\end{prop}

Fig.~\ref{fig:radius} gives an illustrative explanation for Proposition~\ref{lm:elemenEzxtbound}. That is, $|\nabla_1 f(\z)_{l}|$ is centered at $|\nabla_1 \mathbb{E}[f(\z)]_{l}|$ with ``radius'' no more than $\epsilon \|\x\|^2 |x_{\min}|$. Combining~\eqref{eq:gradF0},~\eqref{eq:1231}, and~\eqref{leq:elemenEzxtbound}, we have the following cases. 
\begin{enumerate} [~~~~~~~~~~~~]
\renewcommand{\theenumi}{\arabic{enumi}}
\renewcommand{\labelenumi}{Case \theenumi:}
\item\label{case1} For $l \in \supp(\x)\backslash S^0$, the center point of $\nabla_1 f(\x^0)_l $ is $|\x^*\x^0| |x_l|$.

\item\label{case2}  For $l\in \left(\operatorname{supp}(\x) \cup S^0\right)^c$, the center point of $\nabla_1 f(\x^0)_l $ is $0$. 

\item\label{case3}  For $l\in S^0$, $	\nabla_1 f(\x^0)_l = 0$.
\end{enumerate}
It is trivial to see that $\left(\supp(\x)\backslash S^0\right) \cup\left( \left(\operatorname{supp}(\x) \cup S^0\right)^c\right)\cup S^0 = \{1,2,\cdots,n\}$, and the three sets are disjoint. Thus the three cases covers all $l\in  \{1,2,\cdots,n\}$.

Since the term $|\x^*\x^0| \neq 0$ (see~\eqref{eq:60comparable}), the center points $|\x^*\x^0| |x_l|$ (where $l \in \supp(\x)\backslash S^0$) and $0$ are always separated. Therefore, in order to distinguish between the cases of $l \in \supp(\x)\backslash S^0$ (Case \ref{case1}) and $l \not\in \supp(\x)\backslash S^0$ (Case \ref{case2} \& Case~\ref{case3}), i.e., whether index $l$ belongs to the remaining support, it essentially requires that the two dashed circles in Fig.~\ref{fig:radius} do not overlap. In fact, this can be readily guaranteed by properly controlling $\epsilon$ in the ``radius'' $\epsilon \|\x\|^2 |x_{\min}|$. This explains the main idea for proving Theorem~\ref{thm:3}.

\subsubsection{Pruning} 
The non-initial step of SPR involves a third operation called pruning, which narrows down the number of candidates in the estimated support set to $k$. Since $\x^1$ already achieves exact recovery in theory, its support must equals $\text{supp}(\x)$, which has no more than $k$ elements. Hence, we need not to prune the estimated support set any more.

We now have all ingredients to prove Theorem~\ref{thm:3}, except for a dependence issue caused from the random setting. That is, outputs of different operations of SPR all depend on the random matrix $\mathbf A$. In fact, this will make our analysis complicated whenever the expectation $\nabla_1 \mathbb{E}[f(\z)]$ is involved (e.g., in Proposition~\ref{lm:elemenEzxtbound}). The following example explains why this issue affects our analysis.

Since $\mathbf A$ is used to estimate $\x^0$, $\x^0$ is a random vector that depends on $\mathbf A$. When applying Proposition~\ref{lm:elemenEzxtbound} to analyze the matching operation, therefore, it is incorrect to directly plug $\z = \x^0$ into the analytical expression of $\nabla_1 \mathbb{E} [f( \z)]$ in~\eqref{eq:Expg}.
Instead, it would require to re-deduce $\nabla_1 \mathbb{E}[f(\x^{0})]$ by taking into account the dependency. In particular, the expectation should also be taken for $\x^0$ with respect to $\mathbf A$. However, this can be difficult due to the iterative nature of the algorithm.

Here, we tentatively ``ignore'' the dependency on the sampling vectors among different operations of SPR and present the proof of Theorem~\ref{thm:3}. This dependence issue will be formally addressed in the proof of Theorem~\ref{thm:4} via a simple partition strategy. 
 
\vspace{2mm}
\noindent {\bf Proof of Theorem~\ref{thm:3}}
\begin{proof}  
As mentioned, our idea for proving this theorem is to apply Proposition~\ref{lm:elemenEzxtbound} to distinguish whether a newly chosen index belongs to the set of remaining support indices. To apply this proposition, we need to establish a relation between $\|\x^0\|$ and $\|\x\|$. 

Since $S^0$ satisfies $\frac{\|\x_{S^0}\|^2}{\|\x\|^2} > \frac{9}{10}$, applying Proposition~\ref{prop:geoproOnsub} with $\mathcal T = S^0$ yields that
\begin{equation} \label{eq:x0good}
\operatorname{dist}(\x^{0}, \omega_{S^0} \x_{S^0}) < \epsilon \|\x\| 
\end{equation}
holds with probability exceeding $1 - cm^{-1}$ when $m \geq Ck \log^3 k$. 

Denote
\begin{equation}
\x^{0} = \omega_S^0 \x_{S^0} \mathrm{e}^{\mathsf{j}\phi(\x^{0})} +t\g
\end{equation} with $t\in[0, \epsilon\|\x\|)$, $\supp(\g) = S^0$ and $\|\g\| =1$. Since 
\begin{equation}
|\x^*\x^{0}| = |\x^*\omega_{S^0} \x_{S^0} \mathrm{e}^{\mathsf{j}\phi(\x^{0})} + t\x^*\g|,
\end{equation} we have
\begin{align} \label{eq:60comparable}
|\x^*\x^{0}|& \geq |\x^*\omega_S^0 \x_{S^0}|-t|\x^*\g| \nonumber\\
& \geq|\omega_{S^0}| \|\x_{S^0}\|^2 - \epsilon\|\x\|\|\x_{S^0}\| \nonumber\\	
&=\sqrt{\frac{\|\x\|^2+\|\x_{S^0}\|^2}{2}} \|\x_{S^0}\| - \epsilon\|\x\|\|\x_{S^0}\|\nonumber\\
&> \left (\frac{9}{10} - \epsilon \right ) \|\x\|^2.
\end{align}

Moreover, from~\eqref{eq:x0good} we know that
\begin{equation}\label{Eq:controlnormx01}
\|\x^0\|\leq \|\omega_{S^0} \x_{S^0}\|+\epsilon\|\x\|,
\end{equation}
Since $\frac{\|\x_{S^0}\|^2}{\|\x\|^2}\leq 1$ and $|\omega_{S^0}| = \sqrt{ \frac{  \|\x\|^2 + \|\x_{S^0}\|^2  }  {2\|\x_{S^0}\|^2  } }$,~\eqref{Eq:controlnormx01} implies that
\begin{equation}
\|\x^0\|\leq (1+\epsilon)\|\x\|.
\end{equation} 
Therefore, by taking $\epsilon<0.1$, we have
\begin{equation}
\frac{\|\x^{0}\|^2}{\|\x\|^2} \leq \frac{121}{100}.
\end{equation}

Then, we proceed to distinguish between indices $j\in \supp(\x)\backslash S^0$ and $j'\not\in \supp(\x)\backslash S^0$. Note that \begin{equation}\label{leq:Efz-fz}
\Big | |\nabla_1 f(\z)_{l}| - |\nabla_1 \mathbb{E}[f(\z)]_{l}| \Big | \leq \Big |\nabla_1 f(\z)_{l} - \nabla_1 \mathbb{E}[f(\z)]_{l} \Big |.
\end{equation}

\begin{itemize}
\item  For all $j\in\supp(\x)\backslash S^0$,  by Proposition~\ref{lm:elemenEzxtbound}, we have that 
\begin{align}
	|\nabla_1 f(\x^{0})_{j}| &\overset{\eqref{leq:elemenEzxtbound},\eqref{leq:Efz-fz}}{\geq} |\nabla_1 \E [f(\x^{0})]_{j}|-\epsilon\|\x\|^2 |x_{\min}|\nonumber \\
	&\overset{\eqref{eq:1231},\eqref{eq:60comparable}}{\geq} \left (\frac{9}{10}- \epsilon \right ) \|\x\|^2|x_j|-\epsilon\|\x\|^2 |x_{\min}|\nonumber \\
	&	~~\geq~~ \left (\frac{9}{10} -2\epsilon \right ) \|\x\|^2|x_\text{min}| \label{eq:lowerboundnabla}
\end{align}
with probability exceeding $1 - cn^{-1}$ when 
\begin{equation} \label{eq:klognmaxother}
	m\geq C\max\{k\log n, \sqrt{k}\log^3n\}.
\end{equation}

\item For $j'\in\left(\operatorname{supp}(\x) \cup S^0\right)^c$,~\eqref{eq:1231} suggests that $$|\nabla_1 \mathbb{E}[f(\x^{0})_{j'}] | = 0.$$ Plugging this into~\eqref{leq:elemenEzxtbound} and~\eqref{leq:Efz-fz}, we have 
\begin{equation} \label{eq:upperboundnabla}
	|\nabla_1 f(\x^{0})_{j'}| \leq  \epsilon\|\x\|^2 |x_{\min}| .
\end{equation}

\item For $j'' \in S^0$, we have $|\nabla_1 f(\x^{0})_{j''}|=0$ (see~\eqref{eq:gradF0}). 
\end{itemize}
By relating~\eqref{eq:lowerboundnabla} and~\eqref{eq:upperboundnabla} and taking $\epsilon<0.3$, we obtain
\begin{equation} \label{eq:nablagreater}
|\nabla_1 f(\x^{0})_{j}| > |\nabla_1 f(\x^{0})_{j'}|.
\end{equation}
Moreover, combining~\eqref{eq:gradF0} with~\eqref{eq:nablagreater} yields that
\begin{equation} 
|\nabla_1 f(\x^{0})_{j}| > |\nabla_1 f(\x^{0})_{j'}|\geq 0 = |\nabla_1 f(\x^{0})_{j''}|.
\end{equation}
holds for all $j\in\supp(\x)\backslash S^0$, $j'\in\left(\operatorname{supp}(\x) \cup S^0\right)^c$, and $j'' \in S^0$.

Since $\left(\operatorname{supp}(\x) \cup S^0\right)^c \cup S^0 = \left(\supp(\x)\backslash S^0\right)^c$, we have shown that
\begin{equation} 
|\nabla_1 f(\x^{0})_{j}| > |\nabla_1 f(\x^{0})_{j'}|,
\end{equation}
holds for all $j\in\supp(\x)\backslash S^0$ and $j'\not\in\supp(\x)\backslash S^0$.

Recall that the matching operation selects $k$ indices corresponding to the most significant $k$ elements in $\nabla_1 f(\x^{0})$. Thus,~\eqref{eq:nablagreater} implies that all the remaining support indices will be chosen in this operation. In other words, 
\begin{equation}
\mathcal C(\nabla_1 f(\x^{0}), k) \supseteq \supp(\x)\backslash S^0. 
\end{equation} 
Therefore, 
\begin{equation} \label{eq:s1good}
S^1 = S^0\cup \mathcal C(\nabla_1 f(\x^{0}), k) \supseteq  \supp(\x).
\end{equation}

Furthermore, from the condition that $|S^0|= k$, we have $|S^1| \leq 2k$. Applying Proposition~\ref{prop:geoproOnsub} again with $\mathcal T = S^1$ yields that
\begin{equation}
\x^{1} = \x  \label{eq:x1good}
\end{equation}
with probability exceeding $1 - cm^{-1}$ when 
\begin{equation}
m \geq Ck \log^3 k. \label{eq:klogk3}
\end{equation}
That is, the original signal $\x$ is exactly recovered. 

The sampling complexities in~\eqref{eq:klognmaxother} and~\eqref{eq:klogk3} are both guaranteed by
\begin{equation}
m \geq  C\max\left\{k\log^3 k, k\log n,\sqrt{k}\log^3 n\right\},
\end{equation}
which, therefore, completes the proof. 
\end{proof}

\begin{figure*}[t]
\centering
\includegraphics[width=.95\textwidth]{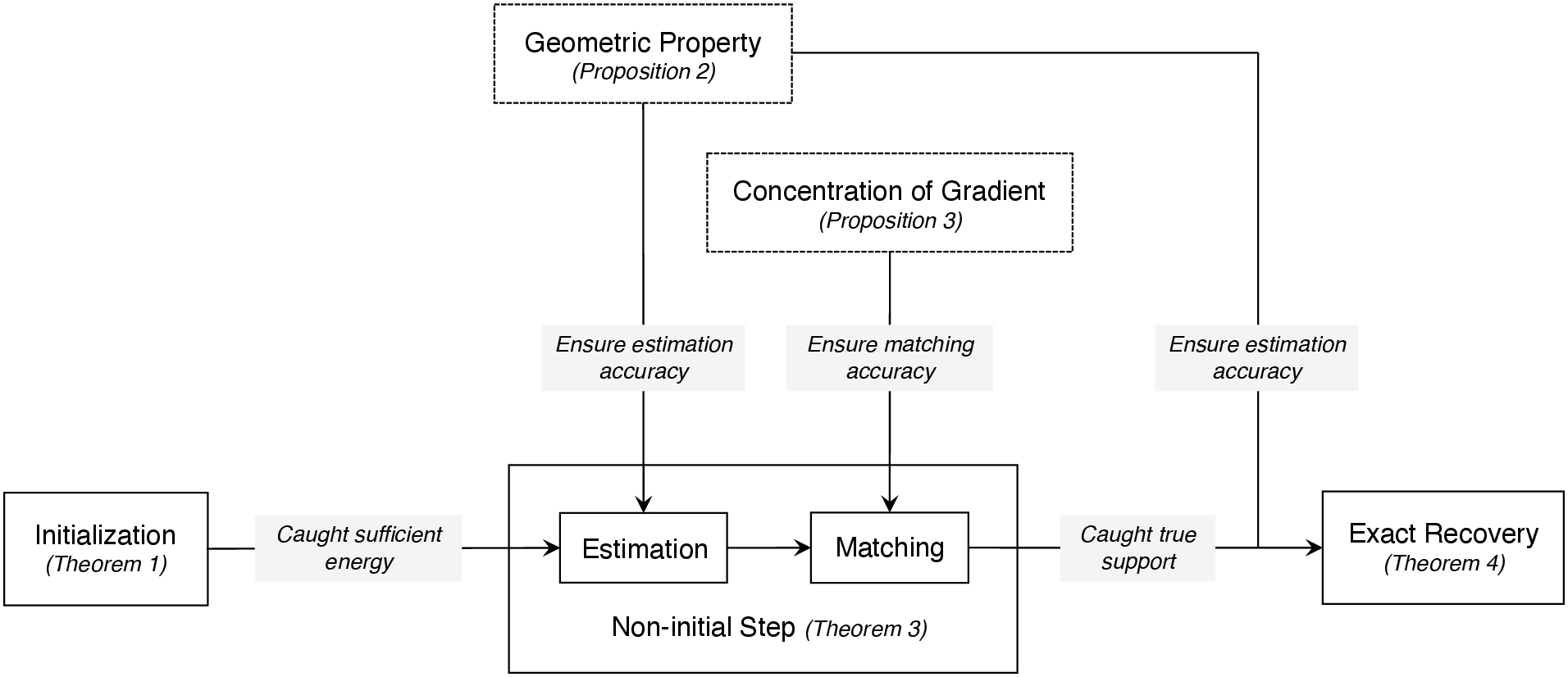}
\caption{Illustrative diagram of our proof structure for exact recovery via SPR.}
\label{fig:proof}
\end{figure*}

\subsection{Proof of Theorem \ref{thm:4}}\label{sec:proofofthm4}
In this proof, we will deal with the dependence issue among different operations of SPR. Our strategy is to  partition the sampling matrix $\mathbf A$ and the corresponding observations $\mathbf y$ into distinct groups and sequentially use each group to perform one operation. As a result, sampling vectors in different groups are statistically independent from each other. This idea is inspired from~\cite{JTROPP,WF,AltMin}. In Fig.~\ref{fig:proof}, we provide an illustrative diagram of our proof structure for exact recovery via SPR.
 
\begin{proof}
We partition $\mathbf A$ and $\mathbf y$ into 4 parts, i.e., \begin{eqnarray}
\mathbf A &=& [\mathbf A_{M_1}^*\quad \mathbf A_{M_2}^* \quad \mathbf A_{M_3}^*\quad \mathbf A_{M_4}^*]^*, \label{eq:partiA}\\
\mathbf y &=& [\mathbf y_{M_1}^\top \quad \mathbf y_{M_2}^\top \quad \mathbf y_{M_3}^\top \quad \mathbf y_{M_4}^\top ]^\top, \label{eq:partiy}
\end{eqnarray}
where \begin{eqnarray} 
|M_1|  &\hspace{-2mm} \geq& \hspace{-2mm} C_1 \bar{s}^2\log (n/\bar{s}), \label{eq:M1} \\ 
|M_2|  &\hspace{-2mm} \geq& \hspace{-2mm} C_2 k\log^3 k, \label{eq:M2}  \\
|M_3|  &\hspace{-2mm} \geq& \hspace{-2mm}  C_3\max \big \{k\log n, \sqrt{k}\log^3n \big \}, \label{eq:M3}  \\
|M_4|  &\hspace{-2mm} \geq& \hspace{-2mm} C_4 k\log^3 k,  \label{eq:M4} 
\end{eqnarray} and \begin{align}
m &= |M_1|+ |M_2|+ |M_3|+ |M_4| \nonumber\\
&\geq   C\max\left\{\bar{s}^2\log (n/\bar{s}), k\log^3 k, k\log n,\sqrt{k}\log^3 n\right\}\label{eq:m=M1_M4}.
\end{align}
Then, we sequentially use a fresh set of sampling vectors  and samples for each operations: 
\begin{eqnarray}
\begin{cases} 
\mathbf A_{M_1} ~\text{and}~\mathbf y_{M_1}~\text{for initializing}~S^0,   \nonumber \\
 \mathbf A_{M_2} ~\text{and}~\mathbf y_{M_2}~\text{for estimating}~\x^0,   \nonumber \\
 \mathbf A_{M_3} ~\text{and}~\mathbf y_{M_3}~\text{for matching}~S^1,   \nonumber \\
 \mathbf A_{M_4} ~\text{and}~\mathbf y_{M_4}~\text{for estimating}~\x^1.   \nonumber    
\end{cases}
\end{eqnarray}
 
In doing so, the output of the previous operation is statistically independent from the groups of submatrices used in the subsequent operations. 
For example, when matching $S^1$ given the previous outputs $\x^0$ and $S^0$, which are functions of $\mathbf A_{M_1}$ and $\mathbf A_{M_2}$, we can make sure that $\x^0$ and $S^0$ are statistically independent from $\mathbf A_{M_3}$. Thus, the expectation $\nabla_1 \mathbb{E}[f(\x^0)]$ in Proposition~\ref{lm:elemenEzxtbound} can be taken just with respect to $\mathbf A_{M_3}$.

Since $\x^1$ already achieves exact recovery by Theorem~\ref{thm:3}, we do not need to analyze more iterations. Thus, there is no need to partition $\mathbf A$ into more groups. For analytical convenience, we rewrite SPR in Algorithm~\ref{alg:SPR_theo} by omitting the subsequent iterations.

Consider the event $\Sigma_{\text{succ}}$ where the SPR algorithm correctly recovers
the input signal $\x$. Also, define the events $\Sigma_{1}, \cdots, \Sigma_{4}$ as follows:
\begin{eqnarray} 
\Sigma_{1} &\overset{\text{def}}{=}& \left \{ \text{Initialized an}~S^0~\text{satisfying}~\frac{ \| \x_{S^0} \|^2}{\| \x \|^2} > \frac{9}{10} \right \}, \nonumber \\ 
\Sigma_{2} &\overset{\text{def}}{=}& \left\{ \text{Estimated a desired}~\x^0~\text{satisfying~\eqref{eq:x0good}} \right \}, \nonumber \\ 
\Sigma_{3} &\overset{\text{def}}{=}& \left\{ \text{Matched a desired}~S^1 \supseteq \supp(\x)\backslash S^0  \right \}, \nonumber \\ 
\Sigma_{4} &\overset{\text{def}}{=}& \left\{ \text{Estimated a desired}~\x^1 = \x \right \},   \nonumber 
\end{eqnarray}
which exactly correspond to the four steps in Algorithm~\ref{alg:SPR_theo}. 
\begin{algorithm}[t!]
\caption{\textsc{{SPR with Partitioning}} }
\label{alg:SPR_theo}
\begin{algorithmic}[1]
\STATE {\textbf{Input}: sparsity $k$, partitioned matrix $\mathbf A$ in~\eqref{eq:partiA}, partitioned observations $\y$ in~\eqref{eq:partiy}.}
\STATE \textbf{Initialize}: ${S}^{0}$ by correlation-promoting method using $k$, $\y_{M_1}^*$ and $\mathbf A_{M_1}^*$. \\~~~~~~~~~~~~~\(\mathbf x^{0} \gets   {\arg \min}_{\mathbf z : \operatorname{supp}(\mathbf z)={S}^{0}} f(\mathbf z; \mathbf A_{M_2})\). 

\STATE {\textbf{Matching}: \({S}^{1} \leftarrow {S}^{0} \cup \mathcal C(\nabla_1 f(\x^{0}; \mathbf A_{M_3}), k)\).}

\STATE {\textbf{Estimate}: \(\mathbf x^{1}\gets   {\arg \min}_{\mathbf z : \operatorname{supp}(\mathbf z)={S}^{1}} f(\mathbf z; \mathbf A_{M_4})\).}

\STATE {\textbf{Output}: estimated signal $\hat{\mathbf x} = \mathbf x^1$.}
\end{algorithmic}
\end{algorithm} 

Then, we can combine the theoretical results of events $\Sigma_{1}$--$\Sigma_{4}$ together. 

Specifically, applying the definition of conditional probability, we have 
\begin{eqnarray}
\mathbb P (\Sigma_{\text{succ}}) 
& \hspace{-2mm} \geq & \hspace{-2mm} \mathbb P (\Sigma_{\text{succ}} \cap \Sigma_{1} \cap \Sigma_{2} \cap \Sigma_{3} \cap \Sigma_{4} ) \nonumber \\
& \hspace{-2mm} = & \hspace{-2mm} \mathbb P (\Sigma_{\text{succ}} | \Sigma_{1},  \Sigma_{2}, \Sigma_{3}, \Sigma_{4} ) \mathbb P (\Sigma_{1} \cap \Sigma_{2} \cap \Sigma_{3} \cap \Sigma_{4} ) \nonumber \\
& \hspace{-2mm} = & \hspace{-2mm} \mathbb P (\Sigma_{1} \cap \Sigma_{2} \cap \Sigma_{3} \cap \Sigma_{4} ) \nonumber \\
& \hspace{-2mm} = & \hspace{-2mm} \mathbb P ( \Sigma_{1} ) \mathbb P (\Sigma_{2} | \Sigma_{1})  \mathbb P (\Sigma_{3} | \Sigma_{2}, \Sigma_{1} )  \mathbb P (\Sigma_{4} | \Sigma_{3}, \Sigma_{2},  \Sigma_{1} ) \nonumber \\
& \hspace{-2mm} \geq & \hspace{-2mm} 1-cm^{-1}, 
\end{eqnarray}
where the last inequality is due to~\eqref{eq:M1}-\eqref{eq:M4}. Thus the proof is completed.
\end{proof}

\section{Discussion}\label{sec:discussion}
In this section, we discuss several issues that arise from our analysis.
 
\subsection{Sampling Complexity}\label{sec:disTobars}

As mentioned, the overall sampling complexity for most exisiting sparse phase retrieval algorithms (e.g.,~\cite{SWF,SPARTA,HWF,HTP,SAM,CoPRAM}) is dominated by the initialization stage.  The so far best result, obtained in~\cite{CJF}, is given by
\begin{equation}\label{eq:CJFsamplingComplexity}
m = \Omega \left( \frac{\|\x\|^2}{|x_{\max}|^2 }k \log n \right)
\end{equation} 
under the probability of $1-m^{-1}$. The sampling complexity of SPR initialization is
\begin{equation}\label{eq:ourcomplexity}
m = \Omega \left ( \bar{s}^2\log \frac{n}{\bar{s}} \right )
\end{equation}
under the probability of 
$
1-\exp \left ( -c\bar{s} \log \frac{n}{\bar{s}}\right ).
$
To compare our result with that of~\cite{CJF}, we consider three typical cases, as specified in Table~\ref{tab:relsmp},

\begin{enumerate}[i)]
\item {\bf Case 1.} $\bar{s} = \Theta (1)$: 

In this case, our result is given by 
\begin{equation}
	m = \Omega(\log n).
\end{equation}
Whereas for~\eqref{eq:CJFsamplingComplexity}, one can easily see that 
\begin{equation}
	\frac{\|\x\|}{|{x_{\max}}|}= \Theta (1)
\end{equation}	
and thus the result of~\cite{CJF} becomes 
\begin{equation}
	m = \Omega\left(k\log n\right),
\end{equation}	 
which is higher than ours.

\vspace{1mm}
\item {\bf Case 2.} $\Theta(1) < \bar{s} < \Theta (\sqrt{k})$:

In this case, one can verify that 
\begin{equation}
	\bar{s}^2\log \frac{n}{\bar{s}} < \Theta \left ( k \log n \right ), 
\end{equation}
which  is always lower than~\eqref{eq:CJFsamplingComplexity} because
\begin{equation}
	\frac{\|\x\|^2}{|x_{\max}|^2 }k \log n = \Omega \left ( k \log n \right ). 
\end{equation} 
For example, when $\bar{s} = \Theta (\sqrt[4]{k})$, our result becomes
\begin{equation}
	m = \Omega ( \sqrt k \log n ).
\end{equation}
 
\vspace{1mm}
\item {\bf Case 3.}  $\bar{s} \geq \Theta({\sqrt k})$:

In this case, it is difficult to directly compare~\eqref{eq:ourcomplexity} with~\eqref{eq:CJFsamplingComplexity}, because ${\|\x\|^2}/{|{x_{\max}}|^2}$ cannot be determined from the condition $\bar{s} = \Omega ({\sqrt k})$, and vice versa. For example, when $\bar{s} = \Theta({k^{3/4}})$, our result in~\eqref{eq:ourcomplexity} becomes 
\begin{equation}\label{eq:ourk3/2logn}
	m =  \Omega \left (k^{3/2}\log n \right ).
\end{equation}
Since it is not possible to determine ${\|\x\|^2}/{|{x_{\max}}|^2}$ from $\bar{s} \geq \Theta({\sqrt k})$, we consider two scenarios: 
\begin{itemize}
	\item 
	When ${\|\x\|^2}/{|{x_{\max}}|^2} = \Theta(k^{3/4})$, i.e., when the energy of the $\bar{s}$ most significant entries of $\x$ is pretty ``even'', the sampling complexity in~\eqref{eq:CJFsamplingComplexity} is 
	\begin{equation}\label{eq:k3/2logn}
		m = \Omega\left(k^{7/4}\log n\right),
	\end{equation} 
	which is higher than our result.
	
	\item It is also possible that ${\|\x\|^2}/{|{x_{\max}}|^2} = \Theta(1)$ given $\bar{s} = \Theta({k^{3/4}})$, i.e., only $|x_{\max}|$ is on the order of $\|\x\|$ while the others are small. Consequently, the sampling complexity in~\eqref{eq:CJFsamplingComplexity} is just 
	\begin{equation}\label{eq:k3/2logn}
		m = \Omega\left(k \log n\right),
	\end{equation} 
	which is lower than ours.
\end{itemize}

\end{enumerate}
In summary, when $\bar{s} < \Theta (\sqrt{k})$, our sampling complexity for initialization is always better than that in~\cite{CJF}. However, when $\bar{s} \geq \Theta (\sqrt{k})$, neither of them always outperforms the other. Our result can be better in some cases.

\begin{table}[t!]
\renewcommand\arraystretch{1.5}
\begin{center}
	\caption{The relationship between $m$ and $\bar{s}$ in~\eqref{eq:ourcomplexity}} 
	\label{tab:relsmp}
	\begin{tabular}{ l l l}   
		\hline
		$\bar{s}$                  &  $m$                                &  Probability \\
		\hline
		$\Theta(1)$            & $= \Omega(\log n)$            & $\geq 1 - \exp\left(-c \log n\right)$ \\
		$\Theta(\sqrt{k})$ & $= \Omega(k\log n)$          & $\geq 1 - \exp\left(-c \sqrt{k}\log \frac{n}{\sqrt{k}} \right)$ \\
		$\Theta(k)$            & $= \Omega(k^2\log n)$      & $\geq 1 - \exp\left(-c k\log \frac{n}{k} \right)$ \\
		\hline
	\end{tabular} 
\end{center}
\end{table}

Moreover, we compare our probability with that in~\cite{CJF}. Since our probability is larger than $1 - \exp\left(-c \log n\right)$, and also noting that $m<n$ in our setting, we have
 {
\begin{equation}\label{eq:forrebucomment51}
1 - \exp\left(-c \log n\right)  =  1-\mathcal O(n^{-1})
> 1-\mathcal O(m^{-1}).
\end{equation} }
Thus, our probability is always better than that in~\cite{CJF}.

\subsection{Definition of $\bar S$}

As a subset of $\supp(\x)$, $\bar{S}$ is defined with respect to the constant $0.999$, that is, 
\begin{equation}  
\bar{S} = {\arg \min}_{S:  \|\x_S\|^2 \geq 0.999 \|\x\|^2}~|S|.
\end{equation}
Based on this definition, we derive a lower bound 
\begin{equation}
\frac{\|\x_\T\|^2}{\|\x\|^2} \geq \frac{9}{10}
\end{equation}
in Theorem~\ref{thm:1}, which allows to prove a useful geometric property in Proposition~\ref{prop:geoproOnsub}. 
Our main purpose of using an absolute constant $0.999$, rather than a parameter, to define $\bar S$ is to simplify the proof of Proposition~\ref{prop:geoproOnsub}. In fact, if we use a parameter to define $\bar S$, the related analysis and results (especially the proof of Proposition~\ref{prop:geoproOnsub}) could be much more complicated.

At a first glance, the constant $0.999$ seems to impose a restrictive constraint on the energy distribution of $\x$. 
We stress that $0.999$ is not particularly given, but just for analytical convenience. If other smaller constants (e.g., $0.5$, $0.1$, or even smaller) were used to define the subset $\bar S$, then the lower bound for ${\|\x_\T\|^2}/{\|\x\|^2}$ would become smaller, accordingly.  
However, we do not pursue optimizing this constant in our paper. As long as it is in $(0, 1)$, our analysis may still hold. That is, it still allows to prove a geometric results similar to Proposition~\ref{prop:geoproOnsub}, thus leading to the same sampling complexity in Theorem~\ref{thm:3}.

\subsection{Assumption on $|x_{\min}|$}
\label{sec:DisonMin}
 It is worth noting that the assumption $|x_{\min}| = \Omega ({\|\x\|}/{\sqrt{k}})$ is necessary for our analysis (Proposition~\ref{lm:elemenEzxtbound}). This can be a drawback of our method compared to some existing approaches that do not rely on this assumption, such as CoPRAM~\cite{CoPRAM}, HTP~\cite{HTP} and SAM~\cite{SAM}.

Typically, those approaches require an initialization that falls into a $\delta$-neighbour of the gound truth $\x$, such as the one proposed in~\cite{CJF}, and then iteratively refine the signal estimate via, e.g., descent methods. As they only require a $\delta$-neighbour initialization, it does not matter if the index of $|x_{\min}|$ is selected or not. Thus, they do not require the assumption on $|x_{\min}|$ in their analysis. For their non-initial stage, a convergence analysis is commonly adopted to characterize the recovery error, which also does not need to bound the value of $|x_{\min}|$.

We intuitively explain the reason why assuming a low bound on $|x_{\min}|$ is necessary for proving Proposition~\ref{lm:elemenEzxtbound}. In essence, this is because SPR uses a ``matching'' operation to identify all support indices of $\x$ (see Algorithm~\ref{alg:SPR}). If some nonzero elements were extremely small, they would have little chance to be selected via ``matching''. Similar situation happens to~\cite{SPARTA,SWF,HWF}, where their initialization essentially implies an accurate support recovery and hence the constraint on ${|x_{\min}|}$ is unavoidable.

Actually, if there were no assumption on $|x_{\min}|$, then our algorithm could identify all those indices of nonzero elements $x_j$ satisfying 
\begin{equation}
|x_j| = \Omega \left (\frac{\|\x\|}{\sqrt k} \right ).
\end{equation}
In this case, by following some techniques developed in~\cite{CJF}, it can be possible to show that our algorithm also provides a good estimation $\hat{\x}$ within a $\delta$-neighbour of $\x$.

\subsection{Relationship between conditions on $|x_{\min}|$ and $\bar s$}
In essence, the conditions on $\bar s$ and $|x_{\min}|$, respectively, assume certain structures of the input signal $\x$:
\begin{itemize}
	\item The lower bound on $|x_{\min}|$ imposes a restriction on the smallest entries of $\x$. In general, $|x_{\min}| =\Omega\left(\frac{\|\x\|}{\sqrt{k}}\right)$ may imply {``flat''} signals, but not always. Here, the ``flat'' signal means $|x_{i}| =\Omega\left(\frac{\|\x\|}{\sqrt{k}}\right)$, $\forall i \in \supp(\x)$. 
 \item The parameter $\bar{s}$ characterizes the behavior of the largest entries of $\x$, akin to the condition on $\frac{|x_{\max}|}{\|\x\|}$ appeared in~\cite{HWF, CJF}.  In particular, when $\bar{s}$ is relatively small (e.g., $\bar{s}=1$),  some large entries already occupies the most energy of $\x$, thus indicating ``non-flat'' signals. 
\end{itemize} 
We use an example to illustrate this. Consider a $k$-sparse signal $\x$ with
\begin{equation}
	x_i = \begin{cases}
		\sqrt{0.999},&i=1\\
		\sqrt{\frac{0.001}{k-1}} ,& i = 2,3,\cdots, k,\\
		0,  &  i = k+1, k+2, \cdots, n.
	\end{cases}
\end{equation}
Then, one can easily compute that 
\begin{align}
\begin{cases} 
\bar{S} = {\arg \min}_{S:  \|\x_S\|^2 \geq 0.999 \|\x\|^2}~|S| = \{1\} & \\
\bar{s} = |\bar{S}| = 1   & 
	\end{cases},
\end{align}
which seems to imply an extremely ``non-flat'' signal. Nevertheless,
the lower bound $|x_{\min}| = \Omega\left(\frac{\|\x\|}{\sqrt{k}}\right)$ still holds because
\begin{align}
\begin{cases} 
\|\x\| = 1, &  \\
|x_{\min}| = \sqrt{\frac{0.001}{k-1}} = \Omega\left(\frac{\|\x\|}{\sqrt{k}}\right) & 
	\end{cases}.
\end{align}
\begin{figure}[t]
	\centering
	\includegraphics[width=0.33 \textwidth]{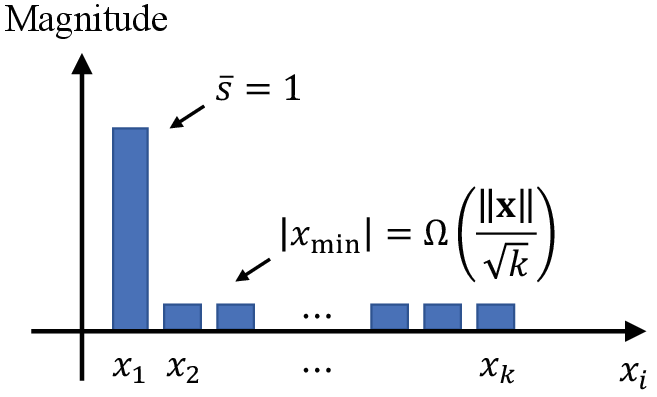}
	\caption{An illustrative example where $\bar s$ and $|x_{\min}|$ affect different parts of $\x$.}
	\label{fig:xsignal}
\end{figure}

In fact, the conditions on $\bar s$ and $|x_{\min}|$ impose restrictions on different parts of $\x$ (see  Fig.~\ref{fig:xsignal}); thus, they do not contradict each other and can hold simultaneously.

\subsection{The Computational-to-statistical Gap}\label{dis:compustatgap}
In~\cite{PRbyCSPhaseLift}, the author proposed an interesting approach to achieve the information-theoretic bound for sparse phase retrieval. Specifically, suppose in~\eqref{eq:yAxmatrixPR} that $\mathbf A$ can be decomposed into the product of matrices $\mathbf B\in\mathbb R^{m\times d}$ and $\mathbf C \in \mathbb R^{d\times n}$, where $d = \Omega (k \log (n/k))$. $\mathbf B$ allows for phase retrieval (using algorithm like PhaseLift~\cite{PhaseliftOn}), whose sampling complexity is $m = \Omega(d)$, and $\mathbf C$ allows for compressed sensing (using algorithms like  CoSaMP~\cite{Cosamp}), whose sampling complexity is $d = \Omega(k\log (n/k))$. Then, it is feasible to reconstruct an $n$-dimensional $k$-sparse complex-valued signal within $\Omega (k \log (n/k))$ measurements, thus bridging the computational-to-statistical gap. However, such measurement matrix $\mathbf A$ is not a Gaussian random matrix.

The analysis on the proposed SPR algorithm is based on Gaussian random measurements. However, we have to consider a union bound of size $\mathcal O({n\choose \bar{s}})$ in~\eqref{eq:unionboundinInit}, which leads to a sampling complexity of $\Omega (\bar{s}^2\log n)$ in the initialization stage. In the case where $|x_j| = \Theta({\|\x\|}/{\sqrt k})$ for $j\in \supp(\x)$, i.e., when the input signal is ``flat'', the required sampling complexity is still
\begin{equation}
   m = \Omega(k^2 \log n),
\end{equation}
which is significantly higher than the information-theoretic bound $\Omega(k\log n)$.

In fact, recovery of ``flat'' signals also represents a challenging case for many existing non-convex phase retrieval approaches (e.g.,~\cite{SWF, SPARTA, HWF, HTP, SAM, CoPRAM}), whose sampling complexity remains $\Omega(k^2\log n)$ for this case. 
The same story happens to many greedy algorithms for compressed sensing~\cite{OMP1,OMP2,davenport2010analysis}, which can better recover signals that have randomly distributed magnitudes of nonzero entries or exhibit a strong decay, compared to recovering ``flat'' signals (e.g., 0-1 signals).

This computational-to-statistical gap mainly results from the spectral method in initialization. More precisely, it arises due to recovering ``flat'' signals, for which it is challenging to identify the maximum nonzero entry that is not so significant to dominate the whole energy of signal. To bridge this gap, it would require to develop a new spectrum that is more capable of catching the maximum nonzero entry of input signal, or a refined analysis for initialization, especially on the concentration property of the spectrum. To date, whether the gap can be closed remains an interesting open question.

\subsection{Geometric Analysis}\label{sec:disTogeoana}
Our geometrical analysis owns a lot to~\cite{sunju}, where an analogous geometric property for the case where $\x$ is nonsparse was studied under the condition of $$m \geq C n\log^3 n.$$ 
Our analysis has two major distinctions. 
\begin{itemize}
\item Firstly, we extend the result of~\cite{sunju} to the sparse case. That is, we analyze the geometric property for the subproblem~\eqref{eq:subproblem1}, where the solution space is restricted to a subspace $\mathbb C^\T$, rather than $\mathbb C^n$. In their nonsparse case, the expected local minima and the real minima are just identical, which equals $\x$, as implied in~\cite[Theorem~2]{sunju}. As for the sparse case, however, we could neither derive an analytical solution for~\eqref{eq:subproblem1}, nor claim that it is equal to the expected global minima $\omega_\T \x_{\T}$. This obstacle is detrimental to our geometric analysis.  Our novelty here  is to estimate the gap between the expected local minima and the real ones with concentration techniques. In particular, we show that this gap is well controlled by a constant $\epsilon$ that can be arbitrarily small. This already suffices to demonstrate a promising performance for the non-initial step of SPR.

\item Secondly, our result is more general than that of~\cite{sunju} in the sense that it covers not only the case of ${\|\x_\T\|^2}/{\|\x\|^2} =1$ (i.e., caught all support indices), but also the case when ${\|\x_\T\|^2}/{\|\x\|^2} \in ( \frac{9}{10}, 1)$  (i.e., caught sufficient energy but not all). In fact, obtaining a geometric property for the case where ${\|\x_\T\|^2}/{\|\x\|^2} \in ( \frac{9}{10}, 1)$ is of vital importance for our analysis, as it allows to connect with the initialization step. Recall that we weakened the goal of initialization by letting it capture only sufficient energy. Precisely, this enables to derive an improved sampling complexity for SPR. Note that initializing with ${\|\x_\T\|^2}/{\|\x\|^2} =1$ would just require $m \geq C k^2 \log n$~\cite{SPARTA,SWF}.

~~Besides, the geometric result for the former case (${\|\x_\T\|^2}/{\|\x\|^2} =1$), which can be implied from~\cite{sunju}, is also important to our analysis. Specifically, it guarantees exact recovery of the input signal when SPR has already selected all support indices after some iterations.

\end{itemize}

\subsection{Partitioning}\label{sec:distoPartitioning}
As a commonly used strategy for eliminating the dependence issue, the partitioning used in Algorithm~\ref{alg:SPR_theo} brings considerable convenience to our analysis. We only partition $\y$ and $\mathbf{A}$ into four parts, because our analysis is just for one iteration of SPR, where the dependence issue only occurs to four variables. However, if our algorithm were analyzed for running more iterations, say, $\Theta(\log k)$ iterations, then we would need to partition $\y$ and $\mathbf{A}$ into $\Theta(\log k)$ parts for our analysis. In this case, the sampling complexity would increase $\Theta(\log k)$ times. Similar examples can be found in~\cite{WF,AltMin}.

Despite the analytical benefits, the partitioning strategy is of little practical value. Indeed, we have empirically confirmed that using the entire $\mathbf A$ for each operation leads to better performance. This is because the partition operation essentially reduces the number of samples for signal reconstruction.

\subsection{Number of Iterations} \label{eq:numberofiteration}
The SPR algorithm has a computational advantage in theory. Specifically, it can achieve an accurate support recovery using only one iteration (Theorem~\ref{thm:3}). In comparison, the CoSaMP and SP algorithm in the compressed sensing literature, from which our algorithm was inspired, need $\Omega(k)$ iterations to converge in theory. Phase retrieval approaches such as CoPRAM, HTP and SAM also require a number of iterations that is proportional to $k$, or $\log k$.

In practice, however, $m$ may not be large enough to satisfy the theoretical requirement of sampling complexity. As a result, the recovery performance of SPR may not be so satisfactory in one iteration. Thus, it is often desired for the SPR algorithm to iterate more times.

\section{Experiments}

\label{sec:simulation}
In this section, numerical experiments are carried out to test the performance of SPR for phase retrieval. Our experiments are performed in the MATLAB 2019a environment on a server with an Intel(R) Xeon(R) Silver 4116 CPU and 4 GeForce RTX 2080 Ti GPUs. The code for SPR is available at \url{https://github.com/mengchuxu97/SPR}.

\subsection{Geometric Property}
We empirically verify the geometric property in Proposition \ref{prop:geoproOnsub}. Specifically, we construct a random matrix $\mathbf{A} \in \mathbb{C}^{10,000 \times 20,000}$ with entries drawn {\it i.i.d.} from the standard complex Gaussian distribution. Also, we generate a vector $\x \in \mathbb{C}^{10,000}$ whose first $10$ entries are $1$ and the rest are $0$. Then, we use the BB algorithm~\cite{BBalg} to find out the solution $\hat{\x}$ to~\eqref{eq:subproblem1} over the estimated support $\hat{S}^t = \{1, \cdots, 5, 11, \cdots, 25\}$. In this case, $\hat{S}^t \cap \operatorname{supp}(\x) = \{1, \cdots, 5\} \neq \emptyset$.

\begin{figure}[t]
\centering
\hspace{-2mm}\includegraphics[width=0.45\textwidth, height = 70 mm]{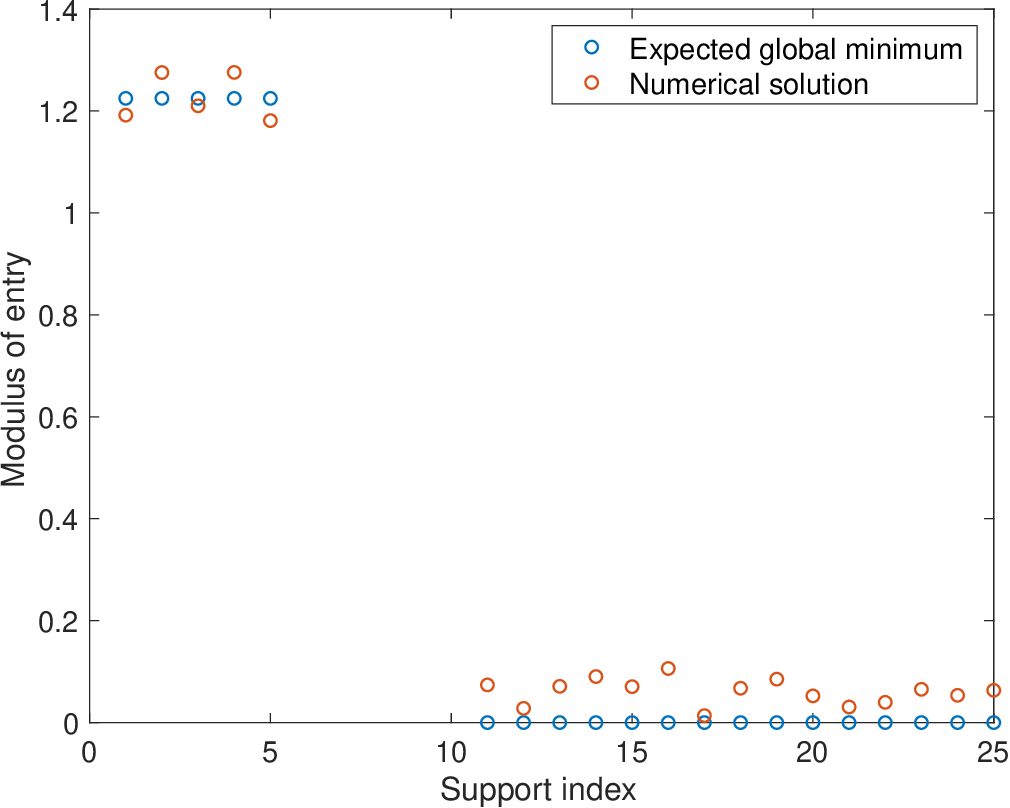}
\caption{The numerical solution is clustered around the expected global minimum.}
\label{fig:expTest}
\end{figure}

In Fig.~\ref{fig:expTest}, teal points represent the expected global optimum (i.e., $\omega_\T\x_\T$), while red points are the numerical solution to~\eqref{eq:subproblem1}. Since the original values of the solution are complex numbers, we plot their modulus for illustration. One can observe that the numerical solution is well clustered around the expected global optimum, which matches the proposed geometric property.

\begin{figure*}[t]%
\centering
\subfloat[Frequency of exact phase retrieval as a function of the number of samples.]{
\includegraphics[width=.3\linewidth]{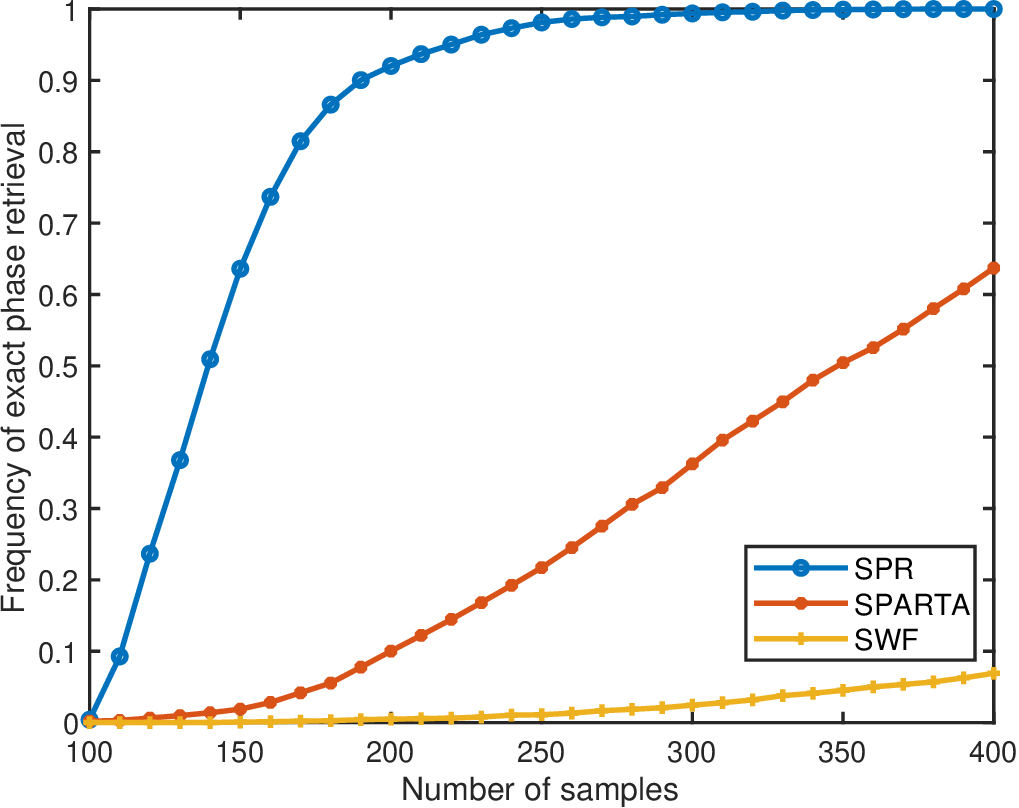}
\label{fig:cc}
}\hfill
\subfloat[Frequency of exact phase retrieval as a function of the sparsity $k$.]{
\includegraphics[width=.3\linewidth]{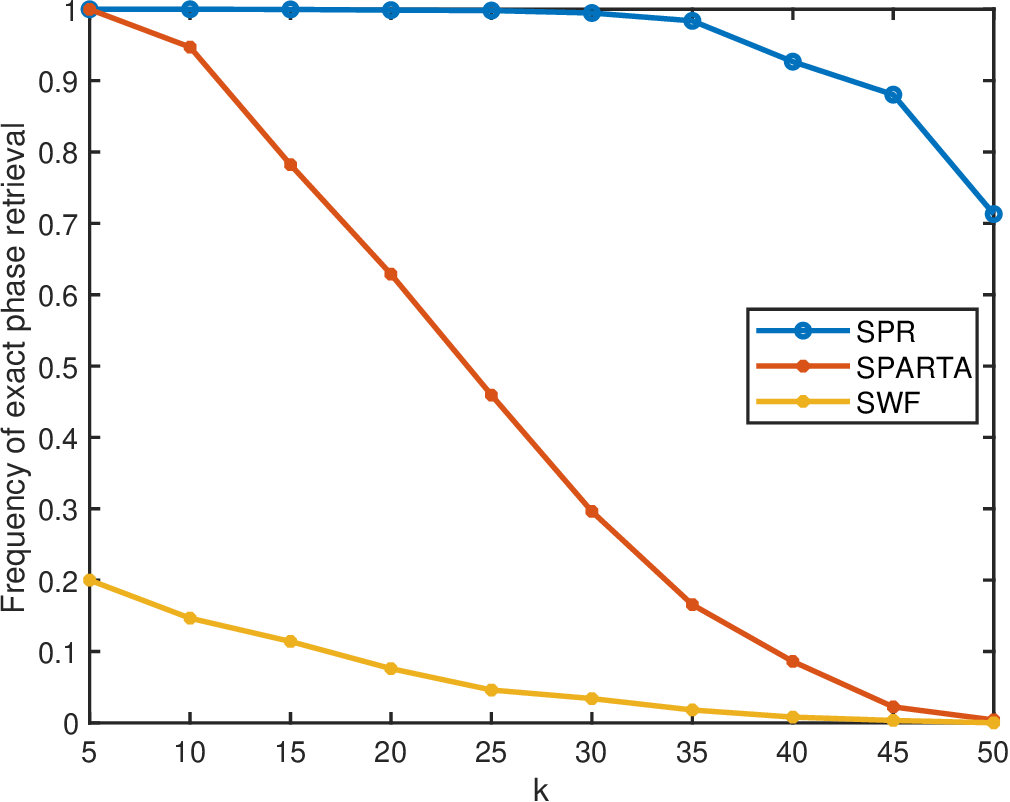}
\label{fig:kranges}
}\hfill
\subfloat[Frequency of exact phase retrieval when $k$ is overestimated.]{
\includegraphics[width=.3\linewidth]{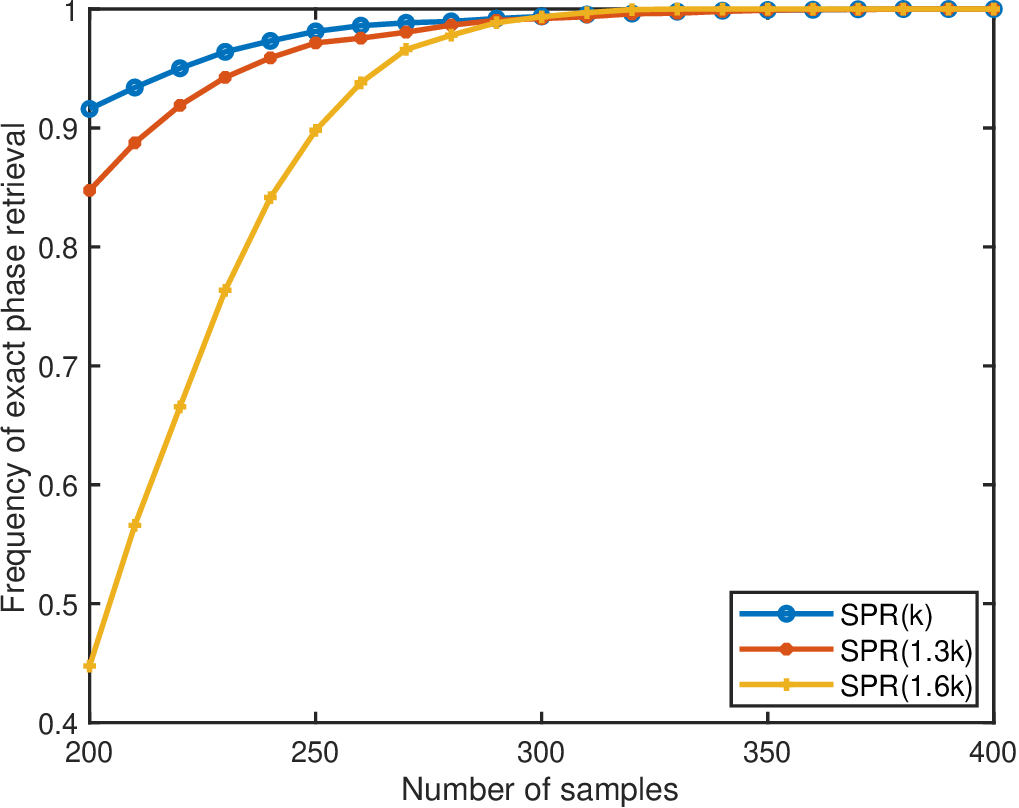}
\label{fig:moreK}
}
\caption{Numerical simulations in noiseless cases}
\end{figure*}

\begin{figure*}[t]%
\centering
\subfloat[Normalized mean squared error as a function of noise (dB) when $m=300,n=1,000,k=10$.]{
	\includegraphics[width=.31\textwidth]{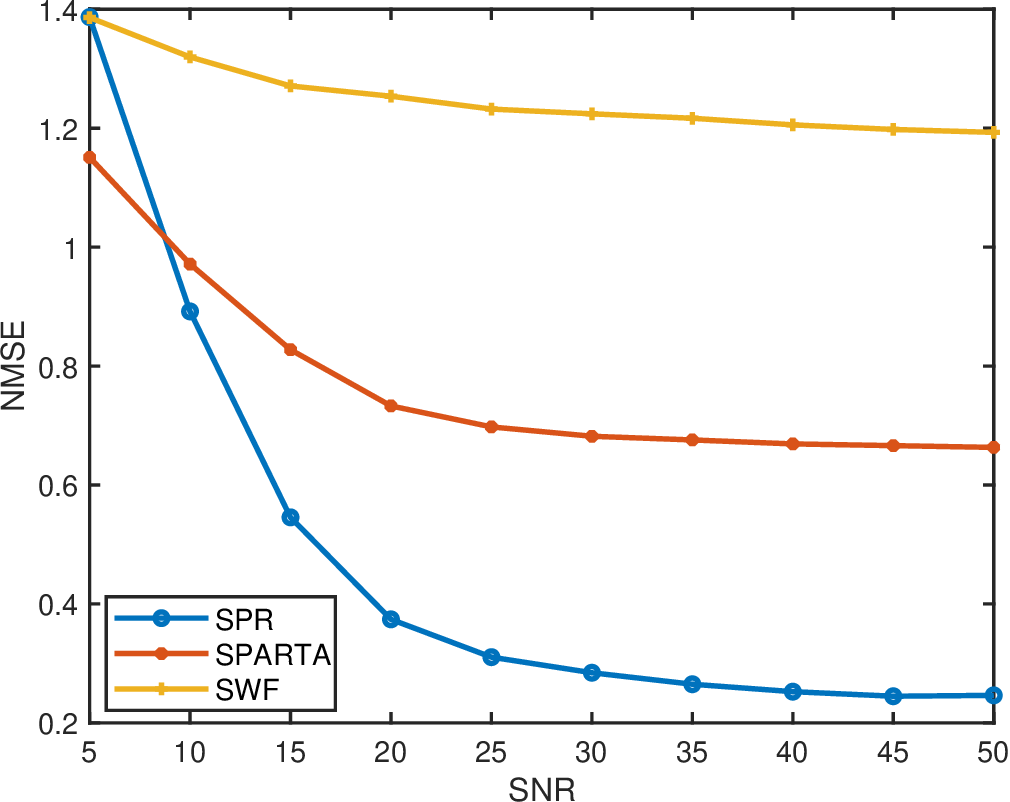}
	\label{fig:noise300}
} \hspace{10mm}
\subfloat[Normalized mean squared error as a function of noise (dB) when $m=800,n=1,000,k=10$.]{
	\includegraphics[width=.31\textwidth]{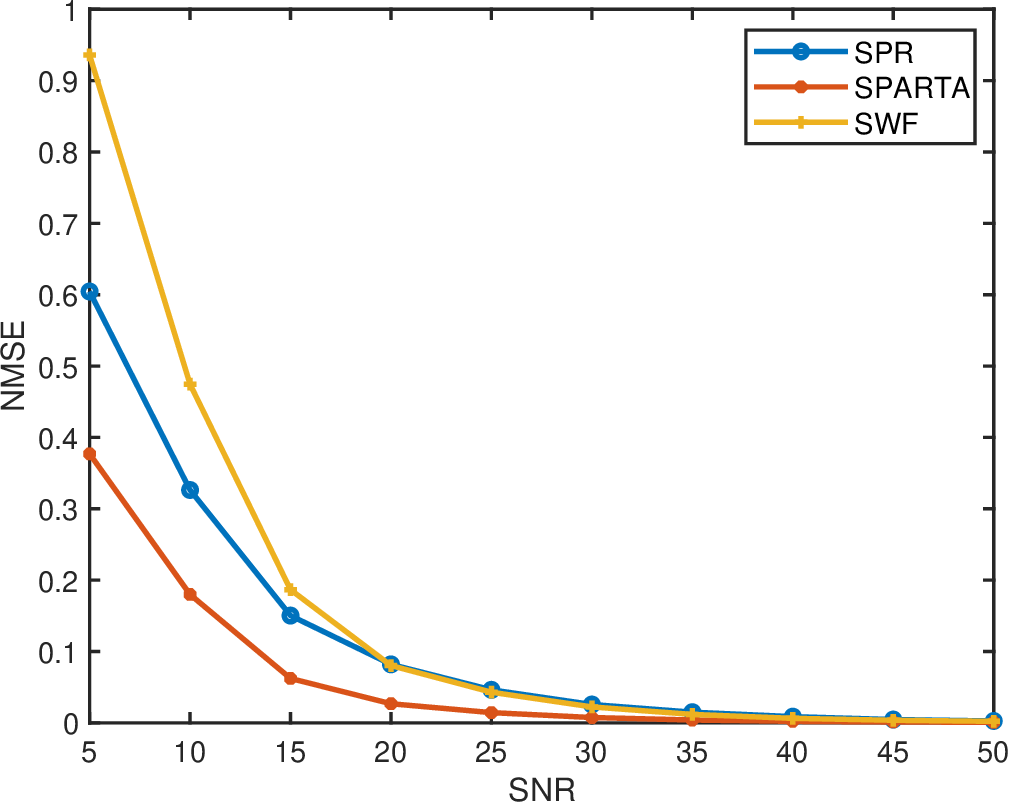}
	\label{fig:noise800}
}
\caption{Numerical simulations in noisy cases}
\end{figure*}

\subsection{Recovery of 1D signals}\label{sec:simOnoneD}
We consider the original signal $\x \in \mathbb C^{1,000}$. Each entry of $\mathbf A$ and the nonzero elements of $\x$ are drawn {\it i.i.d.} from the standard complex Gaussian distribution. For comparative purpose, SPARTA~\cite{SPARTA}, SWF~\cite{SWF} are included.\footnote{The codes of SPARTA and SWF are from \url{https://gangwg.github.io/SPARTA/index.html} and  \url{https://github.com/Ziyang1992/Sparse-Wirtinger-flow}, respectively.}  

We first briefly review these two algorithms, and compare them with SPR. 
\begin{enumerate}[i)]
\item {\it Optimization function:} The optimization function of SWF is as the same as ours, i.e.,
\begin{equation}
f(\z) = \frac{1}{2m}\sum_{i=1}^m \left(|\a_i^*\z|^2- y_i^2\right)^2,
\end{equation}
while that of SPARTA adopts the amplitude loss:
\begin{equation}
f(\z) = \frac{1}{2m}\sum_{i=1}^m \left(|\a_i^*\z|- y_i\right)^2.
\end{equation}

\item {\it Initialization:} Both SPARTA and SWF use the same initialization method, for which the required sampling complexity is given by~\cite{SPARTA,SWF}
\begin{equation}
m = \Omega(k^2\log n),
\end{equation}
under the assumption that $|x_{\min}| = \Omega ({\|\x\|}/{\sqrt{k}})$. Via this initialization, they can obtain a good estimation $\x^*$, which falls into the $\delta$-neighborhood of $\x$, i.e.,
\begin{equation}
\operatorname{dist}(\x^*, \x)\leq \delta \|\x\|.
\end{equation}
For comparison, SPR adopts a new initialization method, which does not seek an estimation that falls into the $\delta$-neighborhood of $\x$.  Instead, it only requires to capture a set $S^0$ of support indices with sufficient energy ($90\%\|\x\|^2$).

\item {\it Non-initial stage:} The main iterative steps of the SPARTA and SWF have the following form:
\begin{equation}\label{eq:swfspartarefinex}
\x^{t+1} = \mathcal H_k \left(\x^t - \mu \nabla f (\x^t)\right).
\end{equation}
Here, the gradient function $\nabla f (\cdot)$ corresponds to their respective optimization functions, and $\mathcal H_k(\cdot)$ is the hard thresholding operator which keeps the $k$ largest values in magnitude and sets others to be $0$. These steps iteratively refine the estimator $\x^t$ of $\x$.

~~The non-initial stage of SPR is different with that of SPARTA and SWF. Specifically, SPR maintains an estimated support $\hat{S}^t$ of size $k$, while refining it iteratively until convergence. In each iteration, it uses  optimization algorithms (e.g., PGD~\cite{PGD} or BB~\cite{BBalg}) to estimate the sparse signal.  
When $\hat{S}^t$ contains the true support $\supp(\x)$, SPR exactly recovers $\x$.
\end{enumerate} 

Then, we move on to the performance comparison. Both the noiseless and noisy cases are considered.

\begin{itemize}
\item {\it The noiseless case:} In order to measure the error between the recovered signal $\hat{\mathbf x}$ and the original signal $\x$, we use the normalized mean squared error (NMSE), defined as
\begin{equation}\label{eq:NMSE}
\text{NMSE}  \doteq \frac{\text{dist}(\x, \hat{\mathbf x})}{\|\x\|}.
\end{equation}
The signal recovery is considered successful if the \text{NMSE}  between the original signal $\x$ and the recovered signal $\hat{\x}$ is smaller than $10^{-6}$. We use the frequency of exact phase retrieval as a performance metric to evaluate the performance of different algorithms. 

~~We conduct two simulations for 1,000 independent Monte Carlo trials.
\begin{enumerate}[i)]
\item {\bf Case 1:} The sparsity $k$ is fixed to be $10$ and the number $m$ of measurements varies from $100$ to $400$ with step $10$. The result is shown in Fig.~\ref{fig:cc}.

\item {\bf Case 2:} The sampling number $m$ is fixed to be $800$ and $k$ varies from $5$ to $50$ with step $10$. The result is shown in  Fig.~\ref{fig:kranges}.
\end{enumerate}
~~Overall, it can be observed that SPR performs uniformly better than other algorithms under test. In particular, when $m=300$, SPR achieves $100$\% exact recovery, while SPARTA and SWF have only $35$\% and $5$\% exact recovery rate, respectively. The result indicates that refining the support set perhaps can better maintain true support indices of $\x$ than refining the signal $\x^t$ itself.  
For SPR, once the correct support sets are identified (i.e., $\supp(x) \subseteq S^t$), $\x$ is exactly recovered. However, refining $\x^t$ via~\eqref{eq:swfspartarefinex} may eliminate correct support elements in subsequent iterations, even if the support set of some midway $\x^t$ is exactly $\supp(\x)$. 

~~We discuss an issue regarding the sparsity $k$ .As shown in Algorithm~\ref{alg:SPR}, SPR is implemented under the assumption of knowing the sparsity $k$. In practice, however, the exact value of $k$ is not known in prior. We could only obtain an approximate sparsity $\tilde{k}$ of $k$. An empirical method for estimating the sparsity is based on the {correlations} $\{Z_j\}_{j=1}^n$ defined in~\eqref{eq:ZjZij}. Specifically, by sorting them in descending order, we can choose $\tilde{k}$ to be slightly greater than the number of $Z_j$'s that are significantly larger compared to the others. As such, $\tilde{k}$ is usually overestimated. In that case, SPR could perform worse than that with the exact $k$. 

~~Fig.~\ref{fig:moreK} shows such a comparison, where $m=800$, $n=1,000$, $k=10$ and $\tilde{k} \in \{1.3k, 1.6k \}$. 
The degraded performance of SPR with overestimated sparsity is due to larger ambiguity in signal space. To be specific, given a $\tilde{k}$ larger than $k$, we have to solve the subproblem~\eqref{eq:subproblem1} over a larger subspace. However, this would require more samples to ensure the subspace to have a benign geometric property. According to Proposition~\ref{prop:geoproOnsub}, the required sampling complexity would increase from $\Omega  ({k}\log^3 {k} )$ to $\Omega (\tilde{k}\log^3 \tilde{k} )$. 

\begin{figure*}[t]
\centering
\includegraphics[width=.9\textwidth]{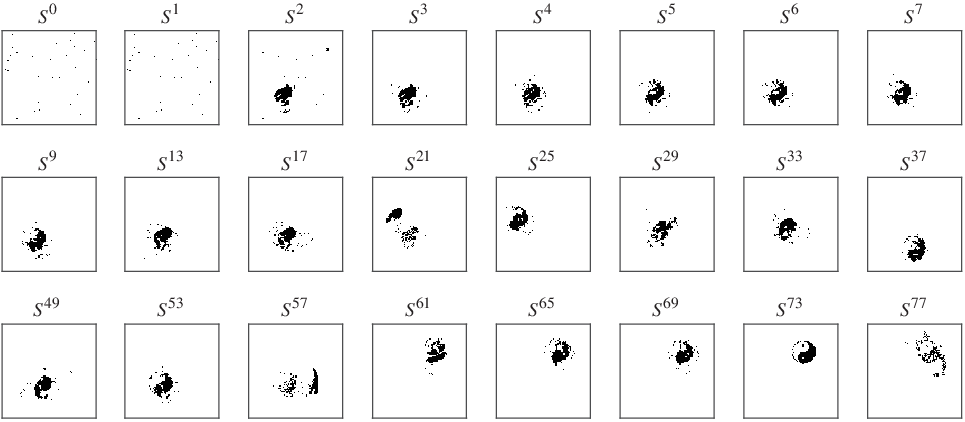}
\caption{The update of the support estimate}
\label{fig:updateofsupport}
\end{figure*}
 
\item {\it The noisy case:}
Following~\cite{SPARTA}, we consider the following noisy model to evaluate the robustness of our SPR algorithm:
\begin{equation} 
y_i = |\langle \a_i, \x \rangle| + \epsilon_i,~i = 1,\cdots, m,
\end{equation}
where $\epsilon_i$ is white Gaussian noise. In our experiment, we set $n=1,000$ and $k=10$. The Signal-to-noise ratio (SNR) varies from $5$dB to $50$dB with step size $5$dB. We still use the NMSE in~\eqref{eq:NMSE} to measure the error. 

~~Due to the presence of noise, it is impossible to exactly recover the original signal. Thus, we let SPR runs $100$ iterations at maximum, and record a best $\hat{\x}$ (i.e., the one with the smallest $f(\hat{\x})$ among all $f(\x^t)$'s). We conduct two simulations with $1,000$ independent Monte Carlo trials.

\begin{enumerate}[i)]

\item {\bf Case 1. $m=300$:} This represents a very underdetermined case. Fig.~\ref{fig:noise300} demonstrates the performance of the three algorithms. It is obvious to see that SPARTA has the smallest NMSE when the SNR is $5$dB, while SPR performs the best when the SNR is above $10$dB. 

\vspace{1mm}
\item {\bf Case 2. $m=800$:} In this setting, SPR, SWF and SPARTA all work well in the noiseless case. Fig.~\ref{fig:noise800} demonstrates the NMSE performance of three algorithms with respect to different SNR's. One can observe that SPR performs better than SWF, but worse than SPARTA, especially when the noise level is high. This is perhaps due to the matching operation of SPR, which has strong effect on the NMSE performance. Note that SPR uses a matching operation to identify support indices, and then estimates the sparse signal on the given support. However, the matching operation of SPR may be more sensitive to noise, compared to SPARTA's truncated gradient descent method.

\end{enumerate}

\end{itemize}


\begin{figure*}[t]
\centering
\includegraphics[width=.9\textwidth]{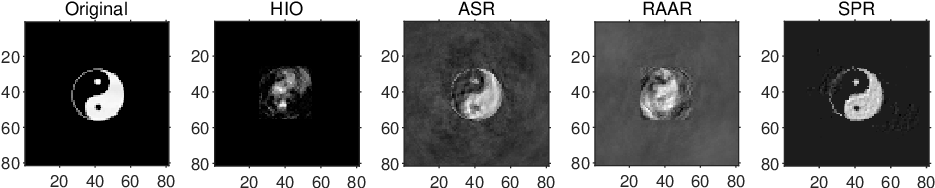}
\caption{Comparison of reconstruction results for the Tai-Chi image using HIO, ASR, RAAR, and SPR.}
\label{fig:Taiji}
\end{figure*}

\subsection{Recovery of 2D image}\label{sec:simOntwoD}

We consider the recovery of signals from Fourier samples. In this case, the matching step of SPR allows fast implementation. Specifically, the Wirtinger derivative can be computed efficiently via the fast Fourier transform (FFT):
\begin{eqnarray}\label{eq:WFFFT}
\nabla_1 f(\mathbf{z}) &=& \frac{1}{m} \sum_{i=1}^{m}{\left(\left|\mathbf{a}_{i}^{*} \mathbf{z}\right|^{2}-y_{i}^{2}\right)\left(\mathbf{a}_{i} \mathbf{a}_{i}^{*}\right) \mathbf{z}} \nonumber\\
&=& \frac{2}{m}\textsf{ifft}\left[\mathbf u \right],  \end{eqnarray}
where 
\begin{equation}\label{Eq:76}
u_i = \left(|\textsf{fft}[\z]|^{2}_i-y_{i}^{2}\right) \textsf{fft}[\z]_i,~i = 1, \cdots, m,
\end{equation}
and \textsf{fft} and \textsf{ifft} denote the FFT and inverse FFT, respectively.  
When the signal $\x$ is an $n\times n$ image, \textsf{fft} is replaced with \textsf{fft2}, i.e., the 2D FFT. 
Then, the sampling model is given by
\begin{equation}\label{eq:fft2}
\mathbf{Y} = \mathbf{FXF}^\top,
\end{equation}
where $\mathbf{F}$ is Fourier transformation matrix. \textsf{fft2} (i.e., $\mathbf{Y} = \textsf{fft2}(\mathbf{X})$) can save some matrix operations to speed up the computation. The task of phase retrieval is to recover the original image $\mathbf{x}$ from the amplitudes of elements in $\mathbf{Y}$.  {It is worth noting that in the Fourier case, $Z_j$'s of SPR are identical. As a result, $S^0$ just contain $k$ randomly selected indices, which may be far from the true support. In Fig.~\ref{fig:updateofsupport}, we report the update of support estimate over iterations. For better visualization, we have cropped the image size into $120 \times 120$ pixels.  One can observe that the support estimate converges gradually to that of Tai-Chi image. In particular, we choose the recovered image (i.e., that produced in the $73$th iteration) with the smallest NMSE\footnote{Here we use a modified version of NMSE (i.e., $\frac{\text{dist}(\y, \hat{\mathbf y})}{\|\y\|}$), since the ground-truth $\x$ is not available for computing NMSE.} as the final output.}

Fig.~\ref{fig:Taiji} shows the recovery result for a 2D Tai-Chi image of size $300 \times 300$, whose sparsity is approximately $430$. 
We choose HIO~\cite{HIO}, Averaged Successive Reflections (ASR)~\cite{ASR}, and Relaxed Averaged Alternating Reflections (RAAR)~\cite{RAAR} for comparison. The specific implementation for these algorithms is elaborated on as follows.
\begin{enumerate}[i)]
\item For HIO, ASR, and RAAR, we imposed support constraint --- a rectangular box of size $30\times 30$, on the object domain. The image \textsc{Taichi} is of size $300\times 300$.
\item For SPR, the image \textsc{Taichi} is also of size $300\times 300$ and we set the approximate sparsity to be $430$. We use~\eqref{eq:WFFFT} (the \textsf{fft2} version) as the Wirtinger derivatives which appears in Algorithm~\ref{alg:SPR}.
\end{enumerate}

To clearly show the recovery result, we crop the recovered image into a smaller size of $81\times 81$ and center the Tai-Chi symbol. One can observe that the recovery quality of SPR is the best, which implies the superiority of the proposed SPR algorithm.

Note that some works proposed to use masks to facilitate phase retrieval; see, e.g.,~\cite{maskPR,Songli}. Typically, the mask model is common in the algorithms designed for the dense case, which essentially increases the number of the samples, and thus improves the performance. We believe that it is worth future exploration to design the SPR algorithm for the mask model, which could be a promising way to further enhance its performance.

\section{Conclusion} \label{sec:conclusion}
Recently, phase retrieval has received much attention in many fields such as optical imaging. In this paper, we have proposed a new algorithm called SPR for phase retrieval. Under a mild conditions on the signal structure, our algorithm is able to reconstruct a given sparse signal when the number of magnitude-only samples is nearly linear in the sparsity level of the signal. The result outperforms some previous works, while bridging the gap to the information-theoretic bound for the sparse phase retrieval problem.

Numerical experiments have demonstrated that SPR has competitive recovery quality compared to the state-of-the-art phase retrieval techniques, especially when the number of samples is small. On account of the fast convergence and ease of implementation (e.g., no need to know the support of input signal in advance), the proposed SPR approach can serves an attractive alternative to the classic HIO~\cite{HIO}, ASR~\cite{ASR}, and RAAR~\cite{RAAR} method for phase retrieval. 

We would like to point out some directions worth of future investigation. 
\begin{itemize}
\item First of all, while in this paper we are primarily interested in analyzing SPR with complex Gaussian samples, the SPR algorithm nevertheless shows promising reconstruction result empirically for Fourier samples. Thus we speculate that our theoretical results may hold for the Fourier setup as well. In essence, this would require that the (random) Fourier samples satisfy some concentration inequalities such as those in  Lemma~\ref{lm:matrixconcentration}. 
\vspace{1mm}
\item The second direction concerns the performance guarantee of SPR in the presence of noise. A favorable geometric property and some techniques in~\cite{WUFAN,Songli} may offer a route to a theoretical guarantee for this scenario and help to uncover the whole story of SPR. This remains a topic of ongoing work. 

\vspace{1mm}
\item Thirdly, note that our result is still worse than the information-theoretic bound $\Omega (k \log n)$~\cite{infobound1, infobound2}.  Very recently, Xia and Xu~\cite{PhaseLiftOff} have proposed the Sparse PhaseLiftOff model based on the $\ell_1$-relaxation and difference-of-convex algorithm (DCA).  Through restricted isometry property (RIP) analysis, it is shown that with $\Omega(k \log \frac{n}{k})$ phaseless Gaussian samples, the global minimum of their proposed model results in a recovery error (to $\x \x^*$) that is upper bounded by a constant multiple of $\frac{k}{m}$. While this development should be considered a major step forward, the DCA solving the Sparse PhaseLiftOff model is only proved to have local convergence. It is not clear if the global convergence can be achieved. At present, whether there exist practical algorithms attaining the information-theoretic bound for exact recovery remains an interesting open question. 
\end{itemize}

\begin{appendix}

\section{Basic tools and lemmas}
\begin{lem}[Even Moments of Complex Gaussian, Lemma 25 in~\cite{sunju}]\label{Lm:ComplexGauMoment}
For any standard complex Gaussian variable $a \sim \mathcal{CN}(1)$, it holds that for any positive integer $p$
\begin{equation}
\E\left[|a|^{2p}\right] = p!.
\end{equation}

\end{lem}
\begin{lem}[Integral Form of Taylor’s Theorem, Lemma 26 in~\cite{sunju}]\label{lem:Taylor}
Consider any continuous function $f(\z): \C^n \to \mathbb R$ with continuous first and second-order Wirtinger derivatives. For any $\boldsymbol{\delta} \in \C^n$ and scalar $t \in \mathbb R$, we have
\begin{align}
f(\z+t \boldsymbol{\delta})  =&f(\z)+t \int_0^1\begin{bmatrix}
	\boldsymbol{\delta} \\
	\overline{\boldsymbol{\delta}}
\end{bmatrix}^{*} \nabla f(\z+s t \boldsymbol{\delta}) \mathrm{d} s, \\
\nonumber		f(\z+t \boldsymbol{\delta})  =&f(\z)+t\begin{bmatrix}
	\boldsymbol{\delta} \\
	\overline{\boldsymbol{\delta}}
\end{bmatrix}^{*} \nabla f(\z)\\
&+t^2 \int_0^1(1-s)\begin{bmatrix}
	\boldsymbol{\delta} \\
	\overline{\boldsymbol{\delta}}
\end{bmatrix}^{*} \nabla^2 f(\z+s t \boldsymbol{\delta})\begin{bmatrix}
	\boldsymbol{\delta} \\
	\overline{\boldsymbol{\delta}}
\end{bmatrix}\mathrm{d} s.
\end{align}
\end{lem}
\begin{lem}[Lemma 21 in~\cite{sunju}]\label{lem:6.3}
Let $\a_1, \dots, \a_m$ be {\it i.i.d.} copies of $\a \sim \mathcal{CN}(n)$. For any $\delta \in (0, 1)$ and any $\mathbf v \in \mathbb{C}^n$ such that $\text{supp}(\mathbf v) = \mathcal{T}$ where $|\mathcal T| = k$, when $$m\ge C(\delta, \v) k\log k,$$ it holds with probability at least $1-c_1 \delta ^{-2}m^{-1}-c_2 \exp(-c_c \delta ^ 2 m / \log m)$ that 
\begin{align}
\left\|\left(\frac1m \sum_{i=1}^m |\a^*_i \v|^2 \a_i\a_i^* - (\v\v^* + \|\v\|^2 I)\right)_\TT\right\| &\le \delta \|\v\|^2,\label{Eq78}\\
\left\|\left(\frac1m \sum_{i=1}^m (\a_i^* \v)^2 \a_i\a_i^T - 2\v\v^T\right)_{\mathcal{T}}\right\|&\le \delta \|\v\|^2.\label{Eq79}
\end{align}
Here $C(\delta, \v)$ is a constant depending on $\delta, \v$ and $c_1, c_2$ and $c_3$ are positive absolute constant.
\end{lem}

\begin{lem}[Lemma 22 in~\cite{sunju}]\label{lem:6.4}
Let $\a_1, \dots , \a_m$ be {\it i.i.d.} copies of $\a \sim \mathcal{CN}(n)$. For any $\delta \in (0, 1)$, when $$m\ge C(\delta) k \log k,$$ it holds with probability at least $1-c_1\exp(-c(\delta)m) - c_2 m^{-k}$ that 
\begin{align}
&\frac1m \sum_{i = 1}^m |\a_i^* \z|^2|\a_i^*\w|^2 \ge (1-\delta)(\|\w\|^2 \|\z\|^2 + |\w^* \z|),\\
&\frac1m \sum_{i =1}^m [\Re{}(\a_i^* \z)(\w^* \a_i)]^2 \ge \nonumber\\
&\hspace{10mm} (1-\delta)(\frac12 \|\z\|^2 \|\w\|^2 + \frac32 [\Re{}\z^* \w] - \frac12 [\Im{} \z^* \w])
\end{align}
for all $\z,\w \in \mathbb{C}^\T$ with $|\T| \le k$. Here $C(\delta)$ and $c(\delta)$ are constants depending on $\delta$, and $c_1$ and $c_2$ are positive absolute constants.
\end{lem}

\begin{lem}[subgaussian Lower Tail for Nonnegative RV's Problem $2.9$ in~\cite{lemmaa7}, Lemma 31 in~\cite{sunju}]\label{lem:A.7}
Let $X_1, \dots, X_n$ be {\it i.i.d.} copies of the nonnegative random variable $X$ with finite second moment. Then it holds that
\begin{equation}
\mathbb P\left[\frac1n \sum_{i=1}^n (X_i - \mathbb E[X_i]) < -\epsilon\right]\le \exp\left(-\frac{n\epsilon^2}{2\sigma^2}\right) 
\end{equation} 
for any $\epsilon > 0$, where $\sigma^2 = \mathbb E[X^2]$.
\end{lem}

\begin{lem}[Moment-Control Bernstein's inequality for Random Variables, Lemma 32 in~\cite{sunju}]\label{lem:A.8}
Let $X_1,\dots , X_n$ be {\it i.i.d.} copies of a real-valued random variable $X$ suppose that there exist some positive number $R$ and $\sigma_X^2$ such that for all integers $m\ge 3$,
$\E[|X|^2] \le \sigma_X^2$ and
\begin{equation}
    \E[|X|^m] \le \frac{m!}{2}\sigma_X^2 R^{m-2}.
\end{equation}
Let 
\begin{equation}
 S = \frac1n \sum_{i=1}^n X_i,
 \end{equation} 
 then it holds that
\begin{equation}
\P(|S - \E[S]|\ge \epsilon)\le 2\exp\left(-\frac{n\epsilon^2}{2\sigma_X^2 + 2R\epsilon}\right).
\end{equation}
\end{lem}

\begin{lem}\label{lm:matrixconcentration}
For any $\z\in \mathbb{C}^\T, |\mathcal{T}| = k$ and $\delta > 0$, when $$m\ge C(\delta, \z)k\log k,$$ it holds with probability at least $1-c(\delta, \z)m^{-1}$ that 
\begin{eqnarray}
\|\nabla f(\z)_\mathcal T - \nabla\mathbb{E}[ f(\z)]_{\mathcal{T}} \| &\hspace{-2mm} \leq &\hspace{-2mm} \sqrt2 \delta\|\z\|(\|\z\|^2 +\|\x\|^2),~~~~~~~\\
\|\nabla^2 f(\z)_{\T} - \nabla^2 \mathbb{E}[f(\z)]_{\T}\| 
&\hspace{-2mm} \leq &\hspace{-2mm} \delta (\|\x\|^2+\|\z\|^2).\label{eq:concentrationforh}
\end{eqnarray}
\end{lem}
\begin{proof}
Denote $\tilde{{\T}}  \doteq {\T} \cup \operatorname{supp}(\mathbf{x})$. 	Under the given condition, we have
\begin{eqnarray}
&\hspace{-4mm}  &\hspace{-8mm} \left\|\nabla f(\mathbf{z})_{\T}-\nabla\mathbb{E}\left[ f(\mathbf{z})\right]_{\T}\right\|   \nonumber\\
&\hspace{-6mm}   = &\hspace{-2mm} \sqrt{2}\left\|\left(\frac{1}{m} \sum_{i=1}^{m}\left(\left(\left|\mathbf{a}_{i}^{*} \mathbf{z}\right|^{2}-y_{i}^{2}\right) \mathbf{a}_{i} \mathbf{a}_{i}^{*}\right)_{\TT} \right. \right. \nonumber\\
&\hspace{-6mm}  &\hspace{-2mm} \left.\left.   -\left(2\|\mathbf{z}\|^{2}-\|\mathbf{x}\|^{2}\right) \boldsymbol{I}_{\TT}+\mathbf{x}_{\T} \mathbf{x}_{\T}^{*}\right) \mathbf{z}\right\| \nonumber\\
&\hspace{-6mm}  \leq &\hspace{-2mm} \sqrt{2}\|\mathbf{z}\|\left(\left\|\left(\frac{1}{m} \sum_{i=1}^{m}\left|\mathbf{a}_{i}^{*} \mathbf{z}\right|^{2} \mathbf{a}_{i} \mathbf{a}_{i}^{*}-\left(\|\mathbf{z}\|^{2} \boldsymbol{I}+\mathbf{z z}^{*}\right)\right)_{\TT} \right\|\right.\nonumber\\
&\hspace{-6mm}  &\hspace{-2mm} \left.+\left\|\left(\frac{1}{m} \sum_{i=1}^{m}\left|\mathbf{a}_{i}^{*} \mathbf{x}\right|^{2} \mathbf{a}_{i} \mathbf{a}_{i}^{*}-\left(\|\mathbf{x}\|^{2} \boldsymbol{I}+\mathbf{x x}^{*}\right)\right)_{\TT}\right\|\right)\nonumber \\
&\hspace{-6mm}  \leq &\hspace{-2mm} \sqrt{2}\|\mathbf{z}\|\left(\left\|\left(\frac{1}{m} \sum_{i=1}^{m}\left|\mathbf{a}_{i}^{*} \mathbf{z}\right|^{2} \mathbf{a}_{i} \mathbf{a}_{i}^{*}-\left(\|\mathbf{z}\|^{2} \boldsymbol{I}+\mathbf{z z}^{*}\right)\right)_{\TT} \right\|\right. \nonumber\\
&\hspace{-6mm}  &\hspace{-2mm} \left.+\left\|\left(\frac{1}{m} \sum_{i=1}^{m}\left|\mathbf{a}_{i}^{*} \mathbf{x}\right|^{2} \mathbf{a}_{i} \mathbf{a}_{i}^{*}-\left(\|\mathbf{x}\|^{2} \boldsymbol{I}+\mathbf{x x}^{*}\right)\right)_{\mathcal{\tilde{T}}}\right\|\right),~~~
\end{eqnarray}
where the last inequality is from Cauchy's interlace theorem. 

By applying Lemma~\ref{lem:6.3} with $\mathbf{v}  =\z$ and $\mathbf{x} = \x$, we can conclude that 
\begin{equation}
\left\|\nabla f(\mathbf{z})_{\T}-\nabla\mathbb{E}\left[ f(\mathbf{z})\right]_{\T}\right\| \leq \sqrt{2}\delta \|\z\|(\|\x\|^2+\|\z\|^2).
\end{equation}
We can use the same method to show 
\begin{align}
\|\nabla^2 f(\z)_{\T} - \nabla^2 \mathbb{E}[f(\z)]_{\T}\| & \leq \delta (\|\x\|^2+3\|\z\|^2)\nonumber\\
& \leq 3\delta (\|\x\|^2+\|\z\|^2),
\end{align} which implies that~\eqref{eq:concentrationforh} holds. The proof is thus complete.  
\end{proof}

\begin{lem}[Theorem 8 and Theorem 9  in~\cite{boundedBernstein}]\label{lm:boundedsubgaussian0}
Let $X_i$ be independent real random variables satisfying $|X_i|\leq M$ for all $i\in \{1,\cdots, n\}$. Let $X=\sum_{i=1}^nX_i$ and $\sigma = \sqrt{\sum_{i=1}^n\E[|X_i|^2]}$. Then, we have
\begin{equation}
\mathbb{P}[|X-\mathbb{E}[X]|\geq \epsilon] \leq 2 \exp \left(-\frac{\epsilon^2}{2\left(\sigma^2+M \epsilon / 3\right)}\right).
\end{equation}
\end{lem}

\begin{lem}\label{lm:boundedsubgaussian}
Let $X_i$ be independent complex random variables satisfying $|X_i|\leq M$ for all $i\in\{1,\cdots, n\}$. Let $X=\sum_{i=1}^nX_i$ and $\sigma = \sqrt{\sum_{i=1}^n\E[|X_i|^2]}$. Then, we have
\begin{equation}
\mathbb{P}[|X-\mathbb{E}[X]|\geq \epsilon] \leq 4 \exp \left(-\frac{\epsilon^2}{4\left(\sigma^2+\sqrt{2}M \epsilon/ 6\right)}\right).
\end{equation}
\end{lem}
\begin{proof}
Note that we can write a complex variable $z$ as $z = a+b\mathsf{j}$. Applying  Lemma~\ref{lm:boundedsubgaussian0} on the real and imaginary parts with $\epsilon\gets\frac{\sqrt{2}}{2}\epsilon$ separately and then using the union bound yields the result.
\end{proof}

\section{Proofs of propositions~\ref{prop:expectedGeo}--\ref{lm:elemenEzxtbound}}

\subsection{Proof of Proposition~\ref{prop:expectedGeo}} \label{app:prop1}
\begin{proof}
We consider the following cases.
\begin{enumerate}[i)]
\item It is trivial that $\z = \mathbf{0}$ is the zero point of $\nabla \mathbb{E}[f(\z)]_{\T}$. Also, it is obvious that
\begin{equation}
	\nabla^2 \mathbb{E}[f(\mathbf{0})]_{\T}  \prec 0.
\end{equation}
Thus $\z = \mathbf{0}$ is the local maximum of $\mathbb{E}[f(\z)]$ in the subspace $\C^\T$.

\vspace{2mm}
\item In the region of $\{\z \in \mathbb{C}^{\T}: 0<\|\z\|^2 < \frac12 \|\x\|^2\}$, observe that
\begin{equation}
	\begin{bmatrix}\z\\\overline{\z}\end{bmatrix}^* \nabla \mathbb{E}[f(\z)]_{\T} = 2(2\|\z\|^2 - \|\x\|^2)\|\z\|^2 - 2|\x_{\T}^*\z|^2 < 0.
\end{equation}
Thus, there is no zero point of $\nabla \mathbb{E}[f(\z)]_{\T}$ in this region.

\vspace{2mm}
\item When $\|\z\|^2 = \frac12 \|\x\|^2$, 
\begin{equation}
	\nabla \mathbb{E}[f(\z)]_{\T} = \begin{bmatrix}-\x_{\T}\x_{\T}^* \z \\ -\overline{\x}_{\T}\x^\top_{\T}\overline{\z}\end{bmatrix}.
\end{equation}
We have
\begin{eqnarray}
	\nabla \mathbb{E}[f(\z)]_{\T} = 0 
	& \Leftrightarrow &  \z\in \operatorname{Null}(\x_{\T}\x_{\T}^*) \nonumber \\
	& \Leftrightarrow &  \x_{\T}^* \z = 0 \nonumber \\
	& \Leftrightarrow &  \x^* \z = 0.
\end{eqnarray}
Thus, any point in $\mathcal S \doteq \big \{\z\in \mathbb{C}^{\T}: \x^* \z = 0, \|\z\| = \frac{  \|\x\|  } {\sqrt2  }  \big \}$ is the zero point of $\nabla \mathbb{E}[f(\z)]_{\T}$.

\begin{figure*}[t]%
	\centering
	\subfloat{
		\includegraphics[width=.44\linewidth]{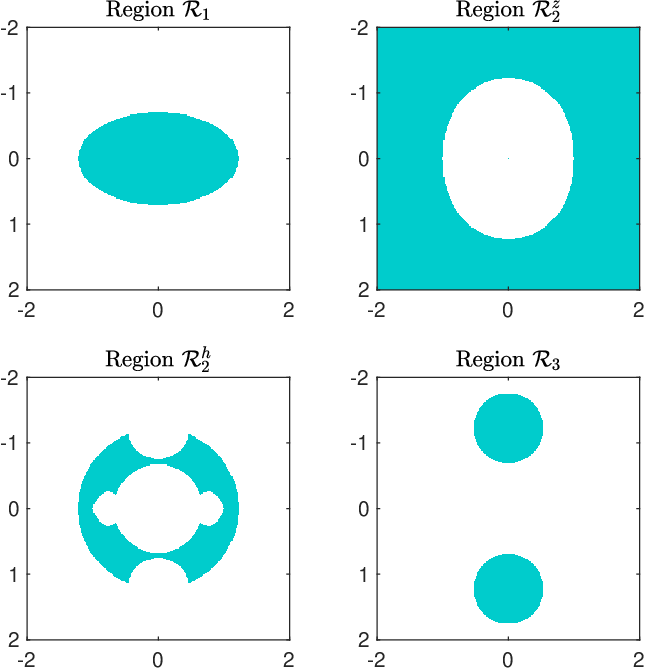}
		\label{fig:subregions}
	}\hfill
	\subfloat{
		\includegraphics[width=.5\linewidth]{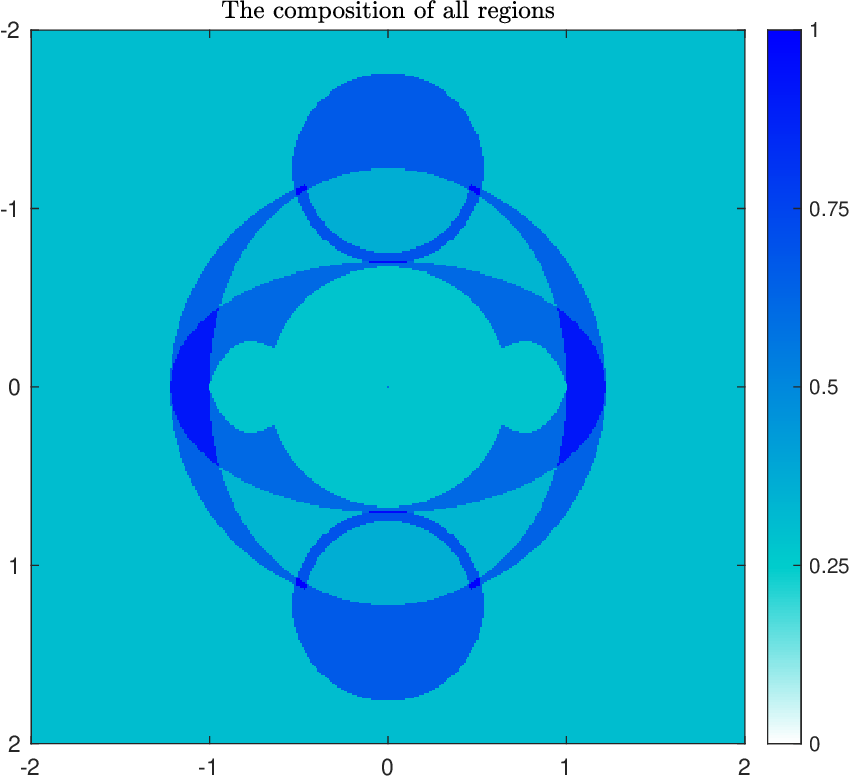}
		\label{fig:compositonregions}
	}
	\caption{\black{An illustration of the region division. Here, we suppose that all quantities are real numbers and the target  $2$-sparse signal is $\x=[1, 0, 1, 0]^\top$ (i.e., sparsity $k=2$). The estimated support is $\T=\{1,2\}$ and thus the subspace is $\R^{\T} = \{\z\in\R^4| \supp(\z) = \{1,2\}\}$. Accordingly, the solution to the subproblem~\eqref{eq:subproblem1} is $\u = \omega_{\T} \x_{\T} = \sqrt{3/2}[1, 0, 0, 0]^\top$. 
	For ease of presentation, we show only the first two dimensions of $\R^{\T}$. That is, the illustration area is $\{\z\in\mathbb R^2|-2\leq z_1\leq 2, -2\leq z_2\leq 2\}$. Regions $\R_1$, $\R_2^z$, $\R_2^h$, and $\R_3$ are displayed on the left, marked in blue, while the composition of these four regions is on the right. The overlapped areas are highlighted in darker blue.  Clearly the union $\R_1 \cup \R_2^z \cup \R_2^h \cup \R_3$ can spread over the entire illustration area. }}
	\label{fig:areaDivision}
\end{figure*}

~~For any $\z \in \mathcal S$, it holds that
\begin{eqnarray}
	\begin{bmatrix}\x_{\T}\mathrm{e}^{\mathsf{j}\phi(\z)}\\ \overline{\x}_{\T} \mathrm{e}^{-\mathsf{j}\phi(\z)}\end{bmatrix}^* \nabla^2 \mathbb{E}[f(\z)]_{\T} \begin{bmatrix}\x_{\T}\mathrm{e}^{\mathsf{j}\phi(\z)}\\ \overline{\x}_{\T} \mathrm{e}^{-\mathsf{j}\phi(\z)}\end{bmatrix} & \hspace{-1mm} = & \hspace{-1mm} -2\|\x_{\T}\|^4 \nonumber \\
	& \hspace{-1mm} < & \hspace{-1mm} 0
\end{eqnarray}
and 
\begin{equation}
	\begin{bmatrix}\z\\\overline{\z} \end{bmatrix}^* \nabla^2 \mathbb{E}[f(\z)]_{\T}\begin{bmatrix}\z\\\overline{\z} \end{bmatrix} = 8\|\z\|^2 > 0.
\end{equation}
Hence, any $\z \in \mathcal S$ is the saddle point of $\mathbb{E}[f(\z)]$ on $\C^\T$.

\vspace{2mm}	
\item Consider the region of $\{\z\in \mathbb{C}^{\T}: \|\z\|^2> \frac12 \|\x\|^2 \}$. From~\eqref{eq:Expg}, the zero points of $\nabla_1 \mathbb{E}[f(\z)]_{\T}$ satisfy
\begin{equation} \label{eq:xtxtz}
	(2\|\z\|^2 - \|\x\|^2)\z =  \x_{\T}\x_{\T}^* \z,
\end{equation}
which implies that $2\|\z\|^2 - \|\x\|^2$ is an eigenvalue of the matrix $\x_{\T}\x_{\T}^*$. Since $\x_{\T}\x_{\T}^*$ is a rank-one positive semi-definite Hermite matrix $\x_{\T}\x_{\T}^*$, and also noting that $\z$ satisfies $2\|\z\|^2 - \|\x\|^2 > 0$ in this given region, the nonzero eigenvalue of $\x_{\T}\x_{\T}^*$ must be $2\|\z\|^2 - \|\x\|^2$. We further deduce that 
\begin{eqnarray}
	2\|\z\|^2 - \|\x\|^2 &=& \sigma_{\max} (\x_{\T}\x_{\T}^*) \nonumber \\
	&=&  \operatorname{Tr} (\x_{\T}\x_{\T}^*) \nonumber \\
	&=&  \operatorname{Tr} (\x_{\T}^* \x_{\T}) \nonumber \\
	&=&  \|\x_{\T}\|^2.
\end{eqnarray} 
That is, 
\begin{equation} \label{eq:|z|}
	\|\z\| = \sqrt{ \frac{\|\x\|^2 + \|\x_{\T}\|^2}{2}}. 
\end{equation}
~~Plugging it into~\eqref{eq:xtxtz} yields
\begin{equation} \label{eq:xtxtz1}
	\|\x_{\T}\|^2 \z = \x_{\T}\x_{\T}^* \z.
\end{equation}
Thus, the zero points of $\nabla_1 \mathbb{E}[f(\z)]_{\T}$  located in the given region satisfy both~\eqref{eq:|z|} and~\eqref{eq:xtxtz1}. It is easily verified that the zero point  in this region is uniquely given by 
\begin{equation}
	\z=\omega_\T \x_{\T}, 
\end{equation}
where $\omega_\T \in \mathbb{C}$ and $|\omega_\T| = \sqrt{ \frac{  \|\x\|^2 + \|\x_{\T}\|^2  }  {2\|\x_{\T}\|^2 } }$.

~~Our remaining task is to show that $\z=\omega_\T \x_{\T}$ is not only a unique zero point of $\mathbb{E}[f(\z)]$ in the region of $\{\z\in \mathbb{C}^{\T}: \|\z\|^2> \frac12 \|\x\|^2 \}$, but also the local minimum of $\mathbb{E}[f(\z)]$ for the whole subspace $\z \in \mathbb{C}^\T$. 
To this end, recall from~\eqref{eq:Expf} that 
\begin{align} \label{eq:efz} 
	\mathbb{E}[f(\z)] & =  \|{\x}\|^{4}+\|{\z}\|^{4}-\|{\x}\|^{2}\|{\z}\|^{2} -  \left|{\x}^{*} {\z}\right|^{2} \nonumber \\
	& \hspace{-2mm}= \|{\z}\|^{4} - \left ( \|{\x}\|^{2} + \left|{\x}^{*} \left ( \frac{\z}{\|\z\|} \right ) \right|^{2} \right ) \|{\z}\|^{2} + \|{\x}\|^{4},  
\end{align}
which attains its minimum when $\left|{\x}^{*} \left ( \frac{\z}{\|\z\|} \right ) \right|^{2} $ is maximized, that is, when \begin{equation}
	\z = \rho {\x}_\T  \label{eq:rhoext}
\end{equation}
where $\rho \in \mathbb{C}$ and $|\rho| = \frac{\| \z \|}{\|\x_\T\|}$.

~~Then, we can rewrite $\mathbb{E}[f(\z)]$ as
\begin{align}
	&\mathbb{E}[f(\z)] \nonumber\\
	&=\|{\x}\|^{4}+|\rho|^4\|{\x_T}\|^{4} -|\rho|^2\|{\x}\|^{2}\|{\x_T}\|^{2}-|\rho|^2\|{\x_T}\|^4\nonumber\\
	&=\|\x_\T\|^4 |\rho|^4 - \left(\|\x\|^2\|\x_\T\|^2 + \|\x_\T\|^4\right)|\rho|^2 +\|\x\|^4. 
\end{align}
Clearly this is a second-order polynomial of $|\rho|^2$, which attains the minimum when 
\begin{equation}
	|\rho| = \sqrt{\frac{\|\x\|^2+\|\x_\T\|^2}{2\|\x_T\|^2} } = |\omega_\T|.
\end{equation}

Therefore, $\z=\omega_\T \x_{\T}$ is the local minimum of $\mathbb{E}[f(\z)]$ for $\z \in \mathbb{C}^\T$.
\end{enumerate}

The proposition is established by combining the results of the above four cases.
\end{proof}

\subsection{Proof of {Proposition}~\ref{prop:geoproOnsub}}\label{SEC:proofpro2}
 
\begin{figure*}[t] 
	\normalsize
	\begin{eqnarray}
		\mathcal{R}_{1} &\doteq &\left\{\z\in\C^\T, \begin{bmatrix}
			\u \mathrm{e}^{\mathsf{j} \phi(\z)} \label{eq:arear1}\\
			\overline{\u} \mathrm{e}^{-\mathsf{j} \phi(\z)}
		\end{bmatrix}^{*} \mathbb{E}\left[\nabla^{2} f(\z)\right]\begin{bmatrix}
			\u \mathrm{e}^{\mathrm{i} \phi(\z)} \\
			\overline{\u} \mathrm{e}^{-\mathsf{j} \phi(\z)}
		\end{bmatrix} \leq-\frac1{100}\|\u\|^2 \|\z\|^2 - \frac1{50} \|\u\|^4\right\},  \\
		\label{eq:arear2z}\mathcal{R}_2^\z &\doteq & \left\{\z\in\C^\T, \Re[\langle \z, \nabla_1 \E[f(\z)]\rangle]\ge \frac1{100}\|\z\|^4+\frac{1}{500}\|\x\|^2\|\z\|^2\right\},\\
		\mathcal{R}_2^\h& \doteq& \left\{\z\in\C^\T, \Re[\langle \h(\z), \nabla_1\E[f(\z)]\rangle]\ge \frac1{250}\|\x\|^2\|\z\|\|\h(\z)\|,
		\frac{11}{20}\|\u\|\le \|\z\|\le \|\u\|,\operatorname{dist}(\z, \u)\ge \frac13\|\x\|\right\},\\
		\mathcal{R}_3 &\doteq& \left\{\z\in\C^\T, \operatorname{dist}(\z, \u)\leq\frac{1}{\sqrt7}\|\x\|\right\}.\label{eq:arear3}
	\end{eqnarray}
	\hrulefill 
	\vspace*{4pt}
\end{figure*} 
Our proof mainly follows the techniques developed in~\cite{sunju}. To begin with, we divide the solution space $\C^\T$ of~\eqref{eq:subproblem1} into four regions (i.e., $\R_1$, $\R_2^\z$, $\R_2^h$, and $\R_3$); see~\eqref{eq:arear1}--\eqref{eq:arear3}, where for notational simplicity, we have denoted  
\begin{align}
\hspace{-4mm} \u & \doteq \omega_\T \x_{\T}, \\
\hspace{-4mm} \phi(\z) & \doteq {\arg \min }_{\phi \in[0,2 \pi)} \|\mathbf{z}-\u \mathrm{e}^{\mathsf{j}\phi}\|, \\
\hspace{-4mm} \h(\z) & \doteq \z - \u \mathrm{e}^{\mathsf{j}\phi(\z)}, \\
\hspace{-4mm} {\g}(\z) & \doteq \left\{ \hspace{-1.5mm} \begin{array}{ll}
\frac{ \z-\u \mathrm{e}^{\mathsf{j} \phi(\z)} } {\left\|\z-\u \mathrm{e}^{\mathsf{j} \phi(\z)}\right\| },~~~~~~~~~~~~~~~~\text{if}~\operatorname{dist}(\z, \u) \neq 0, \\
{\g} \in \left\{{\g}: \Im\left[{\g}^{*} \z\right]=0,\|{\g}\|=1\right\},~\text{otherwise.} 
\end{array}\right.\hspace{-4mm}
\end{align}
\black{Fig.~\ref{fig:areaDivision} visualizes the division of these regions in $\mathbb R^2$, which is plotted by sampling $40,000$ points with interval $0.01$ on the area of interest. }

For each region, we will discuss the geometric property of~\eqref{eq:subproblem1} separately. 

The key difference here is that~\cite{sunju} can directly divide the solution space using $\x$, which is precisely the solution of the problem. However, it is impossible for us to do so because our solution is in a subspace $\C^\T$. Although such a solution must exist, we could by no means derive a concrete form for it, which bring a challenge to our analysis. Our primary novelty is to connect it to the expected solution $\u$ and derive a similar geometric structure (Proposition~\ref{prop:2.3}--Proposition~\ref{prop:2.7}). 

Here we provide some preliminary insights on the way that we divide $\T$ into the four areas.

\begin{enumerate}[i)]
\item {On $\R_1$, the expected Hessian has a negative curvature. One can be convinced that the real Hessian is close to the expected one when $m$ is large. Therefore, no minimizer could occur on $\R_1$.}
\item {On $\R_2^\z$ and $\R_2^\h$, the expected gradient is always nonzero. Similarly, the real gradient should be nonzero as well when $m$ is large. This prevent minimizers from occurring on $\R_2^\z$ and $\R_2^\h$. }
\item {On $\R_3$, this area contains the potential minimizers since there is no condition on it. Fortunately, we can establish a restricted strong convexity (Proposition~\ref{prop:2.4}) on this area and then show that all minimizers can only occur in an $\epsilon\|\x\|$-neighborhood around $\u$ (Proposition~\ref{prop:geoproOnsub}). }
\end{enumerate}

{Then, we proceed to rigorously characterize the geometric property of the function $f(\z)$ in the above four regions, respectively. Our results are mathematically described in the following propositions.}

\begin{prop}[Negative Curvature]\label{prop:2.3}
Consider a given $k$-dimensional subspace $\C^\T$ satisfying $\frac{\|\x_\T\|^2}{\|\x\|^2} > \frac{9}{10}$. When $m \geq C k\log^3 k$, it holds with probability at least $1 - cm^{-1}$ that
\begin{equation}
\begin{bmatrix}
	\u \mathrm{e}^{\mathsf{j} \phi(\z)} \\
	\overline{\u} \mathrm{e}^{-\mathsf{j} \phi(\z)}
\end{bmatrix}^{*} \nabla^{2} f(\z)\begin{bmatrix}
	\u \mathrm{e}^{\mathsf{j} \phi(\z)} \\
	\overline{\u} \mathrm{e}^{-\mathsf{j} \phi(\z)}
\end{bmatrix}\leq -\frac{1}{100}\|\u\|^2
\end{equation}for all $\z\in \mathcal R_1$. Here $C,c$ are positive absolute constants.
\end{prop}

\begin{prop}[Restricted Strong Convexity near $\u$]\label{prop:2.4}
Consider a given $k$-dimensional subspace $\C^\T$ satisfying $\frac{\|\x_\T\|^2}{\|\x\|^2} > \frac{9}{10}$. When $m \geq C k\log^3 k$, it holds with probability at least $1 - cm^{-1}$ that 
\begin{equation}
\begin{bmatrix}
	\g(\z)\\
	\overline{\g(\z)}
\end{bmatrix}^* \nabla^2 f(\z)
\begin{bmatrix}
	\g(\z)\\
	\overline{\g(\z)}
\end{bmatrix} \geq \frac{11}{100}\|\x\|^2 ,
\end{equation}
for all $\z\in \mathcal R_3$. Here $C,c$ are positive absolute constants.
\end{prop}

\begin{prop}[Large gradient]\label{prop:2.5}
Consider a given $k$-dimensional subspace $\C^\T$ satisfying $\frac{\|\x_\T\|^2}{\|\x\|^2} > \frac{9}{10}$. When $m \geq  C k\log k$, it holds with probability at least $1 - cm^{-1}$ that 
\begin{equation}
\z^* \nabla_1 f(\z) \geq \frac{1}{1000}\|\x\|^2\|\z\|^2
\end{equation}
for all $\z\in \mathcal R_2^\z$. Here $C,c$ are positive absolute constants.
\end{prop}

\begin{prop}[Large gradient]\label{prop:2.6}
Consider a given $k$-dimensional subspace $\C^\T$ satisfying $\frac{\|\x_\T\|^2}{\|\x\|^2} > \frac{9}{10}$. When $m \geq C k\log^3 k$, it holds with probability at least $1 - cm^{-1}$ that 
\begin{equation}
\Re\left[\h(\z)^* \nabla_1 f(\z) \right] \geq\frac1{1000} \|\x \|^2 \|\z\|\|\z - \u \mathrm{e}^{\mathsf{j}\phi}\|
\end{equation}
for all $\z\in \mathcal R_2^\h$. Here $C, c$ are positive absolute constants.
\end{prop}

\begin{prop}\label{prop:2.7}
It satisfies that
\begin{equation}
\C^\T = \R_1\cup\R_2^\z\cup\R_2^\h\cup\R_3.
\end{equation}
\end{prop}
Up to this point, we have established all the detailed geometric properties on $\C^\T$. Here we present the proof of Proposition~\ref{prop:geoproOnsub}.
\begin{proof} 
One can see from Proposition~\ref{prop:2.3} to  Proposition~\ref{prop:2.7} that all minimizers of $f(\z)$ on $\C^\T$ can only occur in $\R_3$. Moreover, it is trivial that the saddle points and the maximizers of $f(\z)$ that are out of $\R_3$ possess at least one negative curvature. Next we will show that the critical points of $f(\z)$ in $\R_3$ only occur in the area $ \{\z|\operatorname{dist}(\z,\u)\leq \epsilon\|\x\|\}$.

For any points $\z\in \R_3$ satisfying  $\operatorname{dist}(\z, \u)>0$, it can be written  as $$\z=\u \mathrm{e}^{\mathsf{j} \phi(\z)}+t \g,$$ where $\g \doteq \frac{\boldsymbol{h}(\z) }{\|\boldsymbol{h}(\z)\|}$, and $t \doteq \operatorname{dist}(\z, \u)$.
It is easy to see that $$\frac{1}{2}\|\x\|^2\leq\|\u\|^2\leq\|\x\|^2.$$ Therefore, when applying Lemma~\ref{lm:matrixconcentration}, we can simply write
\begin{equation}
\|\nabla f(\u)_\mathcal T - \nabla \mathbb{E}[ f(\u)]_{\mathcal{T}} \|\le \delta \|\x\|^3.
\end{equation}

Noting that $\nabla\E\left[f\left(\u \mathrm{e}^{\mathsf{j} \phi(\z)}\right)\right]_\T =\mathbf0$ and $\supp(\g) = \T$, applying Talyor's expansion (Lemma~\ref{lem:Taylor}) and also using Lemma~\ref{lm:matrixconcentration} and the property in $\R_3$ yields
\begin{align}
f(\z)=~&f\left(\u \mathrm{e}^{\mathsf{j} \phi(\z)}\right)+t\begin{bmatrix}
	\g\\
	\overline{\g}
\end{bmatrix}^{*}\nabla f\left(\u \mathrm{e}^{\mathsf{j} \phi(\z)}\right) \nonumber\\
&+t^2 \int_0^1(1-s)\begin{bmatrix}
	\g\\
	\overline{\g}
\end{bmatrix}^{*}\nabla^2 f\left(\u \mathrm{e}^{\mathsf{j} \phi(\z)}+s t \g\right)\begin{bmatrix}
	\g\\
	\overline{\g}
\end{bmatrix} \mathrm{d} s\nonumber \\
=~&f\left(\u \mathrm{e}^{\mathsf{j} \phi(\z)}\right)\nonumber\\
&+t\begin{bmatrix}
	\g\\
	\overline{\g}
\end{bmatrix}^* \left[\nabla f\left(\u \mathrm{e}^{\mathrm{j} \phi(\z)}\right)_\T - \nabla\E f\left(\u \mathrm{e}^{\mathsf{j} \phi(\z)}\right)_\T\right] \nonumber\\
&+t^2 \int_0^1(1-s)\begin{bmatrix}
	\g\\
	\overline{\g}
\end{bmatrix}^* \nabla^2 f\left(\u \mathrm{e}^{\mathsf{j} \phi(\z)}+s t \g\right)\begin{bmatrix}
	\g\\
	\overline{\g}
\end{bmatrix} \mathrm{d} s\nonumber \\
\geq~ &f\left(\u \mathrm{e}^{\mathsf{j} \phi(\z)}\right) + \frac{11}{200}\|\x\|^2 t^2 - t\delta \|\x\|^3.\label{eq:applylemmawon}
\end{align}
Similarly,
\begin{align}
f\left(\u \mathrm{e}^{\mathsf{j} \phi(\z)}\right) \geq~ &f(\z)-t\begin{bmatrix}
	\g\\
	\overline{\g}
\end{bmatrix}^* \nabla f(\z)\nonumber\\
&\hspace{-2mm}+t^2 \int_0^1(1-s)\begin{bmatrix}
	\g\\
	\overline{\g}
\end{bmatrix}^* \nabla^2 f(\z-s t \g)\begin{bmatrix}
	\g\\
	\overline{\g}
\end{bmatrix}\mathrm{d} s\nonumber \\
\geq~& f(\z)-t\begin{bmatrix}
	\g\\
	\overline{\g}
\end{bmatrix}^* \nabla f(\z)+\frac{11}{200}\|\x\|^2 t^2 .
\end{align}

Finally, we have
\begin{equation}
t\begin{bmatrix}
	\g\\
	\overline{\g}
\end{bmatrix}^* \nabla f(\z) \geq \frac{11}{100}\|\x\|^2 t^2 -t\delta\|\x\|^3,
\end{equation}
which implies 
\begin{equation}\label{leq:gradient>0}
\|\nabla f(\z)\| \geq \frac{11}{100 \sqrt{2}}\|\x\|^2 t - \frac{\delta}{\sqrt 2}\|\x\|^3.
\end{equation}
Therefore, when $t>\frac{100}{11} \delta\|\x\|$, $\|\nabla f(\z)\| >0$. In other words, take $\epsilon = \frac{100}{11}\delta$, and then the local minimizers of $f(\z)$ only occur in the area $ \{\z|\operatorname{dist}(\z,\u)\leq \epsilon\|\x\|\}$. 

For the special case, i.e., $\T\supseteq\supp(\x)$, we have that $\u=\x$ and \begin{equation}
    \nabla f\left(\u \mathrm{e}^{\mathsf{j} \phi(\z)}\right)=\nabla \E\left[f\left(\u \mathrm{e}^{\mathsf{j} \phi(\z)}\right)\right] = \mathbf{0}.
\end{equation} 
Then,~\eqref{leq:gradient>0} is in fact
\begin{equation}
\|\nabla f(\z)\| \geq \frac{11}{100 \sqrt{2}}\|\x\|^2 t,
\end{equation}
which implies that when $t>0$, the Wirtinger gradient is nonzero. Therefore, all local minimizers must be $\x$, which completes the proof.
\end{proof}

\subsection{Proof of Proposition~\ref{lm:elemenEzxtbound}} \label{sec:proofproposition3}

Observe that 
\begin{align}\label{leq:elementbound}
&~|\nabla_1 f(\z)_{l} - \nabla_1 \mathbb{E}[f(\z)]_{l}| \nonumber \\
\overset{\eqref{eq:10daoshu},\eqref{eq:Expg}}{=}~&\left|\left[\frac{1}{m}\sum_{i=1}^{m} |\a_i^*\z|^2 (\a_i^*\z) a_{il}- \frac{1}{m}\sum_{i=1}^{m} |\a_i^*\x|^2 (\a_i^*\z) a_{il}\right] \right.\nonumber\\
&~\left.-  \left[2\|\z\|^2z_l-\left( \|\x\|^2z_l+(\x^*\z)x_l \right)\right]\right|. \nonumber \\
\leq~~~~&\left|\left[\frac{1}{m}\sum_{i=1}^{m} |\a_i^*\z|^2 (\a_i^*\z) a_{il}- 2\|\z\|^2z_l\right] \right|\nonumber\\
&~+\left|  \left[\frac{1}{m}\sum_{i=1}^{m} |\a_i^*\x|^2 (\a_i^*\z) a_{il}-\left( \|\x\|^2z_l+(\x^*\z)x_l \right)\right]\right|.\nonumber\\
\end{align}
To prove Proposition~\ref{lm:elemenEzxtbound}, we can show that the right-hand side of~\eqref{leq:elementbound} is upper bound by $\epsilon \|\x\|^2|x_{\min}|$. It suffices to show that both terms on the right-hand side of~\eqref{leq:elementbound} are upper bounded by $\frac{\epsilon}{2} \|\x\|^2|x_{\min}|$.

First, we establish an upper bound for the second term, i.e., 
\begin{equation}\label{leq:elementxbound}
\left|\frac{1}{m}\sum_{i=1}^{m} |\a_i^*\x|^2 (\a_i^*\z) a_{il}\hspace{-0.25mm}-\hspace{-0.25mm}\left( \|\x\|^2z_l\hspace{-0.25mm}+\hspace{-0.25mm}(\x^*\z)x_l \right)\right|\hspace{-0.25mm}<\hspace{-0.25mm}\frac{\epsilon}{2} \|\x\|^2|x_{\min}|
\end{equation} 
holds for all $l\in\{1,2,\cdots, n\}$. 

\vspace{2mm}
\noindent {\bf 1. Proof of~\eqref{leq:elementxbound}}
\begin{proof}

To derive this bound, we use a truncation strategy, which is inspired from~\cite[Lemma~1]{HWF} and~\cite[Proposition~2.6]{sunju}. 

Define three indicator functions
\begin{align} \label{eq:indicator1}
s_{i,1} &= \mathbf{1}_{\left\{|\a_i^*\x| < \sqrt{32\log n}\|\x\|\right\}}, \\
s_{i,2} &= \mathbf{1}_{\left\{|\a_i^*\z| < \sqrt{32\log n}\|\z\|\right\}}, \label{eq:indicator2} \\
s_{i,3} &= \mathbf{1}_{\left\{|a_{il}| < \sqrt{32\log n}\right\}}, \label{eq:indicator3}
\end{align}
and denote
\begin{align}
R_l &= \frac{1}{m}\sum_{i=1}^{m} |\a_i^*\x|^2 (\a_i^*\z) a_{il},\\
R_{l1,i} &= \frac{1}{m} |\a_i^*\x|^2 (\a_i^*\z) a_{il} \cdot (s_{i,1}s_{i,2}s_{i,3}),\\
R_{l1} &= \sum_{i=1}^{m}R_{l1,i} ,\\
R_{l2} &= R_l-R_{l1}.
\end{align}
Then,~\eqref{leq:elementxbound} can be rewritten as
\begin{align}
\big |R_{l}-\E[R_{l}] \big | < \frac{\epsilon}{2} \|\x\|^2|x_{\min}|. 
\end{align}

Note that 
\begin{align}
|R_{l}-\E[R_{l}]| &= \big | R_{l1}+R_{l2}-\E[R_{l1}]-\E[R_{l2}] \big | \nonumber \\
&\leq  \big | R_{l1}-\E[R_{l1}] \big | + \big |R_{l2}-\E[R_{l2}] \big |.
\end{align}
Thus, to prove~\eqref{leq:elementxbound}, it suffices to show that $|R_{l1}-\E[R_{l1}]|$ and $|R_{l2}-\E[R_{l2}]|$ are both upper bounded by $\frac{\epsilon}{4}\|\x\|^2|x_{\min}|$. To this end, we shall bound these two terms separately with probabilistic arguments.

\begin{enumerate}[i)]
\item {\bf Upper bound of $\big | R_{l1}-\E[R_{l1}] \big |$:}

We will show that the probability 
$
\P\left(|R_{l1}-\E[R_{l1}] | \geq \frac{\epsilon}{4}\|\x\|^2|x_{\min}|\right).
$
is vanishing by means of Lemma~\ref{lm:boundedsubgaussian}.

~~Before we proceed, we estimate the ``$\sigma^2$'' in this lemma:
\begin{align}
	\sigma^2 &=~\sum_{i=1}^{m}\E \Big [|R_{l1,i}|^2 \Big ]\nonumber\\
	&\leq~\frac{1}{m^2}\sum_{i=1}^{m} \E \Big [|\a_i^*\x|^4 |\a_i^*\z|^2|a_{il}|^2 \Big ]  \nonumber \\
	&\overset{(a)}{\leq}~\frac{1}{m} \sqrt[3]{\E \big [|\a_i^*\x|^{12} \big ]\E \big [|\a_i^*\z|^6 \big ]\E \big [|a_{il}|^6 \big ]}\nonumber\\
	&\overset{(b)}{\leq}~ \frac{30}{m}\|\x\|^4\|\z\|^2, \label{eq:sigma2_160}
\end{align}
where (a) is from the H{\"o}lder's inequality and (b) is from Lemma~\ref{Lm:ComplexGauMoment} and the fact that the complex Gaussian distribution is invariant to unitary transform.

~~Besides, it is trivial to show that $R_{l1,i}$ are mutually independent and that
\begin{equation}
	|R_{l1,i}| \overset{\eqref{eq:indicator1}-\eqref{eq:indicator3}}{<} \frac{1024\log^2n}{m}\|\x\|^2\|\z\|.
\end{equation}

~~Then, applying Lemma~\ref{lm:boundedsubgaussian}, we have
\begin{align}
	&~~~~\P\left(|R_{l1}-\E[R_{l1}] | \geq \frac{\epsilon}{4}\|\x\|^2|x_{\min}|\right) \nonumber\\
	& \leq 4\exp \left(-\frac{\frac{\epsilon^2}{64}\|\x\|^4|x_{\min}|^2}{\frac{30}{m}\|\x\|^4\|\z\|^2 + \frac{512\sqrt{2}\epsilon\log^2 n}{3m} \|\x\|^4\|\z\||x_{\min}| }\right) \nonumber\\    
	& = 4\exp \left(-\frac{m\epsilon^2 |x_{\min}|^2}{1920 \|\z\|^2 + \frac{32768\sqrt{2}}{3} \epsilon \log^2 n \|\z\||x_{\min}|} \right) \nonumber\\
	& \overset{(c)}{\leq} 4\exp \left(-\frac{m\epsilon^2 |x_{\min}|^2}{1920\kappa\|\x\|^2 + \frac{32768\sqrt{2}}{3}\sqrt{\kappa}\epsilon\log^2 n \|\x\||x_{\min}|} \right) \nonumber\\
	&  \overset{(d)}{\leq} 4 \exp\left(-\frac{c_1m}{k + c_2\sqrt{k}\log^2 n)} \right), \label{eq:p_162}
\end{align}
where (c) uses the condition that $\frac{\|\z\|^2}{\|\x\|^2} \leq \kappa$ and (d) is because $|x_{\min}| = \Omega  (\|\x\|/\sqrt{k}).$

~~Therefore, when $m\geq C_1\max\{k\log n, \sqrt{k}\log^3n\}$,  
\begin{equation}\label{leq:rl1}
	|R_{l1}-\E[R_{l1}] | < \frac{\epsilon}{4}\|\x\|^2|x_{\min}|
\end{equation}
holds with probability exceeding $1-\frac{c_3}{n^2}$.

\vspace{2mm}
\item {\bf Upper bound of $\big | R_{l2}-\E[R_{l2}] \big |$:}

We will use Chebyshev's inequality to show that the probability 
$
\P\left( \big |R_{l2}-\E[R_{l2}] \big | \geq \frac{\epsilon}{4} \|\x\|^2|x_{\min}|\right) 
$
is vanishing.

~~For notational simplicity, we define
\begin{align}
	p_{i,1} &= \P\left(|\a_i^*\x| \geq \sqrt{32\log n}\|\x\|\right),\\
	p_{i,2} &= \P\left(|\a_i^*\z| \geq \sqrt{32\log n}\|\z\|\right),\\
	p_{i,3} &= \P\left(|\a_{il}| \geq \sqrt{32\log n}\right),
\end{align}
which correspond to three indicators in~\eqref{eq:indicator1}--\eqref{eq:indicator3}. Then, it can be easily verified that $p_{i,1}$, $p_{i,2}$ and $p_{i,3}$ are all upper bounded by $n^{- 16}$.

~~In order to apply the Chebyshev's inequality, we estimate the variance of $R_{l2}$ as follows:
\begin{eqnarray}
	&& \hspace{-12mm} \operatorname{Var} \left[R_{l2}\right] \nonumber \\
	& \hspace{-2mm} = & \hspace{-2mm}  \operatorname{Var}\left[\frac{1}{m}\sum_{i=1}^{m} |\a_i^*\x|^2 (\a_i^*\z) a_{il}(1-s_{i,1}s_{i,2}s_{i,3})\right] \nonumber \\ 
	& \hspace{-2mm} \leq  & \hspace{-2mm}  \frac{1}{m^2}\sum_{i=1}^{m}\operatorname{Var} \Big [ |\a_i^*\x|^2 (\a_i^*\z) a_{il}(1-s_{i,1}s_{i,2}s_{i,3}) \Big ] \nonumber \\ 
	& \hspace{-2mm} \leq & \hspace{-2mm} 
	\frac{1}{m^2} \sum_{i=1}^{m}\E\Big[|\a_i^*\x|^4 |\a_i^*\z|^2 |a_{il}|^2 (1-s_{i,1}s_{i,2}s_{i,3})^2 \Big] \nonumber \\ 
	& \hspace{-2mm} \overset{(e)}{\leq} & \hspace{-2mm}    \frac{1}{m^2}\sum_{i=1}^{m} \Big (\left.\E \big[ |\a_i^*\x|^{16}\big] \E \big[|\a_i^*\z|^8 \big ] \E \big[|a_{il}|^8\big] \right. \nonumber \\
	& &  \times~ \E \big[ (1-s_{i,1}s_{i,2}s_{i,3})^8 \big ] \Big )^{-4} \nonumber \\ 
	& \hspace{-2mm} \overset{(f)}{\leq} & \hspace{-2mm}    \frac{1}{m^2}\sum_{i=1}^{m} \sqrt[4]{\E \big[ |\a_i^*\x|^{16}\big] \E \big[|\a_i^*\z|^8\big ] \E \big[|a_{il}|^8\big]} \nonumber \\ 
	&    &  
	\times \sqrt[4]{p_{i,1}+p_{i,2}+p_{i,3}} \nonumber \\
	& \hspace{-2mm} \overset{(g)}{\leq} & \hspace{-2mm}   \frac{92\|\x\|^4\|\z\|^2}{mn^2}, \label{eq:var167}
\end{eqnarray}
where (e) is from the H{\"o}lder's inequality and (g) is from Lemma~\ref{Lm:ComplexGauMoment} and the fact that the complex Gaussian distribution is invariant to unitary transform. (f) is because $1 - s_{i,1}s_{i,2}s_{i,3} = 1$ if any of $s_{i,1}$, $s_{i,2}$, $s_{i,3}$ is zero and $1 - s_{i,1}s_{i,2}s_{i,3} = 0$ if $s_{i,1}=s_{i,2}=s_{i,3}=1$, and thus,
\begin{align}
	& \E \left[(1-s_{i,1}s_{i,2}s_{i,3})^8\right] \nonumber \\
	&= \E [1-s_{i,1}s_{i,2}s_{i,3}]\nonumber\\
	&=\P\Big( \{ s_{i,1} = 0\} \cup \{ s_{i,2} = 0\} \cup \{ s_{i,3} = 0 \} \Big)\nonumber\\
	&  \leq  p_{i,1}+p_{i,2}+p_{i,3}.
\end{align}

~~Hence, by Chebyshev's inequality, we have
\begin{align}
	\nonumber	&~\P\left(|R_{l2}-\E[R_{l2}] |\geq \frac{\epsilon}{4}\|\x\|^2|x_{\min}|\right)\\
 \nonumber\leq&~ \frac{92\|\x\|^4\|\z\|^2}{mn^2\frac{\epsilon^2}{16}\|\x\|^4|x_{\min}|^2} \nonumber \\
	\overset{(h)}{\leq}&~\frac{1472\kappa\|\x\|^2}{mn^2\epsilon^2|x_{\min}|^2} \nonumber \\
	\overset{(i)}{\leq}&~ \frac{c_4k}{mn^2}, \label{eq:chebysheve169}
\end{align}
where (h) uses the condition that $\frac{\|\z\|^2}{\|\x\|^2} \leq \kappa$ and (i) uses the condition that $|x_{\min}| = \Omega   (\|\x\|/\sqrt{k}).$
Therefore, when $m\geq C_2 k \log n$, it holds that
\begin{equation} \label{leq:rl2}
	\P\left(|R_{l2}-\E[R_{l2}] | < \frac{\epsilon}{4}\|\x\|^2|x_{\min}|\right) > 1- \frac{c_5}{n^2}.
\end{equation}
\end{enumerate}

By combining~\eqref{leq:rl1} with~\eqref{leq:rl2}, we obtain~\eqref{leq:elementxbound} with probability exceeding $1-\frac{c_6}{n^2}$. Taking the union bound for all $l \in \{1, \cdots, n\}$,~\eqref{leq:elementxbound} holds with probability exceeding $1-\frac {c_6}{n}$ when $$m\geq C_3\max\{k\log n, \sqrt{k}\log^3n\}.$$ 
\end{proof}

Now, what remains is to show that the first term on the right-hand side of~\eqref{leq:elementbound} is upper bounded by $\frac{\epsilon}{2}\|\x\|^2|x_{\min}|$. That is, 
\begin{equation} \label{leq:elementxboundz}
\left|\left[\frac{1}{m}\sum_{i=1}^{m} |\a_i^*\z|^2 (\a_i^*\z) a_{il}- 2\|\z\|^2z_l\right] \right|<\frac{\epsilon}{2} \|\x\|^2|x_{\min}|.
\end{equation}
In fact, this can be proved by following the same technique for proving~\eqref{leq:elementxbound}.

\vspace{2mm}
\noindent {\bf 2. Proof of~\eqref{leq:elementxboundz}}
\begin{proof}
Denote
\begin{align}
R'_l &= \frac{1}{m}\sum_{i=1}^{m} |\a_i^*\z|^2 (\a_i^*\z) a_{il},\\
R'_{l1,i} &= \frac{1}{m} |\a_i^*\z|^2 (\a_i^*\z) a_{il} \cdot (s_{i,2}s_{i,3}),\\
R'_{l1} &= \sum_{i=1}^{m} R'_{l1,i} ,\\
R'_{l2} &= R'_l-R'_{l1}, 
\end{align}
where the indicators $s_{i,2}$ and $s_{i,3}$ are defined in~\eqref{eq:indicator2} and~\eqref{eq:indicator3}, respectively. 
Then,~\eqref{leq:elementxboundz} can be rewritten as
\begin{align}
\big |R'_{l}-\E[R'_{l}] \big | < \frac{\epsilon}{2} \|\x\|^2|x_{\min}|. 
\end{align}
Since
\begin{align}
|R'_{l}-\E[R'_{l}]| &= \big | R'_{l1}+R'_{l2}-\E[R'_{l1}]-\E[R'_{l2}] \big | \nonumber \\
&\leq  \big | R'_{l1}-\E[R'_{l1}] \big | + \big |R'_{l2}-\E[R'_{l2}] \big |, 
\end{align}
to prove~\eqref{leq:elementxboundz}, it suffices to show that $|R'_{l1}-\E[R'_{l1}]|$ and $|R'_{l2}-\E[R'_{l2}]|$ are both upper bounded by $\frac{\epsilon}{4}\|\x\|^2|x_{\min}|$. Similar to the proof of~\eqref{leq:elementxbound}, we shall bound these two terms separately with probabilistic arguments.

\begin{enumerate}[i)]
\item {\bf Upper bound of $\big | R'_{l1}-\E[R'_{l1}] \big |$:}

Following the analysis in~\eqref{eq:sigma2_160}--\eqref{eq:p_162}, we will also use Lemma~\ref{lm:boundedsubgaussian} to bound the probability 
\begin{equation}
	\P\left(|R'_{l1}-\E[R'_{l1}] | \geq \frac{\epsilon}{4}\|\x\|^2|x_{\min}|\right),
\end{equation}
where the ``$\sigma^2$'' in Lemma~\ref{lm:boundedsubgaussian} are estimated as follows:
\begin{align}
	\sigma^2 &=~\sum_{i=1}^{m}\E \Big [|R_{l1,i}|^2 \Big ]\nonumber\\
	&\leq~\frac{1}{m^2}\sum_{i=1}^{m} \E \Big [|\a_i^*\z|^6 |a_{il}|^2 \Big ]  \nonumber \\
	&\overset{(a)}{\leq}~\frac{1}{m} \sqrt{\E \big [|\a_i^*\z|^{12} \big ]\E \big [|a_{il}|^4 \big ]}\nonumber\\
	&\overset{(b)}{\leq}~ \frac{38}{m}\|\z\|^6.
\end{align}
Here, (a) is from the H{\"o}lder's inequality and (b) is from Lemma~\ref{Lm:ComplexGauMoment} and that the complex Gaussian distribution is invariant to unitary transform. Moreover, it is not difficult to see that $R_{l1,i}$ are mutually independent and that
\begin{equation}
	|R'_{l1,i}| \overset{\eqref{eq:indicator2},\eqref{eq:indicator3}}{<} \frac{1024\log^2n}{m}\|\z\|^3.
\end{equation}

~~Now we can apply Lemma~\ref{lm:boundedsubgaussian} to obtain
\begin{align}
	& \hspace{1.3mm} \P\left(|R'_{l1}-\E[R'_{l1}] | \geq \frac{\epsilon}{4}\|\x\|^2|x_{\min}|\right) \nonumber\\
	& \hspace{-1mm} \leq 4\exp \left( \frac{ - \frac{\epsilon^2}{64}\|\x\|^4|x_{\min}|^2}{\frac{38}{m}\|\z\|^6 + \frac{512\sqrt{2}\epsilon\log^2 n}{3m} \|\x\|^2\|\z\|^3|x_{\min}| }\right) \nonumber\\    
	&  \hspace{-1mm} \overset{(c)}{\leq}  4\exp \left(\frac{- m\|\x\|^4\epsilon^2 |x_{\min}|^2}{3432 \kappa^3\|\x\|^6 \hspace{-.5mm} + \hspace{-.5mm} \frac{32768\sqrt{2}}{3} \epsilon \log^2 n \kappa^{\frac{3}{2}}\|\x\|^5|x_{\min}|} \right) \nonumber\\
	&  \hspace{-1mm} \overset{(d)}{\leq} 4 \exp\left(-\frac{c_1'm}{k + c_2'\sqrt{k}\log^2 n)} \right),
\end{align}
where (c) uses the condition that $\frac{\|\z\|^2}{\|\x\|^2} \leq \kappa$ and (d) is because $|x_{\min}| = \Omega   ({\|\x\|}/{\sqrt{k}}  ).$ Thus,  when $m\geq C_1'\max\{k\log n, \sqrt{k}\log^3n\}$, we have that
\begin{equation}\label{leq:rl1z}
	|R'_{l1}-\E[R'_{l1}] | < \frac{\epsilon}{4}\|\x\|^2|x_{\min}|
\end{equation}
holds with probability exceeding $1-\frac{c_3'}{n^2}$.

\vspace{2mm}
\item {\bf Upper bound of $\big | R'_{l2}-\E[R'_{l2}] \big |$:}

Similar to the derivation in~\eqref{eq:var167}--\eqref{eq:chebysheve169}, we will also use Chebyshev's inequality to bound the probability 
\begin{equation}
	\P\left( \big |R'_{l2}-\E[R'_{l2}] \big | \geq \frac{\epsilon}{4} \|\x\|^2|x_{\min}|\right).
\end{equation} 
The variance of $R'_{l2}$ can be estimated as follows:
\begin{eqnarray}
	&& \hspace{-12mm} \operatorname{Var} \left[R'_{l2}\right] \nonumber \\
	& \hspace{-2mm} = & \hspace{-2mm}  \operatorname{Var}\left[\frac{1}{m}\sum_{i=1}^{m} |\a_i^*\z|^2 (\a_i^*\z) a_{il}(1-s_{i,2}s_{i,3})\right] \nonumber \\ 
	& \hspace{-2mm} \leq  & \hspace{-2mm}  \frac{1}{m^2}\sum_{i=1}^{m}\operatorname{Var} \Big [ |\a_i^*\z|^2 (\a_i^*\z) a_{il}(1-s_{i,2}s_{i,3}) \Big ] \nonumber \\ 
	& \hspace{-2mm} \leq & \hspace{-2mm} 
	\frac{1}{m^2} \sum_{i=1}^{m}\E\Big[|\a_i^*\z|^6 |a_{il}|^2 (1-s_{i,2}s_{i,3})^2 \Big] \nonumber \\ 
	& \hspace{-2mm} \overset{(e)}{\leq} & \hspace{-2mm}    \frac{1}{m^2}\sum_{i=1}^{m} \Big (\left.\E \big[ |\a_i^*\z|^{18}\big] \E \big[|a_{il}|^6\big] \right. \nonumber \\
	& &  \times~ \E \big[ (1-s_{i,2}s_{i,3})^6 \big ] \Big )^{-3} \nonumber \\ 
	& \hspace{-2mm} \overset{(f)}{\leq} & \hspace{-2mm}    \frac{1}{m^2}\sum_{i=1}^{m} \sqrt[3]{\E \big[ |\a_i^*\z|^{18}\big]  \E \big[|a_{il}|^6\big]} \nonumber \\ 
	&    &  
	\times \sqrt[3]{p_{i,2}+p_{i,3}} \nonumber \\
	& \hspace{-2mm} \overset{(g)}{\leq} & \hspace{-2mm}   \frac{164\|\z\|^6}{mn^2},
\end{eqnarray}
where (e) is from the H{\"o}lder's inequality and (g) is from Lemma~\ref{Lm:ComplexGauMoment} and the fact that the complex Gaussian distribution is invariant to unitary transform. (f) is because $1 - s_{i,2}s_{i,3} = 1$ if any of $s_{i,2}$, $s_{i,3}$ is zero and $1 - s_{i,2}s_{i,3} = 0$ if $s_{i,2}=s_{i,3}=1$, and thus,
\begin{align}
	\E \left[(1-s_{i,2}s_{i,3})^6\right]&= \E [1-s_{i,2}s_{i,3}]\nonumber\\
	&=\P\Big( \{ s_{i,2} = 0\} \cup \{ s_{i,3} = 0 \} \Big) \nonumber \\
	&  \leq  p_{i,2}+p_{i,3}.
\end{align}

By Chebyshev's inequality, we obtain
\begin{align}
	\nonumber	\P\left(|R'_{l2}-\E[R'_{l2}] |\hspace{-0.5mm}\geq \hspace{-0.5mm}\frac{\epsilon}{4}\|\x\|^2|x_{\min}|\right) \hspace{-1mm}&\leq \frac{164\|\z\|^6}{mn^2\frac{\epsilon^2}{16}\|\x\|^4|x_{\min}|^2} \nonumber \\
	&\overset{(h)}{\leq}\frac{2624\kappa^3\|\x\|^2}{mn^2\epsilon^2|x_{\min}|^2} \nonumber \\
	&\overset{(i)}{\leq} \frac{c'_4k}{mn^2},
\end{align}
where (h) uses the condition that $\frac{\|\z\|^2}{\|\x\|^2} \leq \kappa$ and (i) uses the condition that $|x_{\min}| = \Omega   ({\|\x\|}/{\sqrt{k}}).$
Therefore, when $m\geq C'_2 k \log n$, it holds that
\begin{equation} \label{leq:rl2z}
	\P\left(|R'_{l2}-\E[R'_{l2}] | < \frac{\epsilon}{4}\|\x\|^2|x_{\min}|\right) > 1- \frac{c'_5}{n^2}.
\end{equation}
\end{enumerate}

Combining~\eqref{leq:rl1z} with~\eqref{leq:rl2z} yields~\eqref{leq:elementxboundz} with probability exceeding $1-\frac{c'_6}{n^2}$.
Furthermore, by taking the union bound for all $l \in \{1,\cdots,n\}$, we further obtain that~\eqref{leq:elementxboundz} holds with probability exceeding $1 - \frac {c'_6}{n}$ when $m\geq C'_3\max\{k\log n, \sqrt{k}\log^3n\}$. 
\end{proof}

Finally, by combining the upper bounds in~\eqref{leq:elementxbound} and~\eqref{leq:elementxboundz}, we conclude that
\begin{equation}
|\nabla_1 f(\z)_{l} - \nabla_1 \mathbb{E}[f(\z)]_{l}| \leq \epsilon \|\x\|^2|x_{\min}|
\end{equation}
holds for all $l \in \{1,\cdots,n\}$ with probability exceeding $1-cn^{-1}$ when $m\geq C\max\{k\log n, \sqrt{k}\log^3n\}$.

\section{Proofs of Propositions~\ref{prop:2.3}--\ref{prop:2.7}}\label{Appendix:C}
\subsection{Proof of {Proposition}~\ref{prop:2.3}}
\begin{proof}
Direct calculation shows that 
\begin{align}
&\hspace{-11mm}  \begin{bmatrix}
	\u \mathrm{e}^{\mathsf{j} \phi(\z)} \\
	\overline{\u} \mathrm{e}^{-\mathsf{j} \phi(\z)}
\end{bmatrix}^{*} \nabla^{2} f(\z)
\begin{bmatrix}
	\u \mathrm{e}^{\mathsf{j} \phi(\z)} \\
	\overline{\u} \mathrm{e}^{-\mathsf{j} \phi(\z)}
\end{bmatrix}  \nonumber \\ 
=~& \frac{1}{m} \sum_{i=1}^m\left(4\left|\a_i^* \z\right|^2\left|\a_i^*\u\right|^2-2\left|\a_i^*\u\right|^4 \right. \nonumber	\\
&+2 \left.\Re\left[\left(\a_i^* \z\right)^2\left(\z^* \a_i\right)^2 \mathrm{e}^{-2 \mathsf{j} \phi(\z)}\right]\right) \nonumber	\\
=~& \frac{1}{m} \sum_{i=1}^m\left(2\left|\a_i^* \z\right|^2\left|\a_i^*\u\right|^2-2\left|\a_i^*\u\right|^4\right)  \nonumber \\
& +\frac 1m \sum_{i=1}^m \left(2\left|\a_i^* \z\right|^2\left|\a_i^*\u\right|^2 \right. \nonumber \\
&+\left. 2 \Re\left[\left(\a_i^* \z\right)^2\left(\z^* \a_i\right)^2 \mathrm{e}^{-2 \mathsf{j} \phi(\z)}\right]\right).
\end{align}

Lemma \ref{lem:6.3} implies that when $m\ge C_1 k\log k$, w.h.p,
\begin{align}
\nonumber	\frac{2}{m} \sum_{i=1}^m\left|\a_i^*\u\right|^2\left|\a_i^* \z\right|^2 \leq~& \mathbb{E}\left[\frac{2}{m} \sum_{i=1}^m\left|\a_i^*\u\right|^2\left|\a_i^* \z\right|^2\right]\\
&+\frac{1}{200} \|\u\|^2\|\z\|^2 .
\end{align}
On the other hand, by  Lemma \ref{lem:A.7}, we have that 
\begin{equation}
\frac2m \sum_{i=1}^m |\a_i^*\u|^4 \ge \mathbb E\left[ \frac2m \sum_{i=1}^m |\a_i^*\u|^4\right] - \frac1{100} \|\u\|^4 
\end{equation} 
holds with probability at least $1 - \exp(-c_2 m)$. Note that $$\|\u\|^2=\frac{\|\x_\T\|^2+\|\x\|^2}{2}\geq \frac{\|\x\|^2}{2},$$ i.e., $\|\x\|^2\leq 2 \|\u\|^2$. Then, for the second summation, we have 
\begin{align}
\nonumber	&\hspace{-8mm} \frac{1}{m} \sum_{i=1}^m\left(2\left|\a_i^* \z\right|^2\left|\a_i^*\u\right|^2+2 \Re\left[\left(\a_i^* \z\right)^2\left(\u^* \a_i\right)^2 \mathrm{e}^{-2 \mathsf{j} \phi(\z)}\right]\right) \\
\nonumber	=~& {\begin{bmatrix}
		\z \\
		\overline{\z}
	\end{bmatrix}^* \nabla^2 f\left(\u \mathrm{e}^{\mathsf{j} \phi(\z)}\right)\begin{bmatrix}
		\z \\
		\overline{\z}
\end{bmatrix} } \\
\leq~& \begin{bmatrix}
	\z \\
	\overline{\z}
\end{bmatrix}^* \nabla^2 \mathbb{E}\left[f\left(\u \mathrm{e}^{\mathsf{j} \phi(\z)}\right)\right]\begin{bmatrix}
	\z \\
	\overline{\z}
\end{bmatrix}+\frac{1}{600}(\|\u\|^2+\|\x\|^2)\|\z\|^2 \nonumber\\
\leq~& \begin{bmatrix}
	\z \\
	\overline{\z}
\end{bmatrix}^* \nabla^2 \mathbb{E}\left[f\left(\u \mathrm{e}^{\mathsf{j} \phi(\z)}\right)\right]\begin{bmatrix}
	\z \\
	\overline{\z}
\end{bmatrix}+\frac{1}{200}\|\u\|^2\|\z\|^2, 
\end{align} 
w.h.p, provided $m\ge C_3 k\log k$, according to Lemma \ref{lm:matrixconcentration}.

Collecting the above estimates, we have that when $m\ge C_4 k\log k$ for a sufficiently large constant $C_4$, w.h.p,
\begin{align}
\nonumber	& \hspace{-10mm} \begin{bmatrix}
	\u \mathrm{e}^{\mathsf{j} \phi(\z)} \\
	\overline{\u} \mathrm{e}^{-\mathsf{j} \phi(\z)}
\end{bmatrix}^{*} \nabla^{2} f(\z)\begin{bmatrix}
	\u \mathrm{e}^{\mathsf{j} \phi(\z)} \\
	\overline{\u} \mathrm{e}^{-\mathsf{j} \phi(\z)}
\end{bmatrix} \\
\nonumber	\le~& \mathbb E\left[\begin{bmatrix}
	\u \mathrm{e}^{\mathsf{j} \phi(\z)} \\
	\overline{\u} \mathrm{e}^{-\mathsf{j} \phi(\z)}
\end{bmatrix}^{*} \nabla^{2} f(\z)\begin{bmatrix}
	\u \mathrm{e}^{\mathsf{j} \phi(\z)} \\
	\overline{\u} \mathrm{e}^{-\mathsf{j} \phi(\z)}
\end{bmatrix}\right]\\
&\nonumber + \frac1{100}\|\u\|^2 \|\z\|^2  + \frac1{100}\|\u\|^4\\
\le~& -\frac1{100}\|\u\|^4
\end{align}
for all $\z \in \mathcal R_1$.
\end{proof}

\subsection{Proof of {Proposition}~\ref{prop:2.4}}
\begin{proof}
Denote $a = \frac{\|\x_\T\|}{\|\x\|}$, where $a^2 > \frac{9}{10}$.
\begin{claim}
It is enough to prove that for all unit vector $\g\in \mathcal{S} \doteq \{\h\in\C^\T: \Im[\h^*\u] = 0, \|\h\| = 1\}$ and all $t\in [0, \frac{\|\x\|}{\sqrt 7} ]$, the following holds:
\begin{equation}
	\begin{bmatrix}
		\g\\
		\overline{\g}
	\end{bmatrix}^{*} \nabla^{2} f(\u+t \g)\begin{bmatrix}
		\g\\
		\overline{\g}
	\end{bmatrix} >0.11\|\x\|^2. \label{gfg}
\end{equation}
\end{claim}
\begin{proof}
First, when $\operatorname{dist}(\z, \u) \not = 0$, let $\g = \mathrm{e}^{-\mathsf{j}\phi(\z)}\g(\z)$. Then, 
\begin{equation}
	\g^*\u = \mathrm{e}^{\mathsf{j}\phi(\z)}\frac{\z^*\u - \mathrm{e}^{-\mathsf{j}\phi(\z)}\|\u\|^2}{\|\z - \u \mathrm{e}^{\mathsf{j}\phi(\z)}\|} = \frac{|\z^*\u| - \|\u\|^2}{\|\z - \u \mathrm{e}^{\mathsf{j}\phi(\z)}\|}.
\end{equation}
Therefore, $\Im[\g^*\u] = 0$. Since  
\begin{equation}
	\begin{bmatrix}
		\g\\
		\overline{\g}
	\end{bmatrix}^{*} \nabla^{2} f(\u+t \g)\begin{bmatrix}
		\g\\
		\overline{\g}
	\end{bmatrix}= \begin{bmatrix}
		\g(\z)\\
		\overline{\g(\z)}
	\end{bmatrix}^* \nabla^2 f(\z)
	\begin{bmatrix}
		\g(\z)\\
		\overline{\g(\z)}
	\end{bmatrix}
\end{equation}
and $\|\g\|=1$, it suffice to show~\eqref{gfg} for $\g\in \mathcal{S}$. When $\operatorname{dist}(\z,\u)=0$, note that $\g\in\mathcal S$. Thus, we proved the claim.
\end{proof}
Using the claim above, we now consider the equation~\eqref{gfg}. Direct calculation shows
\begin{align}
\nonumber	&\begin{bmatrix}\g\\ \overline{\g}\end{bmatrix}^{*} \nabla^{2} f(\u+t \g)\begin{bmatrix}\g\\ \overline{\g}\end{bmatrix} \\
\nonumber	=~& \frac{1}{m} \sum_{i=1}^{m} 4\left|\a_{i}^{*}(\u+t \g)\right|^{2}\left|\a_{i}^{*} \g\right|^{2}-2\left|\a_{i}^{*} \x\right|^{2}\left|\a_{i}^{*} \g\right|^{2}\\
\nonumber	&+2 \Re\left[\left(t \a_{i}^{*} \g+\a_{i}^{*} \u\right)^{2}\left(\g^{*} \a_{i}\right)^{2}\right] \\
\nonumber	\geq~& \frac{1}{m} \sum_{i=1}^{m} 4\left|\a_{i}^{*}(\u+t \g)\right|^{2}\left|\a_{i}^{*} \g\right|^{2}-2\left|\a_{i}^{*} \x\right|^{2}\left|\a_{i}^{*} \g\right|^{2}\\
\nonumber	&+4\left[\Re\left(t \a_{i}^{*} \g+\a_{i}^{*} \u\right)\left(\g^{*} \a_{i}\right)\right]^{2}-2\left|\left(t \a_{i}^{*} \g+\a_{i}^{*} \u\right)\left(\g^{*} \a_{i}\right)\right|^{2} \\
\nonumber	\geq~& \frac{1}{m} \sum_{i=1}^{m} 2\left|\a_{i}^{*}(\u+t \g)\right|^{2}\left|\a_{i}^{*} \g\right|^{2}-2\left|\a_{i}^{*} \x\right|^{2}\left|\a_{i}^{*} \g\right|^{2}\\
&+4\left[\Re\left(t \a_{i}^{*} \g+\a_{i}^{*} \u\right)\left(\g^{*} \a_{i}\right)\right]^{2}.
\end{align}

Notice that $\supp(\g)\subseteq\T$. Lemma \ref{lem:6.4} implies when $m\ge C_1 k\log k$ for sufficiently large constant $C_1$, w.h.p.,
\begin{align}
\nonumber	& \hspace{-10mm} \frac{1}{m} \sum_{i=1}^{m} 2\left|\a_{i}^{*}(\u+t \g)\right|^{2}\left|\a_{i}^{*} \g\right|^{2}\\ 
\geq~&\frac{199}{100}(\left|(\u+t \g)^{*} \g\right|^{2}+\|\u+t \g\|^{2}\|\g\|^{2}).
\end{align}
Lemma \ref{lem:6.3} implies that when $m\ge C_2 k\log k$ for sufficiently large constant $C_2$, w.h.p.,
\begin{equation}
\frac{1}{m} \sum_{i=1}^{m} 2\left|\a_{i}^{*} \x\right|^{2}\left|\a_{i}^{*} \g\right|^{2} \leq \frac{201}{100}(\left|\x^{*} \g\right|^{2}+\|\x\|^{2}\|\g\|^{2})
\end{equation}
for all $\g\in \mathcal{S} $. Moreover, Lemma \ref{lem:6.4} implies when $m\ge C_3 k\log k$ for sufficiently large constant $C_3$, w.h.p.,
\begin{align}
\nonumber	& \hspace{-10mm} \frac{4}{m} \sum_{i=1}^{m}\left[\Re\left(t \a_{i}^{*} \g+\a_{i}^{*} \u\right)\left(\g^{*} \a_{i}\right)\right]^{2} \\
\geq~&\frac{398}{100} \left ( \frac12 \|\u+t \g\|^{2}\|\g\|^{2}+\frac32\left|(\u+t\g)^{*} \g\right|^{2} \right )
\end{align}
for all $\g \in \mathcal{S}$, where we have used that $\Im{(\g^*\x)} = 0\Rightarrow \Im{(\x + \g) ^* \g} = 0$ to simplify the results.

Collecting the above estimates, we obtain that when $m\ge C_4 k\log k$, w.h.p.,
\begin{align}
\nonumber	&\begin{bmatrix}\g\\ \overline{\g}\end{bmatrix}^{*} \nabla^{2} f(\u+t \g)\begin{bmatrix}\g\\ \overline{\g}\end{bmatrix} \\
\nonumber	\geq~&\frac{199}{100}(|(\u + t\g)^* \g|^ 2 + \|\u + t\g\|^2) \\
\nonumber	&- \frac{201}{100}(|\x^* \g| + \|\x\|^2\|\g\|^2) \\
\nonumber	&+\frac{199}{100}(\|\u + t\g\|^2 + 3|(\u + t\g )^* \g|^2)\\
\nonumber	\geq~& \frac{199}{25}|\u^*\g + t|^2 + \frac{199}{50}\|\u + t\g\|^2-\frac{201}{100}(|\x^*\g|^2 + \|\x\|^2)\\
=~&\frac{199}{25}|\u^*\g + t|^2 + \frac{199}{50}\|\u + t\g\|^2-\frac{201}{100}(|\x_\T^*\g|^2 + \|\x\|^2).
\end{align}
To provide a lower bound for the above, we let $$\Re{(\x_\T^* \g)} = \x_\T^* \g = \lambda \|\x_\T\|$$ with $\lambda \in [-1, 1]$ and $t = \eta \|\x\|$ with $\eta \in [0, 1/\sqrt 7]$. Then 
\begin{align}
\nonumber	& \hspace{-10mm} \begin{bmatrix}\g\\ \overline{\g}\end{bmatrix}^{*} \nabla^{2} f(\u+t \g)\begin{bmatrix}\g\\ \overline{\g}\end{bmatrix} \\
\nonumber	=~&\|\x\|^2\left[\frac{597}{50}\eta^2+\frac{597}{25}\sqrt{\frac{1+a^2}2}\eta\lambda + (\frac{199}{50}(1+a^2)\right.\\
\nonumber	&\left.-\frac{201}{100}a^2)\lambda^2+ \frac{199}{50}\frac{1+a^2}2-\frac{201}{100}\right]\\
\doteq~ &\phi(\lambda, \eta) \|\x\|^2.
\end{align}

For any fixed $\eta$, it is easy to see that minimizer occur when 
\begin{equation}
\lambda^* = -3\sqrt2 \times  \frac{\sqrt{1+a^2}}{2+\frac{197}{199}a^2}\eta.
\end{equation}
Plugging the minimizer into $\phi(\lambda, \eta)$, one obtains
\begin{align}
\nonumber	\phi(\lambda^*, \eta) &= \frac{597}{50}\eta^2 \frac{-1-\frac{400}{199}a^2}{2+\frac{197}{199}a^2}+\frac{199}{50}\frac{1+a^2}{2}-\frac{201}{100}\\
\	&\ge \frac{597}{350} \frac{-1-\frac{400}{199}a^2}{2+\frac{197}{199}a^2}+\frac{199}{50}\frac{1+a^2}{2}-\frac{201}{100}.\label{leq:philambdaeta}
\end{align}

It is easy to verify that~\eqref{leq:philambdaeta} increases as $a^2$ increases. Thus $$\left.\phi(\lambda^*, \eta)\right|_{a^2=0.9}  \ge 0.11.$$ Finally, we obtain
\begin{equation}
\begin{bmatrix}\g\\ \overline{\g}\end{bmatrix}^{*} \nabla^{2} f(\u+t \g)\begin{bmatrix}\g\\ \overline{\g}\end{bmatrix} \geq \frac{11}{100}\|\x\|^2 ,
\end{equation}
as claimed.
\end{proof}

\subsection{Proof of {Proposition}~\ref{prop:2.5}}
\begin{proof}
Note that 
\begin{equation}
\z^* \nabla_1 f(\z) = \frac1m \sum_{i = 1}^m |\a_i^* \z|^4 - \frac1m \sum_{i = 1}^m |\a_i^* \x|^2 |\a_i^* \z|^2.
\end{equation}
By Lemma \ref{lem:6.4}, when $m\ge C_1 k\log k$ for some sufficiently large $C_1$, w.h.p.,
\begin{equation}
\frac1m \sum_{i=1}^m |\a_i^* \z|^4 \ge \E \left [\frac1m \sum_{i=1}^m |\a_i^* \z|^4 \right ] -\frac{1}{100} \|\z\|^4
\end{equation}
for all $\z \in \mathbb{C}^\T$. On the other hand, Lemma \ref{lem:6.3} implies that when $m\ge C_2k\log k$ for some sufficiently large $C_2$, w.h.p.,
\begin{flalign}
\nonumber&\frac1m \sum_{i=1}^m |\a_i^* \x|^2 |\a_i^* \z|^2\nonumber \\
\le~ & \E \left [\frac1m \sum_{i=1}^m |\a_i^*\x|^2|\a_i^* \z|^2 \right] + \frac1{1000}\|\x\|^2 \|\z\|^2.
\end{flalign}

Then, w.h.p., it holds that
\begin{equation}
\z^* \nabla_1 f(\z) \ge \frac{1}{1000}\|\x\|^2\|\z\|^2.
\end{equation}
This completes the proof. 
\end{proof}

\subsection{Proof of {Proposition}~\ref{prop:2.6}}
\begin{proof}
We abbreviate $\phi(\z)$ as $\phi$ and denote $a = \frac{\|\x_\T\|}{\|\x\|}$, where $a^2 > \frac{9}{10}$. Note that
\begin{align}
\nonumber	(\z - \u \mathrm{e}^{\mathsf{j}\phi })^* \nabla_1 f(\z) =& \frac1m \sum_{i=1}^m |\a_i^* \z|^2 (\z - \u \mathrm{e}^{\mathsf{j}\phi})^* \a_i\a_i^*\z \\
&- \frac1m \sum_{i=1}^m |\a_i^* \x|^2 (\z - \u \mathrm{e}^{\mathsf{j}\phi})^* \a_i\a_i^* \z.
\end{align}

We first bound the second term. By Lemma \ref{lem:6.3}, when $m\ge C_1 k\log k$ for a sufficiently large constant $C_1$, w.h.p., for all $\z\in \mathbb{C}^\T$,
\begin{align}
\nonumber	&\Re\left[\frac1m \sum_{i=1}^m |\a_i^* \x|^2 (\z - \u \mathrm{e}^{\mathsf{j}\phi})^*\a_i\a_i^* \z \right]\\
\nonumber	=~&\Re\left[(\z-\u \mathrm{e}^{\mathsf{j}\phi})^* \E\left[\frac1m \sum_{i=1}^m |\a_i^* \x|^2 \a_i\a_i^*\right]\z\right] \\
\nonumber	&+ \Re\left[(\z - \u \mathrm{e}^{\mathsf{j}\phi})^* \Delta \z\right]~~~\text{(where $\|\Delta\|\le \delta \|\x\|^2$)} \\
\nonumber	\le~& \Re\left[(\z-\u \mathrm{e}^{\mathsf{j}\phi})^* \E\left[\frac1m \sum_{i=1}^m |\a_i^* \x|^2 \a_i\a_i^*\right]\z \right]\\
&+ \frac1{1000} \|\x\|^2 \|\z - \u \mathrm{e}^{\mathsf{j}\phi}\|\|\z\|.
\end{align}

To bound the first term, for a fixed $\tau$ to be determined later, define:
\begin{align}
{S(\z)} & \doteq \frac1m \sum_{i=1}^m |\a_i^* \z|^2 \Re\left[(\z - \u \mathrm{e}^{\mathsf{j}\phi})^* \a_i\a_i^* \z \right],\\
S_1(\z) & \doteq \frac1m \sum_{i=1}^m \left[|\a_i^* \z|^2 \Re\left[(\z - \u \mathrm{e}^{\mathsf{j}\phi})^* \a_i\a_i^* \z\right]\bf{1}_{\{|\a_\mathit{i}^*\u|\le \tau\}}\right],\\
\nonumber	S_2(\z) & \doteq \frac1m \sum_{i=1}^m\left[ |\a_i^* \z|^2 \Re\left[(\z - \u \mathrm{e}^{\mathsf{j}\phi})^* \a_i\a_i^* \z\right]\right.\\
&\left.\hspace{30mm}\bf{1}_{\{|\a_\mathit{i}^*\u|\le \tau\}}\bf{1}_{\{|\a_\mathit{i}^* \z|\le \tau\}}\right].
\end{align}

Obviously, $S_1(\z) \ge S_2(\z)$ for all $\z$ as 
\begin{align}
\nonumber	&S_1(\z) - S_2(\z)\\
\nonumber	=~& \frac1m \sum_{i=1}^m |\a_i^* \z|^2 \Re\left[(\z - \u \mathrm{e}^{\mathsf{j}\phi})^* \a_i\a_i^* \z\right]\bf{1}_{\{|\a_\mathit{i}^*\u|\le \tau\}}\bf{1}_{\{|\a_\mathit{i}^*\z|>\tau\}}\\
\nonumber	\ge~& \frac1m \sum_{i=1}^m |\a_i^* \z|^2 (|\a_i^* \z|^2 - |\a_i^*\u| |\a_i^* \z|)\bf{1}_{\{|\a_\mathit{i}^*\u|\le \tau\}}\bf{1}_{\{|\a_\mathit{i}^*\z|>\tau\}}\\
\ge~& 0.
\end{align}

Now for an $\epsilon \in (0, \|\u\|)$ to be fixed later, consider an $\epsilon$-net $N_\epsilon$ for the ball $\mathbb{C}\mathbb{B}^k(\|\u\|)$, with 
\begin{equation}
    |N_\epsilon| \le \left(\frac{3\|\u\|}{\epsilon}\right)^{2k}.
\end{equation}
On the complement of the event $\left\{\max_{i\in [m]}|\a_i^*\u|> \tau\right\}$, we have for any $t>0$ that 
\begin{align}
\nonumber	& \P \big (S(\z) - \E[S(\z)] < -t, \exists \z \in N_\epsilon \big )\\
\nonumber	\le~& |N_\epsilon| \cdot \P \big (S(\z) - \E[S(\z)]<-t \big )\\
\le~& |N_\epsilon| \P(S_1(\z) - \E[S_1(\z)]<-t + |\E[S_1(\z)] - \E[S(\z)]|).
\end{align}
Because $S_1(\z)\ge S_2(\z)$ as shown above
\begin{align}
\nonumber	&\P(S_1(\z) - \E[S_1(\z)]<-t + |\E[S_1(\z)] - \E[S(\z)]|)\\
\nonumber	\le~& \P\left(S_2(\z) - \E[S_2(\z)]<-t + |\E[S_1(\z)]-  \E[S(\z)]| \right.\\
&\hspace{40mm}\left. + |\E[S_1(\z)] - \E[S_2(\z)]|\right).
\end{align}

Thus, the unconditional probability can be bounded as:
\begin{align}
\nonumber	&\P\left(S(\z) - \E[S(\z)]<-t, \exists \z\in N_\epsilon\right)\\
\nonumber	 \le~& |N_\epsilon| \cdot \P \Big ( \left.S_2(\z) - \E[S_2(\z)]  -t + \big |\E[S_1(\z)] - \E[S(\z)] \big |  \right. \nonumber \\
& \left. +~\big | \E[S_1(\z)] - \E[S_2(\z)] \big | \right.\Big ) + \P\left(\max_{i\in [m]} |\a_i^*\u| > \tau\right).
\end{align}

Taking $\tau = \sqrt{10 \log m}\|\u\|$, we obtain
\begin{equation}
\P\left(\max_{i\in [m]}|\a_i^*\u|>\tau\right)\le m\exp\left(-\frac{10\log m}{2}\right) = m^{-4},
\end{equation}
\begin{align}
\nonumber	&\hspace{-10mm} |\E[S_1(\z)] - \E[S(\z)]|\\
\nonumber	\le~& \sqrt{\E[|\a_i^* \z|^6 |\a_i^* (\z-\u \mathrm{e}^{\mathsf{j}\phi})|^2]}\sqrt{\P(|\a_i^* \u|>\tau)}\\
\nonumber	\le~& \sqrt{\E_{Z\sim \mathcal{CN}}\left[|Z|^8\right]}\sqrt{\P_{Z\sim \mathcal{CN}}(|Z|> \sqrt{10\log m})}\\
\nonumber&\hspace{10mm}\cdot\|\z\|^3 \|\z - \u \mathrm{e}^{\mathsf{j}\phi}\|\\
\le~& 2\sqrt{6}m^{-\frac52}\|\z\|^3 \|\z - \u \mathrm{e}^{\mathsf{j}\phi}\|,
\end{align}
and
\begin{align}\label{leq:Es1-Es2}
\nonumber	&\hspace{-10mm} |\E[S_1(\z)] - \E[S_2(\z)]|\\
\nonumber	\le~& \sqrt{\E[|\a_i^* \z|^6 |\a_i^* (\z-\u \mathrm{e}^{\mathsf{j}\phi})|^2\bf{1}_{|\a_\mathit{i}^* \u|\le \tau}]}\sqrt{\P(|\a_\mathit{i}^* \z|>\tau)}\\
\le~& 2\sqrt{6}m^{-\frac52}\|\z\|^3 \|\z - \u \mathrm{e}^{\mathsf{j}\phi}\|,
\end{align}
where we have used $\|\z\|\le \|\u\|$ to simplify~\eqref{leq:Es1-Es2}. Now we use the moment-control Bernstein's inequality (Lemma \ref{lem:A.8}) to get a bound for probability on deviation of $S_2(\z)$. To this end, we have
\begin{align}
\nonumber	&\E\left[|\a_i^* \z|^6 |\a_i^* (\z - \u \mathrm{e}^{\mathsf{j}\phi})|^2 \bf{1}_{\{|\a_\mathit{i}^* \u|\le \tau\}}\bf{1}_{\{|\a_\mathit{i}^* \z|\le \tau\}}\right]\\
\nonumber	&\le~ \tau^4 \E[|\a_i^* \z|^2 |\a_i^* (\z - \u \mathrm{e}^{\mathsf{j}\phi})|^2]\\
\nonumber	&\le~ \tau^4 \E_{Z\sim \mathcal{CN}}[|Z|^4]    \|\z\|^2\|\z - \u \mathrm{e}^{\mathsf{j}\phi}\|^2\\
\nonumber	&=~ 200\log^2 m \|\u\|^4 \|\z\|^2\|\z - \u \mathrm{e}^{\mathsf{j}\phi}\|^2\\
&\overset{(a)}{\leq}~ 200\log^2 m \|\x\|^4 \|\z\|^2\|\z - \u \mathrm{e}^{\mathsf{j}\phi}\|^2
\end{align}
and
\begin{align}
\nonumber &\E\left[|\a_i^* \z|^{3p} |\a_i^* (\z - \u \mathrm{e}^{\mathsf{j}\phi})|^p \bf{1}_{\{|\a_\mathit{i}^* \u|\le \tau\}}\bf{1}_{\{|\a_\mathit{i}^* z|\le \tau\}}\right]\\
\nonumber	&\le~ \tau^{2p}\E[|\a_i^* \z|^p |\a_i^* (\z - \u \mathrm{e}^{\mathsf{j}\phi})|^p]\\
\nonumber	&\le~ (10\log m\|\u\|^2)^p p! \|\z\|^p \|\z - \u \mathrm{e}^{\mathsf{j}\phi}\|^p\\
&\overset{(b)}{\leq}~ (10\log m\|\x\|^2)^p p! \|\z\|^p \|\z - \u \mathrm{e}^{\mathsf{j}\phi}\|^p,
\end{align}
where (a) and (b) use the fact that $$\|\u\|=\sqrt{\frac{\|\x_\T\|^2+\|\x\|^2}{2}}\leq \|\x\|,$$ and the second inequality holds for any integer $p\ge 3$. 

Hence, one can take 
\begin{align}
\sigma^2 &= 200 \log^2 m \|\x\|^4\|\z\|^2 \|\z - \u \mathrm{e}^{\mathsf{j}\phi}\|^2,\\
R &= 10 \log m \|\x\|^2 \|\z\|\|\z - \u \mathrm{e}^{\mathsf{j}\phi}\|
\end{align}
in Lemma \ref{lem:A.8}, and 
\begin{equation}
t = \frac{1}{1000} \|\x\|^2 \|\z\| \|\z - \u \mathrm{e}^{\mathsf{j}\phi}\|
\end{equation}
in the deviation inequality of $S_2(\z)$ to obtain
\begin{align}
\nonumber	&\P \Big (S_2(\z) - \E[S_2(\z)]\\
\nonumber	&<-t +  \big |\E[S_1(\z)] - \E[S(\z)] \big | + \big |\E[S_1(\z)] - \E[S_2(\z)] \big | \Big )\\
\nonumber	\le~& \P\left(S_2(\z) - \E[S_2(\z)]<-\frac{1}{2000} \|\x\|^2 \|\z\| \|\z - \u \mathrm{e}^{\mathsf{j}\phi}\|\right)\\\le~&  \exp\left(-\frac{c_1 m}{\log^2 m }\right),
\end{align} 
where we have used the fact $\|\z\|\le \|\u\|$ and assume $$2\sqrt 6 m^{-5/2}\le \frac{1}{4000}$$ to simplify the probability.

Thus, with probability at least $$1 - m^{-4} - \exp(-c_1 m / \log ^2 m + 2 k \log (3\|\u \|/\epsilon)),$$ it holds that 
\begin{equation}
S(\z)\ge \E[S(\z)] - \frac{1}{1000}\|\x\|^2 \|\z\| \|\z - \u \mathrm{e}^{\mathsf{j}\phi}\|, \quad \forall \z \in N_\epsilon.
\end{equation}

Moreover, for any $\z, \z' \in \mathcal{R}_2^{\h}$, we have
\begin{align}
\nonumber	&|S(\z) - S(\z')|\\
\nonumber	=~ &\bigg  | \frac1m \sum_{i=1}^m |\a_i^* \z|^2 \Re\left[(\z - \u \mathrm{e}^{\mathsf{j}\phi})^* \a_i\a_i^* \z\right] \\
\nonumber	&- \frac1m \sum_{i=1}^m |\a_i^* \z'|^2 \Re\left[(\z' - \u \mathrm{e}^{\mathsf{j}\phi})^*\a_i\a_i^* \z'\right] \bigg | \\
\nonumber	\le~& \frac 1m\sum_{i=1}^m \left||\a_i^* \z|^2 - |\a_i^* \z'|^2 \right||\h(\z)^* \a_i\a_i^* \z| \\
&+ \frac1m \sum_{i=1}^m |\a_i^* \z'|^2 \Big | \h(\z)^* \a_i\a_i^* \z - \h(\z')\a_i\a_i^* \z' \Big |.
\end{align}
For the first term, we have
\begin{align}
\nonumber	& \frac 1m\sum_{i=1}^m ||\a_i^* \z|^2 - |\a_i^* \z'|^2| |\h(\z)^* \a_i\a_i^* \z|\\
\nonumber	=~& \frac1m \sum_{i=1}^m (|\a_i^* \z| + |\a_i^* \z'|)||\a_i^* \z| - |\a_i^* \z'|| |\h(\z)^* \a_i\a_i^* \z|\\
\nonumber	\le~&\frac1m \sum_{i=1}^m (|\a_i^* \z| + |\a_i^* \z'|)||\a_i^* (\z - \z')|| |\h(\z)^* \a_i\a_i^* \z|\\
\nonumber	\le~&\frac1m \sum_{i=1}^m 2\|(\a_i)_\T\|^4 \max\{\|\z\|, \|\z'\|\} \|\z - \z'\|(\|\z\|^2 + \|\z\|\|\u\|)\\
\nonumber	\le~& \frac 4 m \sum_{i=1}^m \|(\a_i)_\T\|^4 \|\u\|^3\|\z - \z'\|\\
\le~&4\max_{i\in [m]}\|(\a_i)_\T\|^4 \|\u\|^3\|\z - \z'\|.
\end{align}
For the second term, we have
\begin{align}
\nonumber	&\frac1m \sum_{i=1}^m |\a_i^* \z'|^2 |\h(\z)^* \a_i\a_i^* \z - \h(\z')^* \a_i\a_i^* \z'|\\
\nonumber	\le~& \frac1m \sum_{i=1}^m \|(\a_i)_\T\|^2 \|\u\|^2 (|\z^*\a_i\a_i^* \z - (\z')^*\a_i\a_i^* \z'|\\
& + |\u^* \mathrm{e}^{-\mathsf{j}\phi}\a_i\a_i^*\z - \u^*  \mathrm{e}^{-\mathsf{j}\phi'}\a_i\a_i^* \z'|),
\end{align}
where 
\begin{equation}
|\z^*\a_i\a_i^* \z - (\z')^*\a_i\a_i^* \z'| \le 2\|(\a_i)_\T\|^2 \|\u\|\|\z - \z'\|.
\end{equation}

Moreover,
\begin{align}
\nonumber	 & |\u^* \mathrm{e}^{-\mathsf{j}\phi}\a_i\a_i^*\z - \u^* \mathrm{e}^{-\mathsf{j}\phi'}\a_i\a_i^* \z'|\\
\nonumber	=~&|\a_i^*\z(\u \mathrm{e}^{\mathsf{j}\phi})^* \a_i - \a_i^* \z'(\u \mathrm{e}^{\mathsf{j}\phi'})^* \a_i|\\
\nonumber	\le~&\|(\a_i)_\T\|^2 \|(\z \mathrm{e}^{-\mathsf{j}\phi} - \z' \mathrm{e}^{-\mathsf{j}\phi'})\u^* \|\\
\nonumber	=~& \|(\a_i)_\T\|^2 |\u^* (\z \mathrm{e}^{-\mathsf{j}\phi} - \z' \mathrm{e}^{-\mathsf{j}\phi'})|\\
\nonumber	=~& \|(\a_i)_\T\|^2 | (|\u^* \z| - |\u^* \z'|)|\\
\le~& \|(\a_i)_\T\|^2 \|\u\|\|\z - \z'\|.
\end{align}
Thus, we have
\begin{align}
\nonumber	&\frac1m \sum_{i=1}^m |\a_i^* \z'|^2 |\h(\z)^* \a_i\a_i^* \z - \h(\z')^* \a_i\a_i^* \z'| \\
\le~& 3 \max_{i\in[m]}\|(\a_i)_\T\|^4 \|\u\|^3 \|\z -\z'\|
\end{align}
and 
\begin{align}
\nonumber	|S(\z) - S(\z')| &\le 7 \max_{i\in[m]}\|(\a_i)_\T\|^4 \|\u\|^3 \|\z -\z'\|\\
&\le 70 k^2 \log ^2 m \|\u\|^3 \|\z - \z'\|,
\end{align}
since $$\max_{i\in [m]} \|(\a_i)_\T\|^4\le 10 k^2 \log ^2 m$$ with probability at least $1 - c_2 m ^{-4}$. Note that every $\z \in \mathcal{R}_2^\h$ can be written as $$\z = \z' + \e,$$ with $\|\e\|\le \epsilon$ and $\z' \in N_\epsilon$.  Therefore,
\begin{align}
\nonumber	S(\z) &\ge S(\z') - 70 n^2 \log^2 m \|\u\|^3 \epsilon\\
\nonumber	&\ge 2\|\z'\|^4 - 2\|\z'\|^2 |\u^* \z'| - \frac1{1000} \|\u\|^2 \|\z'\|\|\z' - \u \mathrm{e}^{\mathsf{j}\phi}\| \\
\nonumber	&\hspace{10mm}- 70 k^2 \log ^2 m \|\u \|^3 \epsilon\\
\nonumber	&\ge \E[S(\z)] - \frac1{1000} \|\u\|^2 \|\z\| \|\z- \u \mathrm{e}^{\mathsf{j}\phi}\| \\
&\hspace{10mm}- 11\epsilon \|\u\|^3 - 70 k^2 \log^2 m \|\u\|^3 \epsilon,
\end{align}
where the additional $11\epsilon \|\u\|^3$ term in the third line is to account for the change from $\z'$ to $\z$, which has been simplified by assumptions (i.e., $\z\in\R_2^{\h}$) that $$\frac{11}{20}\|\u\|\le \|\z\|\le \|\u\|$$ and that $\epsilon \le \|\u\|$.

Choosing $$\epsilon = \frac{\|\u\|}{c_3 k^2 \log^2 m }$$ for a sufficiently large $c_3 > 0 $, and additionally using $$\operatorname{dist}(\z, \u)\ge \frac{1}{3} \|\x\| = \frac{1}{3}\sqrt{2/(1+a^2)}\|\u\|$$ since $\z\in\R_2^{\h}$ (i.e., $\|\u\|\leq \frac{3}{\sqrt{2/(1+a^2)}}\|\z- \u \mathrm{e}^{\mathsf{j}\phi}\|$), we obtain that 
\begin{align}
\nonumber	S(\z)\ge~& \E[S(\z)] - \frac 1{500} \|\u\|^2 \|\z\|\|\z - \u \mathrm{e}^{\mathsf{j}\phi}\|\\
\ge~& \E[S(\z)] - \frac 1{500} \|\x\|^2 \|\z\|\|\z - \u \mathrm{e}^{\mathsf{j}\phi}\|
\end{align} 
for all $\z \in \mathcal{R}_2^\h$, with probability at least $$1 - c_4 m^{-1} - c_5 \exp\left(-c_1\frac{m}{\log^2 m} + c_7 k\log (c_6 k\log m)\right).$$

Combining the above estimates, when $m\ge Ck\log ^3 k $ for sufficiently large constant $c$, w.h.p.,
\begin{equation}
\Re\left[(\z - \u \mathrm{e}^{\mathsf{j}\phi})^* \nabla_1 f(\z)\right]\ge \frac1{1000} \|\x \|^2 \|\z\|\|\z - \u \mathrm{e}^{\mathsf{j}\phi}\|>0
\end{equation}
for all $\z \in \mathcal{R}_2 ^\h$ as claimed.
\end{proof}

\subsection{Proof of {Proposition}~\ref{prop:2.7}}
\begin{figure*}[t]
\normalsize
\begin{eqnarray}
	\mathcal{R}_a& \doteq& \left\{\z\in\C^\T,|\x_\T^* \z|\le \frac12\|\x_\T\| \|\z\|\right\},\label{area:a}\\
	\mathcal{R}_b &\doteq& \left\{\z\in\C^\T,|\x_\T^* \z|\ge \frac12\|\x_\T\|\|\z\|, \|\z\|\le \frac{56}{100}\|\x\|\right\},\label{area:b}\\
	\mathcal{R}_c &\doteq& \left\{\z\in\C^\T,\frac12\|\x_\T\|\|\z\|\le |\x_\T^* \z|\le \frac{99}{100}\|\x_\T\|\|\z\|, \|\z\|\ge \frac{11}{20}\|\x\|\right\},\label{area:c}\\
	\mathcal{R}_d &\doteq& \left\{\z\in\C^\T,\frac{99}{100}\|\x_\T\|\|\z\|\le |\x_\T^* \z|\le \|\x_\T\|\|\z\|, \|\z\|\ge \frac{11}{20}\|\x\|\right\},\label{area:d}\\
	\mathcal{R}_2^{\h '} &\doteq& \left\{\z\in\C^\T, \Re[\langle \h(\z), \nabla_1\E[f(\z)]\rangle]\ge \frac1{250}\|\x\|^2\|\z\|\|\h(\z)\|,  \|\z\|\le \|\u\|\right\}.\label{area:r2hrelax}
\end{eqnarray}
\hrulefill
\vspace*{4pt}
\end{figure*}
It is difficult to directly show that $\C^\T = \R_1\cup\R_2^\z\cup\R_2^\h\cup\R_3$. We follow the methods in~\cite{sunju} and construct four new area $\R_a$, $\R_b$, $\R_c$, and $\R_d$ (See~\eqref{area:a}-\eqref{area:d}). Besides, we construct a relaxed area $\mathcal{R}_2^{\h '}$, as specified in~\eqref{area:r2hrelax}.

Before showing the proposition, we further explore some detailed conditions and equivalence on the areas $\R_1$, $\R_2^\z$, and $\R_2^{\h'}$. These will be beneficial to our proceeding proof. Here, we denote $a = \frac{\|\x_\T\|}{\|\x\|}$, and consider $a^2 \geq \frac{9}{10}$.
\begin{itemize}
\item For $\mathcal R_1$, notice that $\u = \omega_\T \x_\T$, then 
\begin{align}
\nonumber	&\begin{bmatrix}
	\u \mathrm{e}^{\mathsf{j} \phi(\z)} \\
	\overline{\u} \mathrm{e}^{-\mathsf{j} \phi(\z)}
\end{bmatrix}^{*} \mathbb{E}\left[\nabla^{2} f(\z)\right]\begin{bmatrix}
	\u \mathrm{e}^{\mathsf{j} \phi(\z)} \\
	\overline{\u} \mathrm{e}^{-\mathsf{j} \phi(\z)}
\end{bmatrix}\\
\nonumber	=~& 8|\u^* \z|^2 - 2|\u^* \x|^2 + 4\|\z\|^2 \|\u\|^2 - 2\|\x\|^2 \|\u\|^2\\
\nonumber	=~& |\omega_\T|^2\left(8|\x_\T^* \z|^2 -2\|\x_\T\|^4 \right.\\
&\left.+4\|\z\|^2 \|\x_\T\|^2 - 2\|\x\|^2 \|\x_\T\|^2\right)
\end{align} 
and 
\begin{align}
\nonumber	&-\frac1{100}\|\u\|^2 \|\z\|^2 - \frac1{50} \|\u\|^4 \\
=~& -\frac1{100} |\omega_\T|^2 \|\x_\T\|^2 \|\z\|^2 - \frac{1}{50}|\omega_\T|^4\|\x_\T\|^4 .
\end{align}
Thus, the inequality in $\mathcal R_1$ is indeed
\begin{align}
\nonumber	&8|\x_\T^* \z|^2 - 2\|\x_\T\|^4 + 4\|\z\|^2 \|\x_\T\|^2 - 2\|\x\|^2 \|\x_\T\|^2 \\
\le&  -\frac1{100} \|\x_\T\|^2 \|\z\|^2 - \frac{1}{50}|\omega_\T|^2\|\x_\T\|^4,
\end{align} 
which can be simplified as 
\begin{equation}
8|\x_\T^* \z|^2 + \frac{401}{100}\|\x_\T\|^2 \|\z\|^2 \le \frac{199}{100}\|\x_\T\|^2(\|\x_\T\|^2 + \|\x\|^2) . 
\end{equation}

\item For $\mathcal R_2^\z$, we have 
\begin{align}
\nonumber	\langle \z, \nabla_1 \E[f(\z)]\rangle & = \z^* \left(2\|\z\|^2 \boldsymbol{I}-\|\x\|^2 \boldsymbol{I}-\x \x^*\right) \z\\
\nonumber	& = 2\|\z\|^4 - \|\x\|^2 \|\z\|^2 - |\x^* \z|^2\\
& = 2\|\z\|^4 - \|\x\|^2 \|\z\|^2 - |\x_\T^* \z|^2.
\end{align} 

Thus, the inequality in $\mathcal R_2^\z$ is indeed
\begin{equation}
\frac{501}{500}\|\x\|^2 \|\z\|^2 + |\x_\T^* \z|^2 \le \frac{199}{100}\|\z\|^4.
\end{equation}  
\item For $\mathcal R_2^{\h'}$, we have 
\begin{align}
\nonumber	&\langle \h(\z), \nabla_1 \E[f(\z)]\rangle\\
\nonumber	=~ & (\z - \u \mathrm{e}^{\mathsf{j}\phi})^* \left(2\|\z\|^2 \boldsymbol{I}-\|{\x}\|^2 \boldsymbol{I}-{\x} {\x}^*\right) \z\\
\nonumber	=~ &\left(2\|\z\|^2 -\|{\x}\|^2\right)\\
\nonumber	&\cdot\left(\|\z\|^2 - |\u^* \z|\right) - |\x^* \z|^2 + \|\x_\T\|^2 |\u^* \z|\\
\nonumber	=~ &\left(2\|\z\|^2 -\|{\x}\|^2\right)\\
~& \cdot\left(\|\z\|^2 - |\omega_\T| |\x_\T^* \z|\right) - |\x_\T^* \z|^2 + |\omega_\T|\|\x_\T\|^2 |\x_\T^* \z|.
\end{align} 
Denote $\xi = \frac{\|\z\|}{ \|\x\|}$ and $\eta = \frac{|\x_\T^* \z|} {\|\x_\T\|\|\z\|}$, then
\begin{align}\label{eq:p}
\nonumber		& \langle \h(\z), \nabla_1 \E[f(\z)]\rangle \\
\nonumber		=~ &\left(2\xi^2 \|\x\|^2 -\|\x\|^2\right) \left(\xi^2 \|\x\|^2 - |\omega_\T|\cdot \eta a \xi \|\x\|^2\right) \\
\nonumber		&~  - \eta^2 \xi^2 a^2 \|\x\|^4 + |\omega_\T| \xi a^3 \eta \|\x\|^4\\
\nonumber		=~ &\xi \|\x\|^4 \left[\left(2\xi^2 -1\right)\left(\xi - |\omega_\T|\eta a \right) - \eta^2 \xi a^2 + |\omega_\T| a^3 \eta\right]\\
\doteq~ 	& \xi \|\x\|^4 p(\xi, \eta),
\end{align}
where $$|\omega_\T| = \sqrt{\frac{\|\x_\T\|^2 + \|\x\|^2}{2 \|\x_\T\|^2}} = \sqrt{\frac{1 + a^2}{2a^2}}.$$ For the right hand side of the inequality in $\mathcal R_2^{\h'}$, 
\begin{align}\label{eq:q}
\nonumber	&\frac1{250}\|\x\|^2\|\z\|\|\h(\z)\| \\
\nonumber	=~ & \frac1{250} \|\x\|^2 \xi \|\x\| \sqrt{\|\z\|^2 + \|\u\|^2 - 2|\u^* \z|}\\
\nonumber	=~ &\frac1{250} \xi \|\x\|^4 \sqrt{\xi^2 + |\omega_\T|^2 a^2 - 2|\omega_\T| \eta \xi a}\\
\doteq~ &  \frac1{250} \xi \|\x\|^4 q(\xi, \eta).
\end{align}
Hence, the inequality in $\mathcal R_2^{\h'}$ is equivalent to 
\begin{equation}\label{area:R2hleq}
p(\xi, \eta) \ge  \frac1{250} q(\xi, \eta). 
\end{equation}
\end{itemize}
Now we proceed to prove that 
\begin{equation}
\C^\T = \R_1\cup\R_2^\z\cup\R_2^\h\cup\R_3.
\end{equation}
\begin{proof}
First, it is trivial to verify that $\C^\T = \mathcal{R}_a\cup \mathcal{R}_b\cup \mathcal{R}_c\cup \mathcal{R}_d$. Then we will show that $\R_a$, $\R_b$, $\R_c$, and $\R_d$ are subsets of $\R_1\cup\R_2^\z\cup\R_2^{\h'}\cup\R_3$.
\begin{enumerate}[i)]
\item $\mathcal{R}_a\subseteq \mathcal{R}_1 \cup \mathcal{R}_2^\z$.
\begin{proof}
	$\forall \z \in \mathcal{R}_a$, when $\|\z\|^2<\frac{199}{601}(a^2+1)$, we have
	\begin{align}
		\nonumber	8|\x_\T^* \z|^2+\frac{401}{100}\|\z\|^2\|\x_\T\|^2&\le \frac{601}{100}\|\z\|^2\|\x_\T\|^2\\
		&< \frac{199}{100}(a^2+1)\|\x\|^2\|\x_\T\|^2,
	\end{align}
	which means $$\mathcal{R}_a\cap \left\{\z\in\C^\T, \|\z\|^2<\frac{199}{601}(a^2+1)\|\x\|^2\right\}\subseteq \mathcal{R}_1.$$
	
	For all $\z \in \mathcal{R}_a$, when $$\|\z\|^2>\frac{125a^2+501}{995}\|\x\|^2,$$ we have
	\begin{align}
		\nonumber	\frac{501}{500}\|\x\|^2 \|\z\|^2 + |\x_\T^* \z|^2 &\le \frac{501}{500}\|\x\|^2 \|\z\|^2 + \frac14 \|\x_\T\|^2 \|\z\|^2\\
		&<\frac{199}{100}\|\z\|^4,
	\end{align}
	which means $$\mathcal{R}_a\cap \left\{\z\in\C^\T,\|\z\|^2>\frac{125a^2+501}{995}\|\x\|^2 \right\}\subseteq \mathcal{R}_2^\z.$$ 
	
	~~Since
	\begin{equation}\label{c1}
		\frac{125a^2+501}{995} < \frac{199}{601}(a^2 + 1),\quad \forall a^2\in \left [\frac{9}{10}, 1 \right ],
	\end{equation}
	we can conclude that $\mathcal{R}_a\subseteq \mathcal{R}_1 \cup \mathcal{R}_2^\z$.
\end{proof}
\item $\mathcal{R}_b\subseteq \mathcal{R}_1$.
\begin{proof}
	$\forall \z \in \mathcal{R}_b$, we have 
	\begin{align}
		\nonumber	8|\x_\T^* \z|^2+\frac{401}{100}\|\z\|^2\|\x_\T\|^2 &\le \frac{1201}{100}\|\z\|^2\|\x_\T\|^2\\
		&\le \frac{58849}{15625} \|\x\|^2 \|\x_\T\|^2.
	\end{align}
	
	Since
	\begin{equation}
		\frac{58849}{15625} < \frac{199}{100} (1 + a^2),\quad \forall a^2\in \left [\frac{9}{10}, 1 \right ],
	\end{equation} 
	we have $\mathcal{R}_b\subseteq \mathcal{R}_1$.
\end{proof}

\item $\mathcal{R}_c\subseteq \mathcal{R}_2^\z\cup \mathcal{R}_2^{\h '}$.
\begin{proof}
	$\forall \z \in \mathcal{R}_c$, when $$\|\z\|^2>\frac{10020 + 9801a^2}{19900} \|\x\|^2,$$ we have
	\begin{align}
		\nonumber&	\frac{501}{500}\|\x\|^2 \|\z\|^2 +|\x_\T^* \z|^2\nonumber	\\
		\nonumber\overset{(b)}{\leq}~& \left(\frac{501}{500}+ \left(\frac{99}{100}\right)^2 a^2\right)\|\x\|^2 \|\z\|^2\\
		<~&\frac{199}{100}\|\z\|^4,
	\end{align}
	where (b) is due to $|\x_\T^* \z|<\frac{99}{100}\|\x\|$ for $\z\in \R_c$.
	Thus, we conclude that $$\mathcal{R}_c\cap\left\{\z\in\C^\T,\|\z\|^2>\frac{10020 + 9801a^2}{19900} \|\x\|^2\right\}\subseteq \mathcal{R}_2^\z.$$ 
	
	~~Next, we will show the rest of $\R_c$ is covered by $\R_2^{\h'}$.
	First, for $$\z\in\C^\T,\|\z\|\leq \sqrt{\frac{10020 + 9801a^2}{19900} }\|\x\|.$$ It is easy to see that 
 	\begin{equation}
  \sqrt{\frac{10020 + 9801a^2}{19900}}\le \sqrt{\frac{1+a^2}{2}}
  	\end{equation}
   	holds for all $a^2\in \left [\frac{9}{10}, 1 \right ]$, and $$\sqrt{\frac{1+a^2}{2}}\|\x\|=\|\u\|.$$ Therefore, 
	\begin{align}\label{b12}
		\|\z\|\le \|\u\|.
	\end{align}
	Then, we consider $p(\xi, \eta)$ and $q(\xi, \eta)$ with $$\xi\in \left[\frac{11}{20}, \sqrt{\frac{10020 + 9801a^2}{19900}}\right],$$ $\eta \in \left[\frac12, \frac{99}{100}\right]$, where $p(\xi, \eta)$ and $q(\xi, \eta)$ are defined in~\eqref{eq:p} and~\eqref{eq:q}. We will show~\eqref{area:R2hleq}, i.e., verify that
	\begin{equation}\label{b111}
		\min_{\xi, \eta} \left(p(\xi, \eta) - \frac1{250} q(\xi, \eta)\right)>0.
	\end{equation}
	Note that 
	\begin{flalign}
		&q(\xi, \eta) \nonumber&\\
		\le& \max\left\{q\left(\frac{11}{20},\frac12\right), q\left(\sqrt{\frac{10020 + 9801a^2}{19900}}, \frac12\right)\right\},
	\end{flalign}
	since it is obvious that $q^2(\xi, \eta)$ is a parabola of $\xi$ opens upward and $q(\xi, \eta)$ is a decreasing function of $\eta$. Then, one can verify that $q\left(\frac{11}{20},\frac12\right)$ and $q\left(\sqrt{\frac{10020 + 9801a^2}{19900}}, \frac12\right)$ are both increasing functions of $a$. Take $a = 1$ and we will find them both smaller than $1$. It means
	\begin{equation}
		\frac1{250} q(\xi,\eta) < \frac{1}{250},\quad\forall a^2 \in \left [\frac{9}{10}, 1 \right ]. 
	\end{equation}

	~~On the other hand, it is trivial that $p(\xi_0, \eta)$ is a parabola of $\eta$ opens downward. Thus, for any fixed $\xi_0$, the function $p(\xi_0, \eta)$ are minimized at either $\eta = \frac12$ or $\eta = \frac{99}{100}$. Thus,
	\begin{equation}
		\min_{\xi, \eta} p(\xi, \eta) = \min\left\{\min_{\xi}p\left(\xi, \frac12\right), \min_{\xi}p\left(\xi, \frac{99}{100}\right)\right\}. 
	\end{equation}
	We can see that 
	\begin{equation}
		\frac{\partial p}{\partial \xi} (\xi, \eta)= 6\xi^2 - 4a|\omega_\T| \eta \xi - 1- a^2\eta^2 ,
	\end{equation}
	and its right zero point is
	\begin{equation}
		\xi^*(\eta) = \frac{4a|\omega_\T| \eta + \sqrt{16a^2 |\omega_\T|^2 \eta^2 + 24(1 + a^2 \eta^2)}}{12} .
	\end{equation}

	~~Hence, for any fixed $\eta$, $p(\xi, \eta)$ decreases in $[0, \xi^*(\eta)]$ and increases on $[\xi^*(\eta), \infty]$. Further, we get 
	\begin{flalign}
		&\min_{\xi, \eta} p(\xi, \eta) &\nonumber\\
		\geq~& \min\left\{p\left(\xi^*\left(\frac12\right), \frac12\right), p\left(\xi^*\left(\frac{99}{100}\right), \frac{99}{100}\right)\right\}.
	\end{flalign}
	One can verify that the right hand side is increasing with $a$. Thus we take $a^2 = \frac{9}{10}$ and get 
	\begin{equation}
		\min_{\xi, \eta} p(\xi, \eta)\ge 0.017,\quad\forall a^2 \in \left [\frac{9}{10}, 1 \right ].  
	\end{equation}

	~~Hence, in this region we have 
	\begin{equation}
		p(\xi, \eta) >  \frac1{250} q(\xi, \eta), \forall a\in \left [\frac{9}{10}, 1 \right ]. 
	\end{equation}
\end{proof}
\item $\mathcal R_d \subseteq \mathcal R_2^\z \cup \mathcal R_2^{\h'}\cup \mathcal R_3$.
\begin{proof}
	First, for any $\z\in \mathcal R_d$, when $\|\z\|\ge \sqrt{\frac{501}{995}(1 + a^2)} \|\x\|$, 
	\begin{align}
		\nonumber	\frac{501}{500}\|\x\|^2\|\z\|^2 + |\x_\T^* \z|^2 &\le \frac{501}{500}(1 + a^2)\|\x\|^2 \|\z\|^2\\
		&\le \frac{199}{100}\|\z\|^4 .
	\end{align}
	So, $$\mathcal R_d \cap \left\{\z\in\C^\T:\|\z\|\ge \sqrt{\frac{501}{995}(1 + a^2)}\|\x\|\right\}\subseteq\R_2^\z.$$ 
	
	~~Next, we show that any $\z\in \mathcal R_d$ with $\|\z\|\le \frac{47}{50} \|\x\|$ is contained in $\mathcal R_2^{\h '}$. Similar to the above argument for $\mathcal R_c$, it is enough to show that 
	\begin{equation}
		\min_{\xi, \eta} \left(p(\xi, \eta) - \frac1{250} q(\xi, \eta)\right)>0,\quad\forall a^2 \in \left [\frac{9}{10}, 1 \right ], 
	\end{equation} 
	with $\xi \in \left[\frac{11}{20}, \frac{47}{50}\right]$ and $\eta \in \left[\frac{99}{100}, 1\right]$. Notice that in this region, 
	\begin{equation}
		q(\xi, \eta) \le \max\left\{q\left(\frac{11}{20}, \frac{99}{100}\right), q\left(\frac{47}{50}, \frac{99}{100}\right)\right\}. 
	\end{equation}
	It is easy to verify that the right hand side is increasing with $a$. Take $a = 1$ and we can obtain that
	\begin{equation}
		\frac{1}{250} q(\xi,\eta) < 0.0018,\quad\forall a^2 \in \left [\frac{9}{10}, 1 \right ].  
	\end{equation}

	~~On the other hand, recall that $p(\xi_0, \eta)$ is a parabola of $\eta$ opens downward. Then the minimizer must occur on the boundary. For fixed $\xi_0$, the function $p(\xi_0, \eta)$ are minimized at either $\eta = \frac{99}{100}$ or $\eta = 1$. Thus, 
	\begin{equation}
		\min_{\xi, \eta} p\left(\xi, \eta\right) = \min\left\{\min_{\xi}p\left(\xi, 1\right), \min_{\xi}p\left(\xi, \frac{99}{100}\right)\right\}. 
	\end{equation} 
	However, notice that $\xi^* (\eta) \ge \xi^*(\frac{99}{100})$ since $\eta \in \left[\frac{99}{100}, 1\right]$ and $\xi^* (\eta)$ increases as $\eta$ increases. Moreover, it increases as $a$ increases. Therefore, taking $a^2 = \frac{9}{10}$, we get 
	\begin{equation}
		\xi^* (\eta) \ge  \xi^*\left.\left(\frac{99}{100}\right)\right|_{a^2 = \frac{9}{10}} \geq 0.944437 > \frac{47}{50}. 
	\end{equation} 
	Note that $\xi \in \left[\frac{11}{20}, \frac{47}{50}\right]$. It means that for all $a^2 \in [\frac{9}{10}, 1]$, it holds that
	\begin{align}
		\nonumber\min_{\xi, \eta} p(\xi, \eta)& = \min\left\{p\left(\frac{47}{50}, 1\right), p\left(\frac{47}{50}, \frac{99}{100}\right)\right\}\\
		&\ge 0.0046.
	\end{align}
	Hence, in this region we have 
	\begin{equation}
		p(\xi, \eta) >  \frac1{250} q(\xi, \eta),\quad \forall a\in \left [\frac{9}{10}, 1 \right ]. 
	\end{equation}
	Therefore, \begin{equation}
		\mathcal R_d \cap \left\{\z\in\C^\T:\|\z\|\le \frac{47}{50} \|\x\|\right\}\subseteq \mathcal R_2^{\h'}.
	\end{equation}
	
	~~Finally, we consider the case where \begin{equation}
		\frac{47}{50} \|\x\|\le \|\z\|\le \sqrt{\frac{501}{995}(1 + a^2)} \|\x\|.
	\end{equation} A similar $\xi, \eta$ argument as above leads to 
	\begin{equation}
		\|\h(\z)\|^2 = q^2(\xi, \eta)\|\x\|^2 \leq \frac17 \|\x\|^2,\quad \forall a^2 \in \left [\frac{9}{10}, 1 \right ], 
	\end{equation} 
	implying that \\$\mathcal R_d \cap \left\{\z: \frac{47}{50} \|\x\|\le \|\z\|\le \sqrt{\frac{501}{995}(1 + a^2)} \|\x\|\right\} \subseteq \mathcal R_3$.
\end{proof}
\end{enumerate}

In summary, now we obtain that $\C^\T= \mathcal{R}_a\cup \mathcal{R}_b\cup \mathcal{R}_c\cup \mathcal{R}_d \subseteq \mathcal{R}_1\cup \mathcal{R}_2^\z\cup \mathcal{R}_2^{\h'}\cup \mathcal{R}_3 \subseteq \C^\T$, which gives 
\begin{equation}
\mathcal{R}_1\cup \mathcal{R}_2^\z\cup \mathcal{R}_2^{\h'}\cup \mathcal{R}_3 = \z\in\C^\T. 
\end{equation}
Observe that $\mathcal R_2^{\h'}$ is only used to cover $\mathcal R_c\cup \mathcal R_d$, which is in turn a subset of $\{\z: \|\z\|\ge \frac{11}{20} \|\x\|\}$. Thus, 
\begin{align}
\nonumber	\C^\T
&= \mathcal{R}_1\cup \mathcal{R}_2^\z\cup \mathcal{R}_2^{\h'}\cup \mathcal{R}_3\\
\nonumber	& = \mathcal{R}_1\cup \mathcal{R}_2^\z\cup \left(\mathcal{R}_2^{\h'}\cap \Big \{\z: \|\z\|\ge \frac{11}{20} \|\x\| \Big \} \cap \mathcal R_3^c\right)\cup \mathcal{R}_3\\
\nonumber	& \subseteq \mathcal{R}_1\cup \mathcal{R}_2^\z\cup \mathcal{R}_2^{\h}\cup \mathcal{R}_3.\\
& \subseteq \C^\T.
\end{align} 
Therefore, \begin{equation}
\C^\T = \R_1\cup\R_2^\z\cup\R_2^\h\cup\R_3.
\end{equation}
Our proof is thus complete.
\end{proof}

\end{appendix}

\end{document}